\documentclass{theoretics}

\title{Randomized Communication and Implicit Graph Representations}

\ThCSauthor[UBC]{Nathaniel Harms}{nharms@cs.ubc.ca}[0000-0003-0259-9355]
\ThCSauthor[Liv,Mar]{Sebastian Wild}{wild@informatik.uni-marburg.de}[0000-0002-6061-9177]
\ThCSauthor[Liv]{Viktor Zamaraev}{viktor.zamaraev@liverpool.ac.uk}[0000-0001-5755-4141]
\ThCSshortnames{N.~Harms, S.~Wild, V.~Zamaraev}
\ThCSshorttitle{Randomized Communication and Implicit Graph Representations}
\ThCSaffil[UBC]{University of British Columbia, Canada}  
\ThCSaffil[Liv]{University of Liverpool, UK}  
\ThCSaffil[Mar]{University of Marburg, Germany}
\ThCSthanks{This work was done while Nathaniel Harms was a student at the University
of Waterloo, partly supported by an NSERC postgraduate scholarship. Sebastian Wild is supported by EPSRC grant EP/X039447/1. An earlier version of
this paper appeared at STOC 2022 \cite{HWZ22}.}

\ThCSkeywords{Randomized communication complexity; constant-cost protocols; adjacency labeling schemes; probabilistic universal graphs; adjacency sketches; implicit graph representations; hereditary graph classes}

\ThCSyear{2025}
\ThCSarticlenum{20}
\ThCSreceived{Jul 19, 2023}
\ThCSrevised{Jun 30, 2025}
\ThCSaccepted{Jul 18, 2025}
\ThCSpublished{Oct 01, 2025}
\ThCSdoicreatedtrue

\definecolor{BrickReplacement}{rgb}{0.8, 0.25, 0.33}

\usepackage{mathtools}
\usepackage{xspace}
\usepackage{verbatim}
\usepackage[capitalise,noabbrev,nameinlink]{cleveref}
\usepackage{pgf}
\usepackage{tikz}

\usepackage{dsfont}

\usetikzlibrary {calc,shapes.geometric,fit,arrows.meta,patterns.meta,decorations.pathreplacing}
\tikzset{
	complete bipartite/.style={very thick,double},
	any bipartite/.style={black!50,very thick,densely dashed},
	edges/.style={font=\scriptsize,fill=white,inner sep=0pt},
	vertex set/.style={shape=ellipse,draw,minimum height=8ex,minimum width=2.5em,inner sep=1pt},
	vertex sets/.style={vertex set,inner sep=0pt},
}
\tikzset{
	named vertex/.style={circle,draw,inner sep=1pt,minimum size=16pt,font=\scriptsize},
	vertex/.style={circle,draw,fill=black,inner sep=0pt,minimum size=4pt},
	every label/.style={font=\small,label distance=-1pt},
}

\usepackage{xifthen}

\ThCSnewtheoita{question}

\newcommand{\ignore}[1]{}

\usepackage{mleftright}\mleftright

\DeclareMathOperator{\poly}{poly}

\newcommand{\dist}{\mathsf{dist}}

\newcommand{\Ex}[1]{\bE \left[ #1 \right]}

\renewcommand{\Pr}[1]{\bP \left[ #1 \right]} % Probability
\newcommand{\Pru}[2]{\underset{ #1 }\bP \left[ #2 \right]}

\newcommand{\define}{\vcentcolon=}

\newcommand{\ceil}[1]{\ensuremath{\lceil #1 \rceil}}

\newcommand{\ind}[1]{\mathds{1} \left[ #1 \right] }

\newcommand\eqLabelSep{\mathbin|}

\newcommand{\zo}{\{0,1\}}

\newcommand{\note}[1]{{\color{Brown} \;\ifmmode \text{#1} \else #1 \fi\;}}

\newcommand{\cA}{\ensuremath{\mathcal{A}}}
\newcommand{\cB}{\ensuremath{\mathcal{B}}}
\newcommand{\cC}{\ensuremath{\mathcal{C}}}
\newcommand{\cD}{\ensuremath{\mathcal{D}}}
\newcommand{\cE}{\ensuremath{\mathcal{E}}}
\newcommand{\cF}{\ensuremath{\mathcal{F}}}
\newcommand{\cG}{\ensuremath{\mathcal{G}}}
\newcommand{\cH}{\ensuremath{\mathcal{H}}}
\newcommand{\cI}{\ensuremath{\mathcal{I}}}

\newcommand{\cL}{\ensuremath{\mathcal{L}}}

\newcommand{\cM}{\ensuremath{\mathcal{M}}}
\newcommand{\cP}{\ensuremath{\mathcal{P}}}
\newcommand{\cQ}{\ensuremath{\mathcal{Q}}}

\newcommand{\cS}{\ensuremath{\mathcal{S}}}
\newcommand{\cT}{\ensuremath{\mathcal{T}}}
\newcommand{\cU}{\ensuremath{\mathcal{U}}}

\newcommand{\bE}{\ensuremath{\mathbb{E}}}
\newcommand{\bF}{\ensuremath{\mathbb{F}}}
\newcommand{\bN}{\ensuremath{\mathbb{N}}}
\newcommand{\bP}{\ensuremath{\mathbb{P}}}
\newcommand{\bR}{\ensuremath{\mathbb{R}}}

\newcommand{\bZ}{\ensuremath{\mathbb{Z}}}

\newcommand\splitaftercomma[1]{%
  \begingroup
  \begingroup\lccode`~=`, \lowercase{\endgroup
    \edef~{\mathchar\the\mathcode`, \penalty0 \noexpand\hspace{0pt plus .25em}}%
  }\mathcode`,="8000 #1%
  \endgroup
}

\newlength{\origtopset}\newlength{\origitemsep}
\newcommand{\Eq}{\textsc{Eq}}
\newcommand{\Adj}{\textsc{Adj}}
\newcommand{\bip}{\mathsf{bip}}
\newcommand{\eqc}{\mathsf{eqc}}
\newcommand{\beqc}{\mathsf{beqc}}

\newcommand{\tw}{\mathsf{tww}}

\usetikzlibrary{arrows,shapes,snakes}
\usetikzlibrary{arrows.meta}
\usetikzlibrary{decorations.pathreplacing}

\definecolor{NotionColor}{rgb}{0.1,0.1,1}
\definecolor{VioletThemeColor}{rgb}{0.5,0,0.5}

\usepackage[outline]{contour}
\newcommand\nodetextbeforebg[3][.08em]{%
	{%
	\contourlength{#1}%
	\contour{#2}{#3}%
	}%
}

\pgfdeclarelayer{background}
\pgfdeclarelayer{foreground}
\pgfsetlayers{background,main,foreground}
\newcommand{\bb}[1]{{\color{NotionColor} #1}}
\newcommand{\rr}[1]{{\color{BrickReplacement} #1}}

\crefname{claim}{Claim}{Claims}

\creflabelformat{equation}{#2\textup{#1}#3}
\usepackage{xspace}
\newcommand*\ie{i.\kern.1em e.\@\xspace} % inserts space if used in sentence, \ie like this, 
\newcommand*\eg{e.\kern.1em g.\@\xspace} % but also allows the American contention, \eg, this with comma.

\newcommand\doverline[1]{% double overline from https://tex.stackexchange.com/a/488246
  \setbox0=\hbox{$\!\overline{#1}\!$}%
  \ht0=\dimexpr\ht0-.2ex\relax% CHANGE .15 TO AFFECT SPACING
  \,\overline{\copy0}\,%
}

\newcommand{\decomp}{chain\xspace}
\newcommand{\chainNum}{chain number\xspace}
\newcommand*\chainyGraphs{chain-like graphs\xspace}

\newcommand*\EQUALITY{\textsc{Equality}\xspace}

\newcommand{\bc}[1]{\doverline{#1}} % bipartite complement

\newcommand{\sk}{\mathsf{sk}}
\newcommand{\RL}{\mathsf{SK}}
\newcommand{\cl}{\mathsf{cl}}
\DeclareMathOperator\ch{\mathsf{ch}}

\newcommand{\CC}{\mathsf{R}} 
\newcommand{\dcirc}{{\circ\circ}}
\newcommand{\bcirc}{{\bullet\circ}}
\newcommand{\dbullet}{{\bullet\bullet}}
\newcommand{\free}{\mathsf{Free}}

\colorlet{proven}{green!50!black}
\colorlet{conjectures}{orange!80!black}
\colorlet{constPUG}{green!50!black}
\colorlet{constPUGconjecture}{constPUG!50!white!80!gray!60}
\colorlet{igcPositive}{cyan}
\colorlet{igcNegative}{red!80!black!80}
\colorlet{igcPositiveConjecture}{cyan!50!white!50!gray!60}
\colorlet{polyPUG}{blue!80!black}
\colorlet{bell}{red!90!black}
\colorlet{vsep}{black}
\tikzset{
	mybrace/.style={thick,decoration={brace,amplitude=5pt}},
	comm-problems/.style={fill=black!5},
	const-det/.style={fill=black!30},
	const-rand/.style={fill=constPUG!50},
	igc-positive/.style={fill=igcPositive!50},
	mapping/.style={thick,densely dashed,-Stealth[],shorten >=1pt,shorten <=1pt},
	igc-positive-conjecture/.style={fill=igcPositiveConjecture},
	const-pug-conjecture/.style={fill=constPUGconjecture},
	igc-negative/.style={fill=igcNegative!50},
	pug-negative/.style={fill=igcNegative},
}
\newcommand{\resultsfigure}{%
\begin{figure}[tbp]
	\begin{tikzpicture}[
			xscale=-.65,
			yscale=.75,
			every node/.style={font=\footnotesize},
			mybrace/.style={thick,decoration={brace,amplitude=5pt,mirror}},
	]
		\draw[{ultra thick,black!50,-Stealth[]}] 
			(14,-1) -- node[above,sloped,pos=.5,black] {speed} ++(0,12) ;
		
		\foreach \n/\x/\y in {ll/-4/-1,lr/4/-1,ul/-7/10,ur/7/10} 
			{ \coordinate (\n) at (\x,\y) ; }
		\foreach \s in {l,r} { \coordinate (bell-\s) at ($ (l\s)!.45!(u\s) $) ; }
		\foreach \s in {l,r} { \coordinate (abovebell-\s) at ($ (l\s)!.53!(u\s) $) ; }
		\def\height{11}
		\begin{pgfonlayer}{background}
		\fill[black!5] (ur) -- (lr) -- (ll) -- (ul) ;
		\end{pgfonlayer}
		\draw[thick] (ur) -- (lr) -- (ll) -- (ul) ;
		
		\foreach \c/\y in {constant/0,polynomial/1,exponential/2,/8.5} {
			\node at (0,\y-.5) {\c};
			\foreach \s in {l,r} { \coordinate (\c-\s) at ($ (l\s)!{(\y+1)/\height}!(u\s) $) ; }
			\draw[ultra thick] (\c-l) -- (\c-r) ;
		}
		\begin{pgfonlayer}{background}
		\fill[const-det] (exponential-r) -- (lr) -- (ll) -- (exponential-l) ;
		\end{pgfonlayer}

		\node at (0,9.25) {superfactorial};
		\def\cpug{0.33}
		
		\coordinate (cPUG-b) at ($ (exponential-l)!\cpug!(exponential-r) $) ;
		\coordinate (cPUG-u) at ($ (-l)!\cpug!(-r) $) ;
		\coordinate (bell-m) at ($ (bell-l)!\cpug!(bell-r) $);
		\coordinate (abovebell-m) at ($ (abovebell-l)!\cpug!(abovebell-r) $);
		\coordinate (um) at ($ (ul)!\cpug!(ur) $);

		\begin{pgfonlayer}{background}
		\draw[constPUG,fill=constPUG!50] 
			(exponential-r) -- (cPUG-b) -- (abovebell-m) -- (abovebell-r) -- cycle ;
		\fill[fill=igcPositive!50] 
			(bell-l) -- (bell-m) -- (um) -- 
				node[above,polyPUG,font=\scriptsize,align=center] {not stable} 
			(ul) -- cycle ;
		\path (um) -- node[above,constPUG,font=\scriptsize,align=center] {stable}  (ur);
		\end{pgfonlayer}

		\node[scale=2] at ($ (cPUG-b)!.5! (bell-l) $) {$\emptyset$} ;
		
		\draw[mybrace,decorate] (ll -| -11.5,0) -- 
				node[right=5pt,align=left,font=\scriptsize] {$O(1)$ determ.\\labeling scheme} 
				(exponential-l -| -11.5,0) ;
		\draw[mybrace,decorate] (exponential-l -| -12.5,0) -- 
				node[right=5pt,align=left,font=\scriptsize] {$\Omega(\log n)$\\determ.\\labeling\\scheme} 
				(ul -| -12.5,0) ;
			\draw[thin,dotted] (ul) -- ++(-8.5,0);
				\draw[thin,dotted] (exponential-l) -- ++(-11,0);
				\draw[thin,dotted] (ll) -- ++(-12,0);
	\foreach [count=\x] \c in {%
				{\cL^{\bullet\bullet}},{\cL^{\bullet\circ}},{\cL^{\circ\circ}},
				{\cM^{\bullet\bullet}},{\cM^{\bullet\circ}},{\cM^{\circ\circ}}%
		} {
			\node[scale=.5,circle,draw,thick,inner sep=0pt,minimum size=22pt,purple,fill=purple!5] 
				at ($ (cPUG-b)!{(\x)/7}!(exponential-r) + (0,.5) $) {$\c$} ;
		}
		\foreach [count=\x] \c in {%
				{\cC^{\bullet\bullet}},{\cC^{\bullet\circ}},{\cC^{\circ\circ}}%
		} {
			\node[scale=.5,circle,thick,draw,inner sep=0pt,minimum size=22pt,violet,fill=purple!5] 
				at ($ (bell-l)!{(\x-.3)/3.5}!(bell-m) + (0,.5) $) {$\c$} ;
		}
		
		\foreach \x in {1,...,7} {
			\node[rectangle,draw,thick,inner sep=0pt,minimum size=7pt,violet,fill=purple!5] at 
				($ (bell-m)!{(\x+.5)/10}!(bell-r) + (0,.5) $) {};
		}

		\draw[mybrace,decorate,proven!50!black] (abovebell-r -| 6,0) -- 
				node[left=5pt,align=right] (constPUG-label) {$O(1)$ adjacency\\ sketch} 
				(exponential-r -| 6,0) ;
		\node[proven!50!black,left=5pt,align=right] at (7,5.75)
			(constPUG-label-2) {$O(1)$ adjacency\\ sketch} ;
		\foreach \c/\x/\y in {%
					\cI/3/5.25,%
					\cP/4.5/5.25,%
					{S_{123}}/-1/8,%
					{F_{pq}^*}/0.5/8,%
					{P_7}/2/8,%
					{\alpha}/0/7,%
					{G^\star}/5.5/7.5%
		} {
			\draw[very thin,constPUG,Stealth-] (constPUG-label-2.0) .. controls ++(180:1) .. (\x,\y);
			\begin{pgfonlayer}{background}
			\begin{scope}
				\clip (abovebell-m) -- (abovebell-r) -- (-r) -- (cPUG-u) -- cycle ;
				\filldraw[very thin,constPUG,fill=constPUG!50] 
					(\x-0.2,\y) -- ++(-.75,-10) -- ++(1.5,0) -- (\x+0.2,\y) -- cycle;
			\end{scope}
			\end{pgfonlayer}
			\node[scale=.7,ellipse,draw,thick,inner sep=0pt,
					minimum width=18pt,minimum height=13pt,proven,fill=proven!10
				] 
					at (\x,\y) {} ;
		}

		\draw[mybrace,decorate,polyPUG] (bell-l -| -7.5,0) -- 
				node[right=5pt,align=left,polyPUG] {$\Omega(\log n)$\\ adjacency\\ sketch} 
				(ur -| -7.5,0) ;

		\draw[bell,densely dashed,ultra thick] (bell-r) -- (bell-l);
		\draw[bell,very thin] (bell-l) -- ++(220:1) node[right] {\scriptsize Bell numbers};
		\draw[vsep,ultra thick] (cPUG-b) -- (cPUG-u);
		\draw[vsep,ultra thick,dashed] (um) -- (cPUG-u);

	\end{tikzpicture}
	\caption{%
    Overview of our results (in green) in relation to the lattice of hereditary graph classes.
    Circles are minimal factorial classes; purple shapes are minimal classes
    above the Bell numbers.
	}
\label{fig:graph theory}
\end{figure}}
\makeatletter
\newcommand\IfRestateTF{%
  \ifx\label\thmt@gobble@label % or just compared to \@gobble
    \expandafter\@firstoftwo
  \else
    \expandafter\@secondoftwo
  \fi
}
\makeatother
\newcommand{\RestateRemark}{\IfRestateTF{{\normalfont\bfseries (Restated) }}{}}

\begin{document}

\maketitle

\begin{abstract}
We initiate the focused study of constant-cost randomized communication, with emphasis on its
connection to graph representations. We observe that constant-cost randomized communication
problems are equivalent to hereditary (\ie closed under taking induced subgraphs) graph classes
which admit constant-size \emph{adjacency sketches} and \emph{probabilistic universal graphs}
(PUGs), which are randomized versions of the well-studied adjacency labeling schemes and
induced-universal graphs.  This gives a new perspective on long-standing questions about the
existence of these objects, including new methods of constructing adjacency labeling schemes.

We ask three main questions about constant-cost communication, or equivalently, constant-size
PUGs: (1) Are there any natural, non-trivial problems aside from Equality
and $k$-Hamming Distance which have constant-cost communication? We provide a number of new
examples, including deciding whether two vertices have path-distance at most $k$ in a planar graph,
and showing that constant-size PUGs are preserved by the Cartesian product operation.  (2) What
structures of a problem explain the existence or non-existence of a constant-cost protocol?  
We show that in many cases a Greater-Than subproblem is such a structure.
(3) Is the Equality problem \emph{complete} for constant-cost randomized communication? We
show that it is not: there are constant-cost problems which do not reduce to Equality.
\end{abstract}

% \thispagestyle{empty}
% \setcounter{page}{0}
% \newpage
% {\small
% \setcounter{tocdepth}{2} 
% \tableofcontents
% }
% \thispagestyle{empty}
% \setcounter{page}{0}
% \newpage

\section{Introduction}

Randomized communication is a central topic in communication complexity
\cite{GPW18,CLV19,PSW20,HHH23dimfree}, and yet the power of randomness in communication remains
poorly understood. The textbook example of the power of randomness in communication is the
\textsc{Equality} problem (where two players wish to decide if they have received identical inputs),
which has a \emph{constant-cost} randomized protocol (see \eg~\cite{NK96}): the required number of
bits of communication remains constant, regardless of input size, whereas a deterministic
communication protocol cannot do better than simply sending the entire input.

To better understand the power of randomness in communication, we wish to understand the most
extreme examples. Therefore the goal of this paper is to initiate the focused study of the
communication problems that, like \textsc{Equality}, have constant-cost protocols. In related models
of computation like randomized decision trees or even \emph{deterministic} communication, the class
of constant-cost problems is trivial\footnote{Constant-cost randomized decision trees are equivalent
to constant-cost deterministic decision trees, which compute functions $\zo^n \to \zo$ depending on
only $O(1)$ variables \cite{BdW02}, and constant-cost deterministic communication protocols compute
$n \times n$ matrices which have only $O(1)$ distinct rows and columns.}. \emph{A priori}, one may
also expect the class of constant-cost randomized communication problems to be simple or nearly
trivial, but we will see that this is not the case: the class is surprisingly rich, containing a
variety of natural problems and having many connections to other areas.

One of the central observations of this paper is that constant-cost randomized communication
problems are equivalent to hereditary graph classes which admit a constant-size \emph{adjacency sketch}, a
probabilistic version of an \emph{adjacency labeling scheme} (also known as an \emph{implicit
representation}), the subject of a large body of research in graph theory and distributed computing
(\eg \cite{KNR92,FK09,KM12,ACLZ15,Har20,BGK+21a,DEG+21,HH22}). Therefore the study of constant-cost
randomized communication can be viewed from two perspectives, the communication complexity perspective and the
graph representation perspective. In \cref{section:intro-communication} we describe our motivations
and results from the perspective of communication complexity, and in \cref{section:intro-implicit}
from the perspective of graph representations. Each of these perspectives suggest different
questions, but, to begin understanding constant-cost randomized communication, we believe it is
important to answer the following 3 questions, which will be the focus of this paper:

\vspace{0.75em}\noindent\textbf{Question I.}
Are there natural examples of constant-cost randomized communication problems, that are not
trivially obtained from the few known examples (\textsc{Equality} and $k$-\textsc{Hamming
Distance})?  Equivalently, what hereditary classes of graphs admit constant-size adjacency sketches?
We will give a number of new examples, which also lead to new 
techniques for obtaining implicit representations of graphs.
%efficient graph representations.

\vspace{0.75em}\noindent\textbf{Question II.}
What structures of problems determine the existence or non-existence of constant-cost randomized
communication? One of the main applications of communication complexity is to prove lower bounds,
and constant vs.~non-constant is the most basic lower-bound question. One might therefore hope for a
simple criterion to determine whether a problem has constant cost or not.

\vspace{0.75em}\noindent\textbf{Question III.}
What ``types'' of constant-cost randomized protocols are there? In particular,
is there a ``complete'' problem for this class, which all constant-cost problems reduce to? The most
natural candidate is the \textsc{Equality} problem, and we show that \textsc{Equality} is \emph{not}
complete.\\

We refer the reader also to \cite{HHH23dimfree} which has independently and concurrently initiated
the study of constant-cost randomized communication, with an emphasis on the algebraic properties of
these problems, including connections to operator theory and other areas. Since the preprint and
conference version of the current paper \cite{HWZ22}, there have been many subsequent works
extending and improving our results. We discuss these works in \cref{section:discussion}.

\subsection{Communication Complexity}
\label{section:intro-communication}

Let us now discuss our results from the perspective of communication complexity. In
\cref{section:intro-implicit} we will discuss the results from the perspective of structural graph
theory and graph representations.

\newcommand{\R}{\mathsf{R}}
We formally define the relevant notions of communication complexity in \cref{def:communication
protocol}; here we write $\R(P)$ for the minimum cost (on worst-case inputs) of a \emph{public-coin}
randomized communication protocol computing the problem $P$, and note that $\R(P)$ is a function of
the input domain size.  We wish to understand the structure of problems which admit
\emph{constant-cost} randomized protocols, \ie, problems $P$ with $\R(P) = O(1)$. (Note that this
goal would not be sensible for \emph{private-coin} randomized protocols\footnote{With bounded error.
Constant-cost for private-coin \emph{unbounded-error} protocols (as in \cite{PS86}) is studied in a
number of subsequent works; see \cref{section:discussion}.}, since the constant-cost problems for this
class of protocols are the same as the constant-cost problems for deterministic communication.) We
organize the paper in order of studying Questions I-III.

\subsubsection*{Question I. What are natural examples of problems with constant cost?}

\textsc{Equality} is the
standard example of a constant-cost problem, and more generally it has been known since at least
\cite{Yao03} that for any constant $k \in \bN$, the \textsc{$k$-Hamming Distance} problem (where two players
decide if their inputs $x,y \in \zo^n$ differ on at most $k$ bits) also has constant cost. But, to
the best of our knowledge, almost no further non-trivial examples of constant-cost communication problems
have appeared in the past 20 years. This makes it difficult to speculate on the general structure of
constant-cost problems and the techniques required to design these ``hyperefficient'' protocols.

We present a number of new communication problems with non-trivial constant-cost communication
protocols. Let us break these down into two sub-questions:\\

\noindent
\emph{Question I (a).} For which graphs $G$ can two players decide if their vertices $x,y$ are
adjacent, or at distance at most $k$? Every communication problem can be rephrased as: Alice and Bob
each have a vertex in a shared graph $G$, and they want to decide if their vertices are
adjacent\footnote{The communication matrix may be interpreted as the adjacency matrix of a bipartite
graph.}. One may also think of the \textsc{$k$-Hamming Distance} problem as: Alice and Bob have
vertices of the hypercube graph, and wish to decide if their vertices have path-distance $k$. So,
which other classes of graphs have constant-cost communication protocols for deciding
adjacency or path-distance $k$?

Our first example is that Alice and Bob can decide if their vertices have path-distance $k$ in any
\emph{planar graph}, with communication cost depending on $k$ but \emph{not} on the number of
vertices in the graph:

\begin{theorem}[Informal; see \cref{section:twin-width}]
\label{thm:intro-planar}
Let $k$ be any constant, and let Alice and Bob have vertices $x,y$ in a shared planar graph $G$.
Then there is a constant-cost randomized communication protocol for deciding $\dist(x,y) \leq k$.
\end{theorem}

This is a corollary of a more general result, using structural graph theory techniques from the
literature on model checking. We generalize the problem in two ways:
\begin{enumerate}
\item The statement $\dist(x,y) \leq k$ can be written as a
first-order formula $\phi(x,y)$ using the edge relation of the graph $G$, and with variables being
vertices of the graph, \ie:
\[
  \dist(x,y) \leq k \equiv \phi(x,y) \define \left( \exists v_1, \dotsc, v_k : x=v_1 \wedge y = v_k
\wedge E(v_1, v_2) \wedge \dotsm \wedge E(v_{k-1}, v_k) \right) \,.
\]
We generalize the result to give communication protocols for deciding any
such first-order formula.
\item We generalize the result to hold not only for planar graphs, but for any \emph{stable} graph
class of \emph{bounded twin-width}. Including, for example, any proper minor-closed class or any
class with bounded treewidth or bounded clique-width.
\end{enumerate}

\noindent
\emph{Question I (b).}
How can we generate new constant-cost problems? In communication complexity, a common way to
generate new problems is by function composition, while in graph theory a common way to generate new
graphs is by graph product operations. We show that a certain composition operation, corresponding
to the Cartesian product operation for graphs, allows to construct new constant-cost problems 
which generalize the \textsc{$1$-Hamming Distance} problem.

Consider $n$ independent instances $P_1, \dotsc, P_n$ of any constant-cost ``base problem''
$\cP$, so that Alice and Bob have inputs $x_1, \dotsc, x_n$ and $y_1, \dotsc, y_n$, respectively, to
these problems. We may define a new problem $P_1 \square P_2 \square \dotsm \square P_n$ where Alice
and Bob should output $P_i(x_i, y_i)$ if $i \in [n]$ is a unique index where they have received
unequal inputs $x_i \neq y_i$ (otherwise they simply output 0).  With only a budget of $O(1)$
communication, Alice and Bob cannot identify the index $i$ where $x_i \neq y_i$.  Nevertheless, we
show that they can compute the output of this unknown instance. This generalizes the
\textsc{$1$-Hamming Distance} problem, which is obtained in this way when the base problem $\cP$
is \textsc{Equality} on a single input bit.

In graph theory terms, if $\cF$ is any graph class such that, for any graph $G \in \cF$, two players
can decide whether their vertices $x,y$ are adjacent in $G$ using a constant-cost protocol, then
there exists also a constant-cost protocol for deciding adjacency of vertices in any Cartesian
product graph $G^\square = G_1 \mathbin\square G_2 \mathbin\square \dotsm \mathbin\square G_n$,
where $\square$ denotes the Cartesian product and each $G_i \in \cF$. More generally, if deciding whether
two vertices have path-distance $k$ is a constant-cost problem in the class $\cF$ (as is the case
for planar graphs), then deciding path-distance $k$ in $G^\square$ is also constant cost.

\begin{theorem}[\cref{section:cartesian product}]
\label{thm:intro-products}
Suppose that for every constant $k$, there is a constant-cost randomized protocol for deciding
$\dist(x,y) \leq k$ in graphs $G \in \cF$. Then for every constant $k$ there is a constant-cost
randomized protocol deciding $\dist(x,y) \leq k$ in the class $\cF^\square$ of Cartesian products of graphs in~$\cF$.
\end{theorem}

Our proof that \textsc{Equality} is not a complete problem (see Question III below) also shows that
this composition takes any problem $\cP$ and turns it into a problem that does not reduce to
\textsc{Equality}. Subsequent work \cite{FGHH24} extends this composition method to find problems
which do not reduce to any \textsc{$k$-Hamming Distance}.

\cref{thm:intro-products} also allows for a new method of constructing implicit representations of Cartesian
product graphs, which were not known; we will discuss these results in
\cref{section:intro-implicit}.

\subsubsection*{Question II. What structures of a problem explain the existence or non-existence of
a constant-cost protocol?}

Constant-cost communication seems quite restrictive, and constant vs.~non-constant is the
simplest possible lower-bound question that one can ask. So, we may hope for a simple criterion to
determine whether a communication problem has a constant-cost protocol or not. We observe in
\cref{section:correspondence} that constant-cost communication problems may always be expressed as a
hereditary  class $\cF$ of graphs (\ie closed under taking induced subgraphs) where two players can
decide adjacency of vertices $x,y$ in a shared graph $G \in \cF$, using a randomized communication
protocol whose cost is independent of the number of vertices of $G$. Then the question becomes,
which hereditary graph classes $\cF$ admit constant-cost randomized protocols for deciding adjacency
in graphs $G \in \cF$? \\

\noindent
\textbf{Necessary criterion 1:} The number of $n$-vertex graphs in $\cF$ (a quantity sometimes
called the \emph{speed} of $\cF$) must be at most $2^{O(n \log n)}$ (see
\cref{section:correspondence}).  This is interesting because hereditary graph classes with
$2^{\Theta(n \log n)}$ $n$-vertex graphs are already well-studied in structural graph theory and are
sometimes called \emph{factorial classes}. Here are two of the most natural ways to obtain
hereditary factorial classes, which we will study in this paper:
\begin{enumerate}
\item Geometric intersection graphs. A class of graphs is obtained by associating each vertex with a
geometric object (\eg a line in $\bR^2$) and making two vertices adjacent if and only if their
associated objects intersect. These classes of graphs are hereditary by definition, and
constant-dimensional geometric intersection graphs generally have factorial speed (see \eg
\cite{Spin03}), so answering Question II requires understanding which subclasses of geometric
intersection graphs admit constant-cost randomized protocols for deciding adjacency. Geometric
intersection graphs are the subject of one of the main open problems about implicit graph
representations \cite{Spin03}, and they are closely related to 
\emph{unbounded error} randomized communication \cite{PS86}, making them vital for understanding
randomized communication (see discussion in \cref{section:discussion}).
\item Classes defined via forbidden induced subgraphs. Every hereditary graph class may be obtained by choosing a set
$\cH$ of graphs, and taking the class $\free(\cH)$ of graphs which do not contain any $H \in \cH$ as
an induced subgraph (see e.g. \cite[Chapter 2]{KL15}). For hereditary classes of \emph{bipartite} graphs (such as the graph classes
obtained from any constant-cost communication problem), it is known exactly which graph classes
defined by forbidding exactly one graph $H$ have factorial speed \cite{All09,LZ17}.
\end{enumerate}

\noindent\textbf{Necessary criterion 2:} \emph{Stability.}
Although the above-mentioned graph classes have factorial speed, they still may not have constant-cost
communication protocols for deciding adjacency, because they are \emph{unstable}, meaning that the
\textsc{Greater-Than} problem can be reduced to the problem of deciding adjacency in these graphs.
\textsc{Greater-Than} is a standard communication problem, where the two players receive numbers
$i,j \in [N]$ and must decide whether $i > j$. This problem has randomized communication cost
$\Theta(\log\log N)$ \cite{Nis93,Viol15}, which is exponentially smaller than the deterministic
cost $\Theta(\log N)$, but still non-constant. A graph class $\cF$ is \emph{stable} if there is a
fixed constant bound on the largest instance of \textsc{Greater-Than} which appears in the adjacency
relation in $\cF$.

\begin{example}
\label{example:interval-graph}
An interval graph is a geometric intersection graph where each vertex is represented by an interval
$[a,b] \subseteq \bZ$, and two vertices are adjacent when their intervals intersect. Two players
holding vertices of an $N$-vertex interval graph cannot, in general, determine whether their
vertices are adjacent using a constant-cost communication protocol, because one may easily construct
arbitrarily large instances of \textsc{Greater-Than} inside the adjacency relation of interval
graphs.
But this does not yet characterize the non-existence of a constant-cost protocol for this
problem, because \emph{a priori} it is not clear whether there is another structure of the problem
which would \emph{also} prevent a constant-cost protocol from existing, even if we ``remove'' the
\textsc{Greater-Than} subproblems.
\end{example}

This raises the question: For which factorial classes is \emph{stability} both necessary and
sufficient to guarantee constant cost?  For a surprising variety of examples (including interval
graphs), we show that stability precisely characterizes the existence of constant-cost protocols.
We prove:
\begin{theorem}[Informal; \cref{section:intersection graphs}]
\label{thm:intro-intersection}
For interval graphs and permutation graphs (another example of geometric intersection
graphs), deciding adjacency between two vertices is a
constant-cost problem if and only if the graphs are stable. 
\end{theorem}

\begin{theorem}[Informal; \cref{section:bipartite graphs}]
\label{thm:intro-monogenic}
For bipartite graph classes defined as above by forbidding a single bipartite graph
$H$ such that $\free(H)$ has factorial speed, deciding adjacency between vertices is again a
constant-cost problem if and only if the graphs are stable. 
\end{theorem}
These results are also motivated by the connection to implicit graph representations, specifically
\cref{question:igq-to-pug} discussed below. The proofs require new structural decompositions
of interval graphs, permutation graphs, and $H$-free bipartite graphs, and provide further new
examples of problems with constant-cost communication protocols. They also suggest (to us) a deeper
relation between communication complexity and the stability condition.  Stability plays an important
role in structural graph theory and model theory (see \cref{section:stability}). Its role in
communication complexity remains unclear (see \cref{section:discussion}). Subsequent work has shown
that stable factorial classes may not have constant-cost protocols \cite{HHH22counter,EHK22},
and stability is crucial in recent lower bounds \cite{FHHH24}.

\subsubsection*{Question III. Is \textsc{Equality} \emph{complete} for constant-cost randomized communication?}

A useful way to try to characterize the constant-cost communication problems is to find a
\emph{complete} problem: a single constant-cost communication problem that all other constant-cost
problems reduce to (see \cref{section:reductions} for the definition of reductions in the
constant-cost setting). Depending on the properties of the complete problem, this could immediately
answer many of the interesting questions about constant-cost randomized communication.
\textsc{Equality} is the first natural candidate that one might choose.

So, can every constant-cost communication protocol be rewritten in such a way that the only use of
randomness in the protocol is the computation of \textsc{Equality} instances, and the remainder of
the protocol is completely deterministic?  If this were so, we would know that the most extreme
examples of the power of randomized communication are all explained by the power of the
\textsc{Equality} protocol, and we would know that to prove non-constant lower bounds for any
problem, it suffices to prove lower bounds against \emph{deterministic} protocols with access to the
\textsc{Equality} oracle. 

However, we show that the answer is no: the \textsc{1-Hamming Distance} problem is another problem
that has a constant-cost protocol, but it cannot be reduced to \textsc{Equality}. 
\begin{theorem}[\cref{section:equality}]
\textsc{Equality} is not complete for constant-cost randomized communication.
\end{theorem}
A very different proof of this result appeared independently and concurrently in
\cite{HHH23dimfree}, who use Fourier analysis to get optimal bounds on the number of
\textsc{Equality} oracle queries required to compute \textsc{$1$-Hamming Distance}.  Our proof is
combinatorial and relies on a Ramsey-type theorem of \cite{ARSV06}, but it does not give
quantitative
bounds.

\subsection{Implicit Graph Representations}
\label{section:intro-implicit}

We now discuss our results from the perspective of efficient graph representations.  A central
observation of this paper is that constant-cost randomized communication is the natural
probabilistic version of \emph{implicit graph representations} (see \cref{section:correspondence}),
which can be introduced with the following combinatorial question. Consider a hereditary graph class
$\cF$. How tightly can we ``pack'' all the $n$-vertex graphs together? Or in other words, what is
the smallest graph into which we can put all the $n$-vertex graphs in~$\cF$?  More formally, write
$\cF_n$ for the set of graphs in $\cF$ with vertex set $[n]$.  A \emph{universal graph} (or
\emph{induced-universal graph}) for $\cF$ is a sequence $U = (U_n)_{n \in \bN}$ of graphs, such that
all $n$-vertex graphs $G \in \cF_n$ appear as induced subgraphs of $U_n$. How big must $U_n$ be? 

This combinatorial question may also be expressed as an equivalent question in
distributed computing and data structures, well-studied since \cite{Mul89,KNR92}: For a graph $G \in
\cF$, we wish to represent $G$ as efficiently as possible in an \emph{implicit} or distributed
manner: we would like to assign short binary labels $\ell(v)$ to each vertex $v$ of $G$ such that,
given two labels $\ell(x), \ell(y)$, adjacency between $x$ and $y$ in $G$ can be determined by a
\emph{decoder} which knows the class $\cF$ but not the specific graph $G$. This is an
\emph{adjacency labeling scheme}, and the existence of an adjacency labeling scheme with labels of
size $s(n)$ for the $n$-vertex graphs is equivalent to the existence of a universal graph of size
$2^{s(n)}$ \cite{KNR92}. The main open problem is:

\begin{question}[Implicit Graph Question (IGQ)]
Which hereditary graph classes admit a universal graph of size $\poly(n)$ (\ie an $O(\log n)$-bit
adjacency labeling scheme)? 
\end{question}

Adjacency labeling schemes with $O(\log n)$ bits are also known as \emph{implicit representations}.
This question has seen significant attention but little progress, although many positive examples
are known.  We introduce a probabilistic version of this question, which replaces adjacency labeling
schemes with \emph{randomized} adjacency labeling schemes, which we call \emph{adjacency sketches},
and the equivalent notion of \emph{probabilistic universal graphs (PUGs)}. Randomized adjacency
labels for trees were introduced in \cite{FK09}, and have been studied more generally in connection
to communication complexity by \cite{Har20}; see \cref{def:pug,def:sketch}. We are interested in the
classes that admit \emph{constant-size} adjacency sketches (equivalently, constant-size PUGs), so we
ask:

\begin{question}[PUG Question]
\label{question:main} Which hereditary graph classes admit a PUG of constant size (\ie, a
constant-size adjacency sketch)?
\end{question}

This question is interesting for three main reasons. First, we think it is both a natural question
on its own, as well as being a natural probabilistic variant of the IGQ, and it is surprising that,
as we will see, a rich variety of graph classes can be randomly projected into a
\emph{constant-size} universal graph.

Second, answering the PUG question is in fact \emph{equivalent} to characterizing all communication
problems with constant-cost randomized protocols. To see this, we make two observations (see
\cref{section:correspondence}):
\begin{itemize}
\item As mentioned above, any communication problem with a constant-cost randomized protocol is
equivalent to a hereditary graph class with a constant-cost protocol for deciding adjacency.
\item This constant-cost protocol for adjacency may be transformed into a constant-size adjacency
\emph{sketch}, or constant-size PUG. In the other direction, a constant-size
adjacency sketch may be used as a constant-cost communication protocol for deciding adjacency.
\end{itemize}

Third, a simple derandomization argument shows that any class with a constant-size adjacency sketch
also admits a size $O(\log n)$ adjacency labeling scheme, so the PUG question is an interesting
special case of the IGQ, and it also leads to new techniques for constructing adjacency labeling
schemes; for example, in Question I (b) above, we showed that Cartesian products preserve
constant-size adjacency sketches, which imply new adjacency labeling schemes for Cartesian product
graphs that were previously unknown (see \eg \cite{CLR20}). Since little progress has been made
towards answering the IGQ, the PUG question may serve as an interesting step towards the IGQ, which
makes stronger demands on the graph class.  Since PUGs and adjacency sketches are stronger than
universal graphs and adjacency labeling schemes, one way to start understanding the
relationship between the two is to ask when we can strengthen known adjacency labeling
schemes to adjacency sketches; in other words:
\begin{question}
\label{question:igq-to-pug}
Which hereditary graph classes with a universal graph of size $\poly(n)$ also admit a
constant-size PUG?
\end{question}
This question motivates many of the examples we study in this paper (particularly
interval and permutation graphs, and graphs of bounded twin-width). It is also related to
Question III, since several standard examples of adjacency labeling schemes (\eg for
bounded-arboricity graphs \cite{KNR92}) can in fact be rephrased as reductions to the
\textsc{Equality} problem. More generally, the most natural notion of a \emph{reduction}
between constant-cost communication problems (see \cref{section:reductions}) is equivalent to a natural notion of
reduction between \emph{graph classes}, for the purpose of studying adjacency labeling (which were
also proposed independently in the literature on the IGQ in \cite{Chan23}).

\subsection{Subsequent Work} 
\label{section:discussion}

The concurrent work of Hambardzumyan, Hatami, \& Hatami \cite{HHH23dimfree} initiates the study of
constant-cost communication from a different perspective and proposes a number of interesting
conjectures about the structure of constant-cost problems. Following that paper and the conference
version of the present work \cite{HWZ22}, there have been several follow-up works on communication
complexity and implicit representations \cite{HHH22counter, EHK22, EHZ23, HZ23}, and numerous others
related to constant-cost communication \cite{HHPTZ22,ACHS23,CHHS23,HHM23,FHHH24,FGHH24} and
adjacency sketching \cite{NP24}; we briefly describe a few of these developments here, and we refer
the reader to the subsequent works for discussions of open problems.

\paragraph*{No complete problem.} This paper and \cite{HHH23dimfree} both show that the \EQUALITY
problem is not ``complete'' for the class of constant-cost randomized communication problems.
Subsequent work \cite{FHHH24} shows that in fact there is \emph{no} complete problem.

\paragraph*{New implicit representations, and generalized adjacency sketches.} Extending the
techniques in this paper, \cite{EHZ23} prove optimal bounds on the size of universal graphs for any
class of subgraphs of Cartesian product graphs, which improved upon the best-known bounds of earlier
work \cite{CLR20} and showed that they satisfy the IGQ. Recent work \cite{NP24} has generalized the
notion of adjacency sketching to \emph{adversarial environments}.

\paragraph*{Sign-rank and geometric intersection graphs.}
Informally, a family of matrices has \emph{bounded sign-rank} if each matrix can be represented as a
point-halfspace incidence matrix in a constant-dimensional space. Examples include the interval and
permutation graphs studied in this article. A central open problem about implicit graph
representations is whether all geometric intersection graphs (formally, the \emph{semi-algebraic}
graph classes) have universal graphs of size $\poly(n)$. It is sufficient to consider graph classes
whose adjacency matrices have bounded sign-rank (see \eg \cite{Spin03,Fit19,HZ23}). Understanding
the matrices which have both bounded sign-rank and constant randomized communication cost has also
emerged as one of the main open problems in constant-cost communication, one reason being that
sign-rank is equivalent to \emph{unbounded-error} randomized communication \cite{PS86}; see
\cite{HHPTZ22,ACHS23,CHHS23,HHM23,HZ23}. So, bounded sign-rank is central to both implicit graph
representations and constant-cost communication.

\paragraph*{The implicit graph conjecture.}
As mentioned above, any hereditary graph class with a $\poly(n)$-size universal graph must
contain at most $2^{O(n \log n)}$ distinct $n$-vertex graphs. A long-standing
conjecture was that this condition is also \emph{sufficient}:\\

\noindent\emph{Implicit graph conjecture} \cite{KNR92,Spin03}:  A hereditary graph class
has a universal graph of size $\poly(n)$ (\ie, an $O(\log n)$-size adjacency labeling scheme) if
and only if it has at most factorial speed.\\

In the preprint of this paper, motivated by the idea that constant-size PUGs are the probabilistic
version of $\poly(n)$-size universal graphs, we formulated a probabilistic version of the implicit graph conjecture. Based on the fact that stability plays a fundamental role in the structure of hereditary
graph classes (\eg it plays a role in certain ``speed thresholds'', see \cref{section:graph theory},
and is important in model theory and model checking), and the observation that, in each example
considered in this paper, stability was the condition that determined whether cosntant-size PUGs
exist, we proposed the following as a natural probabilistic variant of the implicit graph
conjecture:\\

\noindent\emph{PUG Conjecture}: A hereditary graph class has a constant-size PUG if and only
if it is both stable and has at most factorial speed.\\

Shortly after the preprint became public, Hambardzumyan, Hatami, \& Hatami refuted our PUG
Conjecture \cite{HHH22counter} using a construction from \cite{HHH23dimfree}, and Hatami \& Hatami
extended this refutation to refute the implicit graph conjecture itself \cite{HH22}.  This approach
was further extended to refute variants of the conjecture for ``small'' graph classes (an important
subfamily of factorial classes \cite{BDSZZ24}) and monotone graph classes \cite{BDSZZ23}. More
examples refuting the PUG conjecture are now known, including (among others) the subgraphs of
hypercubes \cite{EHK22}. \cite{EHK22} has resolved the PUG question (\cref{question:main}) for
\emph{monotone} graph classes.

\subsection{Organization}
\noindent
The paper is organized in order of Questions I -- III asked in \cref{section:intro-communication}.
\begin{description}
\item[\cref{section:preliminaries}: Preliminaries.] Notation and standard definitions of
communication complexity, adjacency labeling, and universal graphs.
\item[\cref{section:definitions}: Definitions and connections.] Adjacency sketching, probabilistic universal
graphs, the correspondence between communication complexity and hereditary graph classes, and
stability.
\item[\cref{sec:new_examples}, Question I -- New examples.] Two new examples of constant-cost communication (or
constant-size adjacency sketches): (1) first-order formulae (including small distances) for planar
graphs and their generalization, and (2) adjacency in Cartesian products of graphs.
\item[\cref{section:stability}, Question II -- Structure and stability.] Some examples where stability is both necessary and
sufficient to guarantee constant-cost communication or constant-size PUGs: (1) subclasses of
permutation graphs, (2) subclasses of interval graphs, and (3) subclasses of monogenic bipartite
graph classes.
\item[\cref{section:equality}, Question III -- \textsc{Equality} is not complete.] A proof that the
\textsc{Equality} problem is not complete for constant-cost communication.
\end{description}

\section{Preliminaries}
\label{section:preliminaries}

We introduce in this section some notation and standard notions of communication complexity,
adjacency labeling schemes, and universal graphs. Since this paper is intended for readers who may
be unfamiliar with communication complexity or adjacency labeling schemes, we attempt to include all
relevant definitions here.

\subsection{Notation and Terminology}
\label{section:notation}

We write $\ind{A}$ for the indicator of event $A$; \ie, the function which is $1$ if and only if
statement $A$ is true. For a finite set $X$, we write $x \sim X$ when $x$ is a random variable drawn
uniformly at random from $X$.

\paragraph{Graphs.}
All graphs in this work are simple, \ie, undirected, without loops and multiple edges.  Let
$G=(V,E)$ be a graph. We write $\Adj_G \in \zo^{|V| \times |V|}$ for the adjacency matrix of $G$.
For a subset $W \subseteq V$, we write $G[W]$ to denote the subgraph of $G$ \emph{induced by}~$W$,
\ie, a graph obtained from $G$ by removing all vertices in $V \setminus W$.  A graph $H$ is an
\emph{induced subgraph} of $G$ if $H = G[W]$ for some $W \subseteq V$.  We write $H \sqsubset G$ to
denote the fact that $H$ is an induced subgraph of $G$.  We also write $\overline G$ for the
\emph{complement} of $G$, \ie, the graph $(V,\overline E)$ where $(x,y)\in \overline E$ if and only
if $(x,y)\notin E$.  A set of pairwise non-adjacent vertices in a graph is an independent set, and a
set of pairwise adjacent vertices is a clique.
For two graphs $G = (V, E)$ and $H = (V', E')$ on disjoint vertex sets, the $G + H := (V \cup V', E
\cup E')$ is called the \emph{disjoint union} of $G$ and $H$.
We denote by $K_n$, $P_n$, and $C_n$, respectively a complete graph, a path, and a cycle, each on
$n$ vertices.

\paragraph{Bipartite graphs.}
A \emph{colored} bipartite graph is a bipartite graph with a given bipartition of its vertex set.
We denote a colored bipartite graph by a triple $(X,Y,E)$, where $X,Y$ is the partition of its
vertex set into two parts, and the function $E : X \times Y \to \zo$ defines the edge relation. If a
bipartite graph $G$ is connected, it has a unique partition of its vertices into two parts and
therefore there is only one colored bipartite graph corresponding to $G$; (note that $(X,Y,E)$ and
$(Y,X,E)$ are considered the same colored bipartite graph).  If $G$ is disconnected, however, there
is more than one corresponding colored bipartite graph.  The \emph{bipartite adjacency matrix} of a
colored bipartite graph $G = (X,Y,E)$ is the matrix $A \in \zo^{|X| \times |Y|}$ with $A(i,j) = 1$
if and only if $(i,j) \in E$. Abusing notation, we also write $\Adj_G$ for the bipartite adjacency
matrix of $G$.

For colored bipartite graphs $G = (X,Y,E)$ and $H = (X',Y',E')$, we say that $H$ is an induced
subgraph of $G$, and write $H \sqsubset G$, when there is an injective map $\phi : X' \cup Y' \to X
\cup Y$ that preserves adjacency and preserves parts. The latter means that the images $\phi(X')$
and $\phi(Y')$ satisfy either $\phi(X') \subseteq X, \phi(Y') \subseteq Y$ or $\phi(X') \subseteq Y,
\phi(Y') \subseteq X$.  A colored bipartite graph $G = (X,Y,E)$ is called \emph{biclique} if every
vertex in $X$ is adjacent to every vertex in $Y$, and $G$ is called \emph{co-biclique} if $E =
\emptyset$.

For a graph $G=(V,E)$ and two disjoint sets $X,Y \subseteq V$, we write $G[X,Y]$ for the colored bipartite graph $(X,Y,E')$ where for $(x,y) \in X \times Y$, $(x,y) \in E'$ if and only if $(x,y) \in E$. A bipartite graph $H$ is a \emph{semi-induced subgraph} of $G$ if there exist disjoint sets $X,Y \subseteq V$ such that $H = G[X,Y]$.

The \emph{bipartite} complement, $\bc{G}$,
of a colored bipartite graph $G=(X,Y,E)$ is the graph $\bc{G} := (X,Y,\bc{E})$ with $(x,y)\in \bc{E}$ if and only if $(x,y)\notin E$ for $x\in X$, $y\in Y$.

\paragraph{Classes of graphs and bipartite graphs.} We define the following for graphs; the same
discussion also applies to bipartite graphs. A \emph{class} of graphs is a set of graphs that is
closed under isomorphism. A class of graphs is \emph{hereditary} if it is closed under taking
induced subgraphs.  It is well-known that any hereditary graph class can be defined by its set of
\emph{minimal forbidden induced subgraphs} (see e.g. \cite[Theorem 2.1.3]{KL15}). That is, for any hereditary class $\cF$, there is a
\emph{unique minimal} set of graphs $\cH$ such that $\cF$ is the class \emph{$\cH$-free} graphs,
\ie, $\cF = \free(\cH)$, where
\[
\free(\cH) \define \{ G : \forall H \in \cH, H \not\sqsubset G \} \,.
\]
For a class of  graphs $\cG$, its \emph{hereditary closure} $\cl(\cG)$ is the set of all induced
subgraphs of graphs in $\cG$, \ie, $\cl(\cG) = \{ H : \exists G \in \cG, H \sqsubset G \}$. Note, by
definition, $\cl(\cG)$ is a (minimal) hereditary graph class that includes $\cG$.

\subsection{Communication Complexity}
\label{section:prelim cc}

We define some basic concepts in communication complexity and refer the reader to \cite{NK96,RY20}
for an introduction to communication complexity. Typically, a communication problem is a sequence $f =
(f_n)_{n \in \bN}$ of functions\footnote{In the literature, the domain is usually $\zo^n \times
\zo^n$. We use $[n] \times [n]$ to highlight the graph interpretation.} $f_n : [n] \times [n] \to
\zo$. In this paper it will be more convenient to define a \emph{communication problem} as a set of
(not necessarily square) matrices $\cM$. For a fixed \emph{communication matrix} $M \in \cM$, we
write $\CC(M)$ for the cost of the optimal two-way, randomized communication protocol computing $M$,
defined as follows.

Informally, two players Alice and Bob share a source of randomness. Alice receives input $x$ which
is a row index of $M$, Bob receives input $y$ which is a column index of $M$, and they communicate by sending
messages back and forth using their shared randomness. After communication, Bob must output a
(random) value $b \in \zo$ such that $b = M(x,y)$ with probability at least $2/3$. The cost of
such a protocol is the maximum, over all inputs $x, y$, of the number of bits communicated between the
players. Formally, the definition is as follows.

\begin{definition}
\label{def:communication protocol}
A two-way public-coin communication protocol is a probability distribution $\cD$ over
\emph{communication trees}.  For communication matrix $M \in \zo^{m \times n}$, a communication tree
$T$ is a binary tree with each inner node being a tuple $(p,\mu)$ where $p \in \{A,B\}$ and either
$\mu : [n] \to \zo$ or $\mu : [m] \to \zo$ depending on whether $p = A$ or $p = B$. Each edge of $T$
is labeled either $0$ or $1$. Each leaf node is labeled either 0 or 1. For any fixed tree $T$ and
inputs $x \in [m]$, $y \in [n]$, communication proceeds by setting the current node $c$ to the root
node.  At each step of the protocol, if $c$ is an inner node $(A,\mu)$ then Alice sends $\mu(x)$ to
Bob and both players set $c$ to the child along the edge labeled $\mu(x)$. If $c$ is an inner node
$(B,\mu)$ then Bob sends $\mu(y)$ to Alice, and both players set $c$ to the child along the edge
labeled $\mu(y)$.  The protocol terminates when $c$ becomes a leaf node, and the output is the value
of the leaf node; we write $T(x,y)$ for the output of communication tree $T$ on inputs $x,y$.

The distribution $\cD$ is required to satisfy the condition that, for all inputs $x,y$, $T(x,y) = M(x,y)$ with probability
at least $2/3$ over the random choice of $T \sim \cD$, and the \emph{cost} of the protocol is the
largest depth $d$ of a tree $T$ in the support of $\cD$. We write $\CC(M)$ for the minimum cost of a
protocol computing $M$.

For a communication problem $\cM$, we denote by $\cM_n$ the subset of matrices $M \in
\cM$ with $r$ rows and $c$ columns such that $r + c = n$. Then we define the complexity
$\CC(\cM)$ of the problem $\cM$ as the function
\[
  n \mapsto \max_{M \in \cM_n} \CC(M) \,.
\]
\end{definition}

\subsection{Adjacency Labeling Schemes, Universal Graphs, and Factorial Classes}

A \emph{graph class} $\cF$ is a set of labeled graphs, closed under isomorphism. It is
\emph{hereditary} if it is also closed under vertex deletion (\ie it is closed under taking induced
subgraphs). For any graph class $\cF$ and any $n \in \bN$, we write $\cF_n$ for the set of graphs $G
\in \cF$ with vertex set $[n]$.

Adjacency labeling schemes were introduced in \cite{Mul89, KNR92} and are defined as follows.

\begin{definition}[Adjacency labeling scheme]
Let $\cF$ be a hereditary graph class. An \emph{adjacency labeling scheme} for $\cF$ of size $s(n)$
consists of a \emph{decoder} $D : \zo^* \times \zo^* \to \zo$ such that, for all $n \in \bN$ and all
$n$-vertex graphs $G \in \cF_n$, there exists a \emph{labeling} $\ell_G : V(G) \to \zo^{s(n)}$ of the
vertices of $G$, where
\[
  \forall x,y \in V(G) : \qquad \{x,y\} \in E(G) \iff D(\ell_G(x), \ell_G(y)) = 1 \,.
\]
\end{definition}
A hereditary graph class which admits an adjacency labeling scheme of size $s(n) = O(\log n)$ is
sometimes said to have an \emph{implicit representation}, and characterizing the graph classes which
have an implicit representation is the main open problem about adjacency labeling schemes, which we
call the IGQ.

Universal graphs were introduced in \cite{Rado64}. A universal graph (or induced-universal graph) is
defined as

\begin{definition}
Let $\cF$ be a hereditary graph class. A \emph{universal graph} for $\cF$ of size $s(n)$ is a
sequence $U = (U_n)_{n \in \bN}$ of graphs with $|U_n| = s(n)$, such that for all $n \in \bN$ and all $G \in
\cF_n$, $G$ is an induced subgraph of $U_n$.
\end{definition}

As observed in \cite{KNR92}, these concepts are equivalent:
\begin{proposition}[\cite{KNR92}]
A hereditary graph class $\cF$ has an adjacency labeling scheme of size $\lceil s(n) \rceil$ if and
only if it has a universal graph of size $2^{s(n)}$.
\end{proposition}

Implicit representations, or adjacency labeling schemes of size $O(\log n)$, are therefore
equivalent to universal graphs of size $\poly(n)$. These graph classes satisfy a certain bound on
the number of $n$-vertex graphs $|\cF_n|$:

\begin{proposition}[\cite{KNR92}]
\label{prop:knr}
If a hereditary graph class $\cF$ has a universal graph of size $2^{O(\log n)}=\poly(n)$
(\ie, an adjacency labeling scheme of size $O(\log n)$), then $|\cF_n|
\leq \binom{\poly(n)}{n} = 2^{O(n \log n)}$.
\end{proposition}

The function $n \mapsto |\cF_n|$ is called the \emph{speed} of the graph class. Graph classes with
speed $2^{\Theta(n \log n)}$ are said to have \emph{factorial speed}, and are called
\emph{factorial classes}.  In \cref{section:graph theory} we summarize some known results about the family of factorial hereditary graph classes.

\section{Connections: Randomized Communication and Implicit Representations}
\label{section:definitions}

In this section we define the central concepts of the paper and prove some basic properties,
including the connection between randomized communication and implicit representations.
\begin{description}[itemsep=0pt]
\item[\cref{section:adjacency-labeling}:] We define adjacency sketching and probabilistic
universal graphs.
\item[\cref{section:correspondence}:] We explain the correspondence between communication
complexity, adjacency sketching, and implicit graph representations.
\item[\cref{section:reductions}:] We define notions of \emph{reductions} between families of
matrices, communication problems, and graph classes.
\item[\cref{section:equality-based-labeling}:] We introduce some notation for
\emph{equality-based} communication and labeling, which will be used often throughout the paper.
\item[\cref{sec:bounded-arboricity}:] We state and prove some basic results about adjacency
sketching, which will be used throughout the paper.
\end{description}

\subsection{Adjacency Sketching and Probabilistic Universal Graphs}
\label{section:adjacency-labeling}

Probabilistic versions of adjacency labeling schemes and universal graphs were introduced in
\cite{Har20}, which we call \emph{adjacency sketches} and \emph{probabilistic universal graphs
(PUGs)}, respectively. A similar definition with one-sided error is given in \cite{FK09}; see that
paper for some results on the limitations of one-sided error. In the field of sublinear algorithms,
a \emph{sketch} reduces a large or complicated object to a smaller, simpler one that supports
(approximate) queries.  An adjacency sketch randomly reduces a hereditary graph class to a PUG that
supports adjacency queries.

\begin{definition}[Probabilistic Universal Graph]
\label{def:pug}
A \emph{probabilistic universal graph sequence} for a graph class $\cF$ (or \emph{probabilistic
universal graph} for short) with error $\delta$ and size $m(n)$ is a sequence $U = (U_n)_{n \in
\bN}$ of graphs with $|U_n|=m(n)$ such that, for all $n \in \bN$ and all $G \in \cF_n$, the
following holds: there exists a probability distribution over maps $\phi : V(G) \to V(U_n)$ such
that
\[
  \forall u,v \in V(G) : \qquad
  \Pru{\phi}{\vphantom{\Big|} \ind{(\phi(u),\phi(v)) \in E(U_n)} = \ind{(u,v) \in E(G)} } \geq 1-\delta \,.
\]
We say that a graph class $\cF$ admits a \emph{constant-size PUG} if there exists a probabilistic
universal graph sequence for $\cF$ with error $\delta = 1/3$ and size $m(n) = O(1)$.
\end{definition}

\begin{definition}[Adjacency Sketch]
\label{def:sketch}
For a graph class $\cF$, an \emph{adjacency sketch} with size $c(n)$ and error $\delta$ is a pair
of algorithms: a randomized \emph{encoder} and a deterministic \emph{decoder}. On input $G \in
\cF_n$, the encoder outputs a (random) function $\sk : V(G) \to \zo^{c(n)}$. The encoder and
(deterministic) decoder $D : \zo^* \times \zo^* \to \zo$ satisfy the condition that for all $G \in
\cF$,
\[
    \forall u,v \in V(G) : \qquad \Pru{\sk}{\vphantom{\Big|} D(\sk(u),\sk(v)) = \ind{(u,v) \in E(G)} } \geq
1-\delta \,.
\]
\end{definition}
In both definitions, we assume $\delta=1/3$ unless otherwise specified.  Setting $\delta=0$ we
obtain the (deterministic) labeling schemes of \cite{KNR92}.  We will write $\RL(\cF)$ for the
smallest function $c(n)$ such that there is an adjacency sketch for $\cF$ with size $c(n)$ and error
$\delta=1/3$. 

PUGs are equivalent to adjacency sketches, by the same argument which shows that universal graphs
are equivalent to adjacency labeling schemes \cite{KNR92}: identify the
vertices of $U_n$ with the binary strings $\zo^{c(n)}$ for $c(n) = \ceil{\log |U_n|}$, identify the
random sketch $\sk : V(G) \to \zo^{c(n)}$ with the map $\phi : V(G) \to V(U_n)$, and identify the
decoder $D$ with the edge relation on $U_n$. We get:

\begin{proposition}
\label{prop:pug labeling equivalence}
A hereditary class $\cF$ has a constant-size PUG if and only if $\RL(\cF) = O(1)$.
\end{proposition}

The choice of $\delta=1/3$ in the definition of constant-size PUGs and adjacency sketches
is arbitrary: by standard probability boosting arguments, one may achieve any constant error
$\delta < 1/2$ while incurring only a constant-factor cost; we provide a proof for the sake of
completeness in \cref{section:missing-proofs}.

\begin{restatable}{proposition}{propboosting}
\label{prop:boosting}\RestateRemark
Let $\cF$ be a class of graphs. For any $\delta \in (0,1/2)$, there is an adjacency sketch with
error $\delta$ and size at most $O(\RL(\cF_n) \cdot \log\frac{1}{\delta})$.  Equivalently, if there
is a PUG $U = (U_n)_{n \in \bN}$ for $\cF$ with size $|U_n|$, then there is a PUG $U' = (U'_n)_{n \in \bN}$ for $\cF$ with
error $\delta$ and $|U'_n| \leq |U_n|^{O(\log(1/\delta))}$.
\end{restatable}

Adjacency sketches can be derandomized to obtain adjacency labeling schemes. This was observed in
\cite{Har20} as a consequence of a bound on the number of bits of randomness that are required, and
can also be proved by simple derandomization; we give a proof for the sake of completeness in
\cref{section:missing-proofs}. It is sometimes required that adjacency labeling schemes be produced by
an efficient algorithm, and we observe that efficiency is preserved by this derandomization:

\begin{restatable}{proposition}{propderandomization}
\label{lemma:derandomization}\RestateRemark
For any hereditary graph class $\cF$, there is an adjacency labeling scheme of size
$O(\RL(\cF) \cdot \log n)$. If there is a randomized algorithm which produces the adjacency sketch
in time $\poly(n)$, then there is a randomized algorithm which produces the adjacency labels in
expected time $\poly(n)$.
\end{restatable}

It follows that hereditary graph classes which admit constant-size PUGs have the $\poly(n)$-vertex
universal graphs, asked for by the IGQ, as illustrated in \cref{fig:prequel}. 

\begin{proposition}
\label{prop:correspondence inclusion}
If a hereditary class $\cF$ has a constant-size PUG (\ie, $\RL(\cF)=O(1)$) then it is a positive example to the IGQ
(\ie, it admits a universal graph of size $\poly(n)$).
\end{proposition}

\begin{remark}
A further motivation for these objects from the perspective of communication complexity, given in
\cite{Har20}, is that adjacency sketching (respectively, labeling) is a generalization of the
randomized (resp. deterministic) \emph{simultaneous message passing (SMP)} model of communication.
In this generalization, the referee who must compute $f(x,y)$ from the messages of the players does
not know $f$ in advance; they only know a certain class $\cF$ which contains $f$, and the players
must include sufficient information in their messages to compensate for the ignorance of the
referee. 
\end{remark}

\subsection{Communication-to-Graph Correspondence}
\label{section:correspondence}

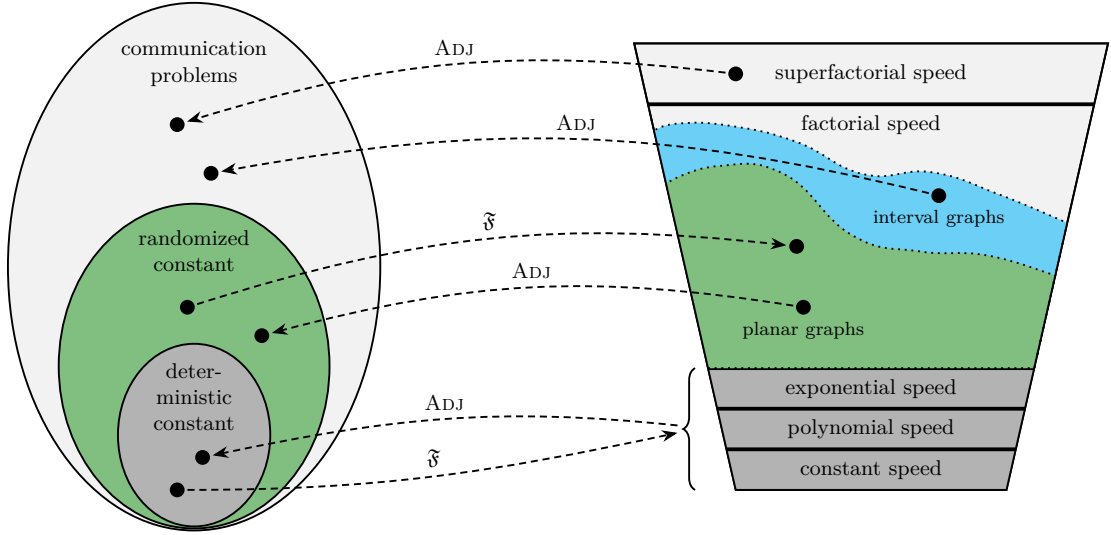
\begin{figure}[tbp]
	\centering 
	\scalebox{1}{\begin{tikzpicture}[
			xscale=.6,
			yscale=.75,
			every node/.style={font=\footnotesize},
	]
		
		\foreach \n/\x/\y in {ll/-4/-1,lr/4/-1,ul/-7/10,ur/7/10} 
			{ \coordinate (\n) at (\x,\y) ; }
		\foreach \s in {l,r} { \coordinate (bell-\s) at ($ (l\s)!.45!(u\s) $) ; }
		\foreach \s in {l,r} { \coordinate (abovebell-\s) at ($ (l\s)!.53!(u\s) $) ; }
		\def\height{11}
		\draw[fill=black!5,thick] (ur) -- (lr) -- (ll) -- (ul) -- cycle ;
		
		\foreach \c/\y in {constant/0,polynomial/1,exponential/2,factorial/8.5} {
			\foreach \s in {l,r} { \coordinate (\c-\s) at ($ (l\s)!{(\y+1)/\height}!(u\s) $) ; }
		}
		\fill[const-det] (exponential-r) -- (lr) -- (ll) -- (exponential-l) -- cycle ;
		\foreach \c/\y in {constant/0,polynomial/1,exponential/2,factorial/8.5} {
			\node at (0,\y-.5) {\c\ speed};
			\draw[ultra thick] (\c-l) -- (\c-r) ;
		}
		\node at (0,9.25) {superfactorial speed};
		\def\cpug{0.33}
		
		\coordinate (cPUG-b) at ($ (exponential-l)!\cpug!(exponential-r) $) ;
		\coordinate (bell-m) at ($ (bell-l)!\cpug!(bell-r) $);
		\coordinate (abovebell-m) at ($ (abovebell-l)!\cpug!(abovebell-r) $);
		\coordinate (um) at ($ (ul)!\cpug!(ur) $);
		
		\draw[thick,dotted,igc-positive] (exponential-l) -- (exponential-r) 
			-- ($ (exponential-r)!.55!(factorial-r)$) 
			--plot[smooth,tension=.8] coordinates 
				{++(0,0) ++(-2,.7) ++(-4,1.25) ++(-6,1.2) ++ (-8,2) ++(-10,2.4) ++(-11.5,2.45) 
						($ (exponential-l)!.9!(factorial-l)$)}
			-- ($ (exponential-l)!.96!(factorial-l)$) 
		;
		\draw[thick,dotted,const-rand] (exponential-l) -- (exponential-r) 
			-- ($ (exponential-r)!.35!(factorial-r)$) 
			--plot[smooth,tension=.8] coordinates 
				{++(0,0) ++(-2,.5) ++(-4,.6) ++(-6,1) ++ (-8,2.5) ++ (-10,2.7) ($ (exponential-l)!.7!(factorial-l)$)}
			-- ($ (exponential-l)!.7!(factorial-l)$) 
		;
		\draw[thick] (ur) -- (lr) -- (ll) -- (ul) -- cycle ;

		\foreach \n/\x/\y in {
					const-PUG/-2.2/5,
					superfac/-4/9.25
			} {
			\node[circle,fill=black,inner sep=2.2pt] (\n) at (\x,\y) {};
		}
		\foreach \n/\x/\y in {planar/-2/3.5} {
			\node[circle,fill=black,inner sep=2.2pt,label=-90:{\scriptsize planar graphs}] (\n) at (\x,\y) {};
		}
		\foreach \n/\x/\y in {interval/2/6.25} {
			\node[circle,fill=black,inner sep=2.2pt,label=-90:{\scriptsize interval graphs}] (\n) at (\x,\y) {};
		}
		
		\draw[mybrace,decorate] (ll -| -5.2,0) -- 
				node[left=5pt,align=right,font=\scriptsize,inner sep=1pt] (subfactorial) {} 
				(exponential-l -| -5.2,0) ;

		\begin{scope}[shift={(-20,4)}]
			\draw[thick,comm-problems] (0,0.5) circle [x radius=5.5,y radius=6.5];
			\draw[thick,const-rand] (0,-1.95) circle [x radius=4,y radius=4];
			\draw[thick,const-det] (0,-3.9+.25) circle [x radius=2.25,y radius=2.25];
			\foreach \l/\y in {communication\\ problems/5.5,randomized\\constant/.9,deter-\\ministic\\constant/-2.7} {
				\node[align=center] at (0,\y) {\l} ;
			}
			
			\foreach \n/\x/\y in {
						det-const-1/-.5/-5,%
						det-const-2/.25/-4.2,%
						mapped-planar/2/-1.2,%
						rand-const/-.2/-.5,%
						mapped-interval/.5/2.8,%
						comm-prob/-.5/4%
				} {
				\node[circle,fill=black,inner sep=2.2pt] (\n) at (\x,\y) {};
			}
			
		\end{scope}
		
		\draw[mapping] (det-const-1) to[out=0,in=190] node[above] {$\mathfrak F$} (subfactorial.south) ;
		\draw[mapping] (subfactorial.north) to[out=175,in=10] node[above] {$\textsc{Adj}$} (det-const-2);
		\draw[mapping] (planar) to[bend right=10] node[above] {$\textsc{Adj}$} (mapped-planar) ;
		\draw[mapping] (rand-const) to[bend left=10] node[above] {$\mathfrak F$} (const-PUG) ;
		\draw[mapping] (interval) to[bend right=10] node[above] {$\textsc{Adj}$} (mapped-interval) ;
		\draw[mapping] (superfac) to[bend right=10] node[above] {$\textsc{Adj}$} (comm-prob) ;

	\end{tikzpicture}}
\caption{The correspondence that motivates this paper (\cref{prop:correspondence 1}).
\cref{section:graph theory} describes the lattice on the right. Communication
problems with constant-cost randomized protocols are mapped to the set of hereditary graph
classes with constant-size PUGs (and therefore $\poly(n)$ universal graphs by
\cref{prop:correspondence inclusion}) by $\mathfrak{F}$.
Families with constant-size PUGs are mapped to constant-cost communication
problems by $\textsc{Adj}$.}
\label{fig:prequel}
\end{figure}

We now explain that communication problems with constant-cost randomized communication protocols
are equivalent to hereditary graph classes with constant-size PUGs. This correspondence is
illustrated in \cref{fig:prequel} 
%For any undirected graph $G = (V,E)$,
%ify $E$ with the function $E : V \times V \to \zo$ where $E(x,y)$ holds true ($E(x,y)=1$) if
% if $(x,y)$ is an edge of $G$. In this paper, bipartite graphs $G = (X,Y,E)$ are colored;
%y are defined with a fixed partition of the vertices into parts $X$ and $Y$.}

Let $\cM$ be any communication problem, which we recall is a set of matrices. For any $M \in \cM$
with $M \in \zo^{n_1 \times n_2}$, we may consider the bipartite graph $G_M = ([n_1], [n_2], E)$
with edge $(i,j) \in E$ if and only if $M(i,j) = 1$.  We define the hereditary class
$\mathfrak{F}(\cM)$ associated with $\cM$ as the smallest hereditary class that contains each $G_M$. Formally,
\[
  \mathfrak{F}(\cM) \define \cl\left( \{ G_M \;:\; M \in \cM \} \right) \,.
\]
One may equivalently think of $\,\mathfrak{F}(\cM)$ as being the communication problem containing
every matrix which appears as a permutation of a submatrix of $\cM$.  It is important to keep in
mind that, due to the fact that it is hereditary, $\mathfrak{F}(\cM)$ can have unintuitive
consequences for the communication complexity: for example, some problems $\cM$ with small (but
non-constant) randomized communication complexity can suddenly have maximum communication complexity
when replaced with $\mathfrak{F}(\cM)$; see \cref{example:hamming distance}.

In the other direction, for any set of graphs $\cF$, we define the natural \textsc{Adjacency}
communication problem, which captures the complexity of the two-player game of deciding
whether the players are adjacent in a given graph $G \in \cF$. 
\[
  \Adj_\cF \define \{ M \;|\; M \text{ is the adjacency matrix of some } G\in \cF \} \,.
\]

We may now state the formal equivalence between constant-cost communication and constant-size PUGs:

\begin{restatable}{proposition}{cctolabeling}
\label{prop:correspondence 1}\RestateRemark
For any communication problem $\cM$ and hereditary graph class $\cF$:
\begin{enumerate}
\item $\cF$ has a constant-size PUG if and only if $\CC(\Adj_\cF) = O(1)$.
\label{prop:cctolabeling-item-1}
\item $\CC(\cM) = O(1)$ if and only if $\,\mathfrak{F}(\cM)$ has a constant-size PUG
(\ie, $\RL(\mathfrak{F}(\cM))=O(1)$).
\label{prop:cctolabeling-item-2}
\end{enumerate}
\end{restatable}
\begin{proof}
We start by proving \cref{prop:cctolabeling-item-1}.
First, suppose that $\cF$ is a hereditary graph class with $\RL(\cF)=O(1)$. Let $D$ be the decoder of
the constant-cost adjacency sketch, and for any graph $G \in \cF$ write $\Phi_G$ for the
distribution over sketches for $G$. We obtain a constant-cost communication protocol for $\Adj_\cF$
as follows. Let $M \in \Adj_\cF$ so that $M$ is the adjacency matrix of some graph $G \in \cF$ and
we may think of the rows and columns of $M$ as being indexed by vertices of $G$. On
inputs $x,y \in V(G)$, Alice and Bob sample $\sk \sim \Phi_G$ and Alice sends $\sk(x)$ to Bob,
which requires at most $\RL(\cF)$ bits of communication. Then Bob simulates the decoder on
$D(\sk(x),\sk(y))$ and sends the result to Alice. By definition
\[
  \Pru{\sk \sim \Phi_G}{ \vphantom{\big|} D(\sk(x),\sk(y)) = M(x,y) } \geq 2/3 \,.
\]
Now suppose that $\CC(\Adj_\cF) = O(1)$. Then there is a constant $d$ such that for any $G \in
\cF_n$ it holds that the adjacency matrix $M \in \zo^{n \times n}$ of $G$ satisfies $\CC(M) \leq d$.
For each $G \in \cF$, let $\cP(G)$ be the probability distribution over communication trees defined
by an optimal communication protocol for the edge relation of $G$. Then it holds that every
communication tree in the support of $\cP(G)$ has depth at most $d$. We define the adjacency sketch
for $\cF$ as follows. For every $G=(V,E) \in \cF$, construct the random sketch $\sk$ by sampling $T
\sim \cP(G)$, and then for every~$v \in V$:

\begin{itemize}
    \item[]\begin{quote}
For every node $c$ of $T$, append to the label $\sk(v)$ the following:
\begin{enumerate}
  \item If $c$ is an inner node $(p,m)$ (with $p \in \{A,B\}$ and $m : [n] \to \zo$), append the
      symbol $p$ and the value $m(v)$.
  \item If $c$ is a leaf with value $b$, append the symbol $L$ and the value $b$.
\end{enumerate}
\end{quote}
\end{itemize}
We define the decoder $D$ as follows. On input $(\sk(u),\sk(v))$, the decoder simulates the
communication tree $T$ on $(u,v)$ using the values $m(u),m(v)$ for each inner node. We therefore
obtain
\[
  \Pru{\sk}{ D(\sk(u),\sk(v)) = E(u,v) }
  = \Pru{T \sim \cP(G)}{ T(u,v) = E(u,v) } \geq 2/3 \,.
\]
We now prove \cref{prop:cctolabeling-item-2}.
From the first argument above, it is clear that for any communication problem $\cM$, if
$\RL(\mathfrak{F}(\cM)) = O(1)$ then $\CC(\cM) = O(1)$. In the other direction, assume that
$\CC(\cM) \leq d$ for some constant $d$, and consider $\RL(\mathfrak{F}(\cM))$. Then it holds for
any $G \in \mathfrak{F}(\cM)$ that there exists $M \in \cM$ such that the adjacency matrix $A$ of
$G$ is a submatrix of $M$. Then $\CC(A) \leq \CC(M) \leq d$. We may then construct adjacency
sketches by the scheme above, so we conclude $\RL(\mathfrak{F}(\cM)) = O(1)$.
\end{proof}

Note that this equivalence is special to constant-cost communication. The transformation of a
communication problem to a hereditary graph class may have unintuitive consequences for problems
with non-constant complexity, as in the next example.

\begin{example}
\label{example:hamming distance}
For a function $k(d)$, the \textsc{$k(d)$-Hamming Distance} problem $\mathrm{HD}^{k(d)}$ requires
Alice and Bob to decide whether the Hamming distance between their inputs $x,y \in \zo^d$ is at most
$k(d)$. Note that this is a family of $n \times n$ matrices with $n = 2^d$. The randomized
communication complexity of $\mathrm{HD}^{k(d)}$ is $\Theta(k(d) \log k(d))$ when $k(d) = o(\sqrt d)$ \cite{HSZZ06,Sag18}.
For example, if $k = \log \log d$ then the communication complexity is $\Theta( (\log\log\log n)
\cdot(\log\log\log\log n))$.

For any non-constant $k(d)$, no matter how small, the hereditary graph class $\cH =
\mathfrak{F}(\mathrm{HD}^{k(d)})$ has a lower bound of $\CC(\Adj_\cH) = \Theta(\log n)$ for
$n$-vertex graphs, meaning that no protocol does better than simply sending the entire input to the
other player. This is because, when we take the hereditary closure (essentially, we include every
submatrix of a \textsc{$k(d)$-Hamming Distance} problem for arbitrarily large $d$), we get the
family of \emph{all} bipartite graphs, as follows.

Suppose that $k(d)$ is increasing in $d$. Fix any $n$ and let $d$ be large enough that $d \geq 2n$
and $k(d) \geq n-1$.  Let $G = (X,Y,E)$ be any bipartite graph with $n = |X|$, and identify $X$ with
$[n]$. To each $x \in X$, assign the string $e_x \in \zo^{2n}$ which is 0 everywhere except on
coordinate $x \in [n]$. Now to each $y \in Y$, let $S \subseteq [n]$ be the set of its neighbors in
$G$. Assign the string $a_y \in \zo^{2n}$ which is 0 everywhere except on the coordinates in $S$, as
well as on the last $n-|S|$ coordinates. Observe that the Hamming distance between $e_x$ and $a_y$
is $(n-|S|)+(|S|-1) = n-1$ if $x$ is adjacent to $y$, and $(n-|S|) + (|S|+1) = n+1$ if $x$ is not
adjacent to $y$. We may now perform a simple padding operation to extend $e_x$ to $e'_x \in \zo^d$
and $a_y$ to $a'_y \in \zo^d$, such that the Hamming distance between $e'_x, a'_y$ is $k(d)$ if $x$
is adjacent to $y$ and $k(d) + 2$ if $x$ is not adjacent to $y$. So we have constructed a submatrix
of the \textsc{$k(d)$-Hamming Distance} matrix which is equal to the adjacency matrix of $G$.
\end{example}

\subsubsection{Chain Number \& Stability}

We define the notion of \emph{chain number}, which will be used in several of our proofs, and use it
to formally define \emph{stability}.

\begin{definition}[Chain Number \& Stability]
\label{def:chain number}
For a graph $G$, the \emph{chain number} $\ch(G)$ is the maximum number $k$ for which there exist
disjoint sets of vertices $\{a_1, \dotsc, a_k\}, \{b_1, \dotsc, b_k\} \subseteq V(G)$ such that
$(a_i,b_j) \in E(G)$ if and only if $i \leq j$. For a graph class $\cF$, we write $\ch(\cF) =
\max_{G \in \cF} \ch(G)$.  If $\ch(\cF) = \infty$, then $\cF$ has \emph{unbounded chain
number}, otherwise it has \emph{bounded chain number}. If $\cF$ has bounded chain number, we also
call it \emph{stable}.
\end{definition}

The name \emph{stable} for these classes is as in \cite{CS18,NMP+21} (they are also called
\emph{graph-theoretically stable} in \cite{GPT21}). These classes have many interesting properties,
including stronger versions of Szemer{\'e}di's Regularity Lemma \cite{MS14} and the Erd{\H
o}s-Hajnal property \cite{CS18} (it is also conjectured in \cite{HHH23dimfree} that graphs which
admit constant-cost protocols for adjacency satisfy the \emph{strong} Erd{\H o}s-Hajnal property),
and they play a central role in algorithmic graph theory \cite{GPT21}.  We observe that stability is
also essential for understanding the IGQ and randomized communication complexity. More specifically,
for a hereditary graph class to have a constant-size PUG, it is necessary for it to be stable. This
follows simply from known lower bounds on the \textsc{Greater-Than} communication problem (see
\cref{section:greater-than-lb}).

\begin{restatable}{proposition}{lowerbound}
\label{prop:lower bound}\RestateRemark
If a hereditary graph class $\cF$ is not stable, then $\RL(\cF) = \Omega(\log n)$.
\end{restatable}

We conclude this section by noting a useful characterization of stable graph classes via forbidden induced subgraphs 
(\cref{prop:ch-finite-iff-no-chainlike}): a graph class $\cF$ has a bounded chain number
(\ie, $\cF$ is stable) if and only if 
$$
	\cF \subseteq \free(H^{\circ\circ}_p, H^{\bullet\circ}_q, H^{\bullet\bullet}_r) \text{, for some choice of $p,q,r$,}
$$
where $H^{\circ\circ}_p$ is a \emph{half-graph}, $H^{\bullet\circ}_q$ is a \emph{co-half-graph}, and $H^{\bullet\bullet}_r$ is a \emph{threshold graph} (depicted in \cref{fig:half graphs}).

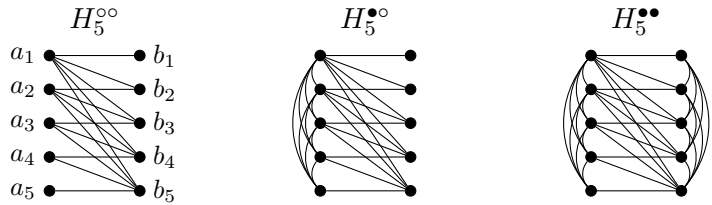
\begin{figure}[thb]
	\centering
	\begin{tikzpicture}[scale=.85]
		% Half graph
		\def\n{5}
		\begin{scope}
			\foreach \i/\ab/\p in {1/a/left,2/b/right} {
				\foreach \x in {1,...,\n} {
					\node[vertex,label={\p:$\ab_\x$}] (\ab\x) at (2*\i-2,-0.75*\x){} ;
				}
			}
			\foreach \i in {1,...,\n} {
				\foreach \j in {1,...,\n} {
					\ifthenelse{\i < \j \OR \i = \j}{\draw (a\i) to (b\j) ;}{}
				}
			}
			\node at (1,0) {$H_5^{\circ\circ}$} ;
		\end{scope}
		% Threshold
		\begin{scope}[shift={(6,0)}]
			\foreach \i/\ab/\p in {1/a/left,2/b/right} {
				\foreach \x in {1,...,\n} {
					\node[vertex] (\ab\x) at (2*\i-2,-0.75*\x){} ;
				}
			}
			\foreach \i in {1,...,\n} {
				\foreach \j in {1,...,\n} {
					\ifthenelse{\i < \j \OR \i = \j}{\draw (a\i) to (b\j) ;}{}
				}
			}
			\foreach \i in {1,...,\n} {
				\foreach \j in {1,...,\n} {
					\ifthenelse{\i<\j}{
						\draw (a\j) to[bend left=40] (a\i) ;
					}{}
				}
			}
			\node at (1,0) {$H_5^{\bullet\circ}$} ;
		\end{scope}
		% ch-chain
		\begin{scope}[shift={(12,0)}]
			\foreach \i/\ab/\p in {1/a/left,2/b/right} {
				\foreach \x in {1,...,\n} {
					\node[vertex] (\ab\x) at (2*\i-2,-0.75*\x){} ;
				}
			}
			\foreach \i in {1,...,\n} {
				\foreach \j in {1,...,\n} {
					\ifthenelse{\i < \j \OR \i = \j}{\draw (a\i) to (b\j) ;}{}
				}
			}
			\foreach \ab/\p in {a/left,b/right} {
				\foreach \i in {1,...,\n} {
					\foreach \j in {1,...,\n} {
						\ifthenelse{\i<\j}{
							\draw (\ab\j) to[bend \p=40] (\ab\i) ;
						}{}
					}
				}
			}
			\node at (1,0) {$H_5^{\bullet\bullet}$} ;
		\end{scope}
	\end{tikzpicture}
	\caption{%
		Examples of the half-graph, co-half-graph, and threshold graphs.
	}
	\label{fig:half graphs}
\end{figure}

\subsection{Reductions}
\label{section:reductions}

We define a general notion of reductions useful for both communication complexity and implicit graph
representations\footnote{This general notion of reductions did not appear in the conference version
of this paper. It is an adaptation of the definitions presented in subsequent work
\cite{HZ23,FHHH24}, which we include here for the sake of clarity, and consistency with subsequent
work.}.
First, we define a notion of reduction between families of matrices. 
Given $k$ matrices $B_1, B_2, \dotsc, B_k \in \zo^{n \times n}$ and a Boolean function $h : \zo^k \to \zo$, we define
\[
  h(B_1, B_2, \dotsc, B_k) \define A \,,
\]
where $A \in \zo^{n \times n}$ is the matrix with $A(i,j) = h(B_1(i,j), B_2(i,j), \dotsc, B_k(i,j))$
for all $i,j \in [n]$.

\begin{definition}[Matrix Reductions]
\label{def:matrix-reductions}
Let $\cA, \cB$ be families of matrices. We say that $\cA$ \emph{matrix-reduces to} $\cB$ if there
exists a constant $k$ such that, for every $A \in \cA$, there
exist $h : \zo^k \to \zo$ and $B_1, B_2, \dotsc, B_k \in \cB$ so that
\[
  A = h(B_1, B_2, \dotsc, B_k) \,.
\]
\end{definition}

\subsubsection{Between communication problems}
\newcommand{\QS}{\mathsf{QS}}
\newcommand{\D}{\mathsf{D}}

Informally, for communication problems $\cA, \cB$, we will say that $\cA$ has a \emph{constant-cost
reduction} to $\cB$ if $\cA$ can be computed by a constant-cost deterministic protocol with oracle
access to $\cB$.  Formally, we define this as follows.

\begin{definition}[Query Set]
\label{def:query-set}
For any set of matrices $\cQ$ we define the \emph{query set}
\[
\QS(\cQ) \define \left\{ M \;\mid\; M \text{ is a blowup of a permutation of a submatrix of some } Q \in \cQ
\right\} \,,
\]
where we say that a matrix $A$ is a \emph{blowup} of $B$ if it is obtained from $B$ by a sequence of row and
column duplications.
\end{definition}

\begin{definition}[Communication with Oracles]
\label{def:oracle-communication}
For any set $\cQ$ of matrices, we define the \emph{deterministic $\cQ$-oracle communication
complexity} of the communication problem $\cA$ as follows. For any $M \in \cA$ with row and columns
indexed by $[n_1], [n_2]$ respectively, a $\cQ$-oracle protocol for $M$ is a tree $T$ whose leaves
$v$ are associated with output values $\ell(v) \in \zo$, and whose inner nodes $u$ are associated
with $n_1 \times n_2$ matrices $Q_u \in \QS(\cQ)$ (note that matrices of every size appear in
$\QS(\cQ)$ since it is closed under taking submatrices and blowups). On input pair $i \in [n_1], j \in [n_2]$,
the protocol initializes a pointer $p$ to the root of $T$, and in every round the protocol proceeds
to
\[
  p \gets \begin{cases}
    \text{left child of $p$} &\text{ if } Q_p(i, j) = 0 \\
    \text{right child of $p$} &\text{ if } Q_p(i, j) = 1 \,,
  \end{cases}
\]
until $p$ is a leaf, at which point it outputs $\ell(p)$, which is required to be equal to $M(i,j)$.
We write $\D^\cQ(M)$ for the minimum cost (\ie, depth) of such a protocol, and we say $n \mapsto
\D^\cQ(\cA_n)$ is the deterministic $\cQ$-oracle communication complexity of $\cA$.
\end{definition}

We may now define \emph{constant-cost reductions} between communication problems.

\begin{definition}[Constant-Cost Reductions]
\label{def:constant-cost-reductions}
A communication problem $\cA$ \emph{reduces} to (or has a \emph{constant-cost reduction} to) $\cB$
if $\D^\cB(\cA) = O(1)$.
\end{definition}

\begin{remark}
This notion of reduction differs from the standard notion of an oracle reduction in communication
complexity (\eg \cite{BFS86}) because it allows the oracle matrices to be submatrices of
\emph{arbitrarily large} instances of $\cB$. Standard oracle reductions require the oracle matrix to
be a submatrix of some $B \in \cB$ with at most $N$ rows and columns, where $\log N = \poly\log(\log
n)$ to preserve communication complexities of the form $\poly\log(\log n)$ (where $\log n$ is the
number of bits required to encode the inputs $i,j \in [n]$). Since we aim to preserve
\emph{constant-cost} communication, this requirement is unnecessary.
\end{remark}

\begin{proposition}
Suppose communication problem $\cA$ has a constant-cost reduction to problem $\cB$, and $\CC(\cB) = O(1)$.
Then $\CC(\cA) = O(1)$.
\end{proposition}
\begin{proof}
The randomized communication cost of a matrix $M$ is not changed by permutations or
by row and column duplications, so $\CC(\QS(\cB)) = O(1)$ also. Let $A \in \cA$; we must show $\CC(A) =
O(1)$. Since $\cA$ reduces to $\cB$, there is a
$\cB$-oracle protocol for $A$, \ie,
a communication tree $T$ of constant depth $d = O(1)$, with
each vertex $v$ assigned a matrix $Q_v \in \QS(\cB)$. We obtain a randomized communication protocol
with cost $O(d \log d)$ by simulating the tree $T$ on inputs $x,y$; namely, for each node $v$ of $T$ we simulate
the randomized protocol for $Q_v$ with error $1/3d$, which can be done with communication cost
$O(\log d)$ using standard error boosting. Since at most $d$ oracle queries are simulated,
the probability that the protocol has an error on any of them is at most $d \cdot 1/3d = 1/3$.
\end{proof}

Constant-cost communication reductions are equivalent to matrix reductions as follows.

\begin{proposition}\label{prop:cc-mat-red}
A communication problem $\cA$ has a constant-cost communication reduction to $\cB$ if and only if
$\cA$ matrix-reduces to $\QS(\cB)$.
\end{proposition}
\begin{proof}
First, assume that $\cA$ matrix-reduces to $\QS(\cB)$, \ie, 
there exists a constant $k$ such that, for each $A \in \cA$, 
here exist a function $h : \zo^k \to \zo$ and $Q_1, \dotsc, Q_k \in
\QS(\cB)$ such that
\[
    A = h(Q_1, Q_2, \dotsc, Q_k) \,.
\]
Then for any $A \in \cA$, we obtain $\D^\cB(A) \leq k$ by using the following communication
protocol: on
inputs $x,y$, the players query $Q_1(x,y), Q_2(x,y), \dotsc, Q_k(x,y)$ and then output the value $h(Q_1(x,y), \dotsc, Q_k(x,y))$.

Now, assume that $\D^\cB(\cA) = O(1)$. Then for any $A \in \cA$, there is an oracle communication
tree $T$ with constant depth $d = O(1)$ and each inner node $v$ assigned a matrix $Q_v \in
\QS(\cB)$. Let $v_1, \dotsc, v_t$ be inner nodes of $T$, where $t \leq 2^d$. Then we may write
\[
  A = h(Q_{v_1}, Q_{v_2}, \dotsc, Q_{v_t}) \,,
\]
where $h : \zo^t \to \zo$ is the function which simulates $T$ given the answers to every query.
\end{proof}

\subsubsection{Between graph classes}
\label{section:graph-class-reductions}

Let $\cF, \cG$ be hereditary classes of graphs. We define a reduction between $\cF, \cG$ with the
purpose of preserving the existence of implicit representations.

Let $\Adj_\cF$ denote the set of adjacency matrices of graphs in $\cF$.  We say $\cF$ reduces to
$\cG$ if $\Adj_\cF$ matrix-reduces to $\QS(\Adj_\cG)$, where $\QS(\cdot)$ denotes the query set
defined above. In graph-theoretic terms, $\QS(\Adj_\cG)$ is the set of bipartite adjacency matrices
obtained as follows: for any graph $G \in \cG$, let $X, Y \subset V(G)$ be any subsets of vertices
and $H = (X,Y,E_H)$ be the bipartite graph where $(x,y) \in X \times Y$ is an edge of $H$ if and only if 
$\{x,y\}$ is an edge of $G$. Then take any $H' = (X',Y',E'_H)$ obtained from $H$ by duplicating vertices. $\QS(\Adj_\cG)$
is the set of bipartite adjacency matrices of any graph obtained in this way.

\begin{proposition}
Suppose that $\cF$ reduces to $\cG$ and $\cG$ admits an adjacency labeling scheme of size $s(n)$.
Then $\cF$ admits an adjacency labeling scheme of size $O(s(n))$.
\end{proposition}
\begin{proof}
First, observe that, if $\cG$ admits an adjacency labeling scheme of size $s(n)$, then the family
$\cH$ of bipartite graphs whose adjacency matrices are in $\QS(\Adj_\cG)$ admits an adjacency labeling
scheme of size $s(n) + 1$: to every vertex $v \in X \cup Y$ of a graph $H = (X, Y, E_H)$ constructed
from $G \in \cG$ as in the paragraph above (without vertex duplication), simply take the original label
for the scheme of $G$ and append one bit to indicate whether $v \in X$ or $v \in Y$. For the graph
$H' = (X',Y',E'_H)$ obtained by duplicating vertices of $H$, simply copy the corresponding labels. It is
straightforward to verify that such labels contain sufficient information to correctly deduce the adjacency 
relation between vertices in graphs from $\cH$.

Now, there exists a constant $k$ such that, for any graph $G \in \cF$, there exist bipartite graphs
$B_1, \dotsc, B_k \in \cH$ with bipartite adjacency matrices $\Adj_{B_i} \in \QS(\Adj_\cG)$ and a
function $h : \zo^k \to \zo$, such that
\[
  \Adj_G = h(\Adj_{B_1}, \Adj_{B_2}, \dotsc, \Adj_{B_k}) \,.
\]
If $\cF$ is a class of bipartite graphs, then for $G = (X,Y,E)$, $\Adj_G$ has rows indexed by $X$ and columns indexed by $Y$, and we may assume that each $\Adj_{B_i}$ also
has rows indexed by $X$ and columns indexed by $Y$. We label each $v \in X \cup Y$, with the
concatenation of the encoding of $h$ and the labels of $v$ for each $B_i$ in the labeling scheme for
$\cH$.  Then the decoder can decide if $u,w \in X \cup Y$ are in different parts of the graph, and
if so, decide adjacency of $u$ and $w$ in each $B_i$ and compute $h$ on these adjacency values. The
size of the labels is at most $2^k + k(s(n)+1) = O(s(n))$ since $k$ is constant.

Now, if $\cF$ is not a class of bipartite graphs, then for $G = (V,E)$, $\Adj_G$ is symmetric and has
rows and columns indexed by $V$. Each $B_i$ is a bipartite graph with the parts $V$ and
$V'$, a copy of $V$. For each $v \in V$, we assign a label by concatenating the encoding of $h$ with two labels $a(v), b(v)$, where $a(v)$ is the concatenation of the labels of $v \in V$ in each $B_i$, and $b(v)$ is the concatenation of the labels of $v \in V'$ in each $B_i$. 
Given two labels for $u,v \in V$, the decoder can determine adjacency in $G$
by means of the function $h$ and the decoder for the labeling scheme of $\cH$ on each $B_i$, using the relevant parts of $a(v)$ and $b(u)$. The size of the labels is at most $2^k + 2k(s(n)+1)= O(s(n))$.
\end{proof}

\begin{remark}
It may seem more natural to define reductions by saying $\cF$ reduces to $\cG$ if $\Adj_\cF$
matrix-reduces to $\Adj_\cG$ instead of $\QS(\Adj_\cG)$. This would preserve more of the graph
structure through the reduction, since $\QS(\Adj_\cG)$ allows taking semi-induced bipartite subgraphs
(instead of induced subgraphs) and allows row and column duplication (\ie, duplicating vertices). But
here we are interested only in preserving existence of adjacency labeling, and the reduction defined
here is more permissive.
\end{remark}

\subsection{Equality-Based Labeling Schemes}
\label{section:equality-based-labeling}

We define a certain type of adjacency labeling scheme called an \emph{equality-based labeling
scheme}. A number of our results will be proved by constructing equality-based labeling schemes, so we
will introduce some special notation for these.  For simplicity of notation, we write
\[
  \Eq(a,b) = \ind{a=b} \,.
\]

\begin{restatable}{definition}{eqlabeling}(Equality-based Labeling Scheme).
\label{def:equality labeling}\RestateRemark
Let $\cF$ be a class of graphs. An \emph{$(s,k)$-equality-based} labeling scheme for $\cF$ is a
labeling scheme defined as follows. For every $G \in \cF$ with vertex set $[n]$ and every $x \in
[n]$, the label $\ell(x)$ consists of the following:
\begin{enumerate}
\item A \emph{prefix} $p(x) \in \zo^s$. If $s = 0$ we write $p(x) = \bot$.
\item A sequence of $k$ \emph{equality codes} $q_1(x), \dotsc, q_k(x) \in \bN$.
\end{enumerate}
The decoder must be of the following form. There is a set of functions $D_{p_1,p_2} : \zo^{k \times
k} \to \zo$ defined for each $p_1,p_2 \in \zo^s$ such that, for every $x,y \in [n]$, it holds that
$(x,y) \in E(G)$ if and only if $D_{p(x),p(y)}(Q_{x,y}) = 1$, where $Q_{x,y} \in \zo^{k \times k}$ 
is the matrix with entries $Q_{x,y}(i,j) = \ind{q_i(x)=q_j(y)}$. If $s = 0$ we simply write
$D(Q_{x,y})$.

We say that $\cF$ admits a \emph{constant-size equality-based} labeling scheme if there exist
constants $s,k$ such that $\cF$ admits an $(s,k)$-equality-based labeling scheme.
\end{restatable}

\begin{remark}\label{rm:conv-form}
We will often use the following notation. A label for $x$ will be written as a constant-size tree of
tuples of the form
\[
  (p_1(x), \dotsc, p_r(x) \eqLabelSep q_1(x), \dotsc, q_t(x) ) \,,
\]
where the symbols $p_i(x)$ belong to the prefix, while the symbols $q_i(x)$ are equality codes. When
$r,t$ are constants, 
%and the label consists of a constant number of tuples, 
it is straightforward to
put such a label into the form required by \cref{def:equality labeling}.
\end{remark}

\begin{proposition}
\label{prop:eq-label-to-sk}
Suppose a class of graphs $\cF$ admits an $(s,k)$-equality based labeling scheme. Then
$\RL(\cF) = O(s + k \log k)$.
\end{proposition}
\begin{proof}
Let $G \in \cF$ and let $n$ be the number of vertices. Let $T \subseteq \bN$ be the set of equality
codes $T = \{ q_i(x) : i \in [k], x \in V(G) \}$. We assign the sketches to $V(G)$ as
follows. For each $t \in T$ we assign a uniformly random $b(t) \in [3k^2]$ and observe that $b(t)$
may be encoded in $O(\log k)$ bits. To each vertex $x$ we assign the sketch
containing the $s$ bits of the prefix $p(x)$, followed by $b(q_1(x)), \dotsc, b(q_q(x))$, which
requires $O(s + k \log k)$ bits.

Consider any two vertices $x,y$. The decoder, upon receiving the sketches for $x$ and $y$,
constructs the matrix $Q'_{x,y}$ defined as $Q'_{x,y}(i,j) = \Eq(b(q_i(x)), b(q_j(y)))$, and then
applies the decoder $D_{p(x),p(y)}(Q'_{x,y})$ from the equality-based labeling scheme.  Observe that with probability at least $1 - \frac{1}{3k^2}$ it holds that
$\Eq(b(q_i(x)), b(q_j(y)))  = \Eq(q_i(x), q_j(y))$.
Then by the union bound, we have $Q'_{x,y} = Q_{x,y}$ with probability at least $2/3$, in which
case the decoder outputs the correct value.
\end{proof}

\begin{example}
\label{example:forests}
Many standard examples of adjacency labeling schemes are in fact equality-based labeling schemes.
For example, we may write the original adjacency labeling scheme of \cite{KNR92} for forests as an
equality-based labeling scheme: the label for any vertex $x \in [n]$ is
\[
  ( - \mid x, p(x) ),
\]
where the first equality code is the vertex itself, and the second equality code is its parent $p(x)$. Given labels $( - \mid x, p(x))$ and $( - \mid y, p(y))$,
the decoder outputs $(x = p(y)) \vee (y = p(x))$.
\end{example}

Let us now define equivalence graphs.
\begin{definition}
	A graph $G$ is an \emph{equivalence graph} if it is a disjoint union of complete graphs.  A colored bipartite graph $G =
	(X,Y,E)$ is a \emph{bipartite equivalence graph} if it is a colored disjoint union of bicliques, \ie, 
	if there are partitions $X = X_1 \cup \dotsm \cup X_m$, $Y = Y_1 \cup \dotsm \cup Y_m$ such that each
	$G[X_i,Y_i]$ is a biclique and each $G[X_i,Y_j]$ is a co-biclique when $i \neq j$.

  Alternatively, the equivalence graphs are exactly the $P_3$-free graphs, and the bipartite
  equivalence graphs are exactly the $P_4$-free bipartite graphs. 
\end{definition}

Equality-based labeling schemes are related to reductions by the following proposition. For the sake
of completeness, we state the simple proof in \cref{section:eq-labeling-proof}.

\begin{restatable}{proposition}{lemmaeqlabeling}
\label{lemma:eq labeling}\RestateRemark
The following are equivalent for a hereditary graph class $\cF$:
\begin{enumerate}
\item $\cF$ admits a constant-size equality-based labeling scheme;
\item The graph class $\cF$ reduces to the class of equivalence graphs;
\item The communication problem $\Adj_\cF$ reduces to \textsc{Equality}.
\end{enumerate}
Therefore, if $\cF$ admits a constant-size equality-based labeling scheme, it admits a constant-size
adjacency sketch (and hence a constant-size PUG).
\end{restatable}

\subsection{Basic Adjacency Sketches}
\label{sec:bounded-arboricity}

Adjacency sketches for graph classes of bounded arboricity will be basic building blocks for some
of our results.
\begin{definition}
A graph $G = (V,E)$ has \emph{arboricity} $\alpha$ if its edges can be partitioned into at most
$\alpha$ forests.
\end{definition}
The classic adjacency labeling scheme of \cite{KNR92} for graph classes of bounded arboricity may be
interpreted as an $(0,\alpha)$-equality-based labeling scheme (as in \cref{example:forests}).
Using \cref{prop:eq-label-to-sk}, this gives an adjacency sketch of size $O(\alpha \log \alpha)$,
which was also stated in \cite{Har20}. We show that this can be improved as follows:

\begin{lemma}
\label{lemma:arboricity}
For any $\alpha \in \bN$, let $\cA$ be the class of graphs with arboricity at most $\alpha$. Then
$\cA$ admits a constant-size equality-based adjacency labeling scheme. $\cA$ also admits an
adjacency sketch of size $O(\alpha)$.
\end{lemma}
\begin{proof}
For any graph $G \in \cA_n$ with vertex set $[n]$, partition the edges of $G$ into forests $F_1,
\dotsc, F_\alpha$ and to each tree in each forest, identify some arbitrary vertex as the root. For
every vertex $x$, assign equality codes $q_1(x) = x$ and for $i \in [\alpha]$ set $q_{i+1}(x)$ to be
the parent of $x$ in forest $F_i$; if $x$ is the root assign $q_{i+1}(x) = 0$.  For vertices $x,y$,
the decoder outputs
\[
  \left(\bigvee_{j = 2}^\alpha \ind{ q_1(x), q_j(y) } \right) \vee
  \left(\bigvee_{j = 2}^\alpha \ind{ q_1(y), q_j(x) } \right) \,.
\]
This is 1 if and only if $y$ is the parent of $x$ or $x$ is the parent of $y$ in some forest $F_i$.

One can apply \cref{prop:eq-label-to-sk} to obtain an $O(\alpha \log \alpha)$ adjacency sketch. We can
improve this using a Bloom filter, since the output is simply a disjunction of equality checks. To
each $i \in [n]$, assign a uniformly random number $r(i) \sim [6\alpha]$, and to each vertex $x$
assign the sketch $(r(x), b(x))$ where $b(x) \in \zo^{6\alpha}$ satisfies $b(x)_i = 1$ if and
only if $r(q_j(x)) = i$ for some $j \in \{2, \dotsc, \alpha+1\}$. On input $(r(x), b(x))$ and
$(r(y), b(y))$, the decoder outputs 1 if and only if $b(x)_{r(y)} = 1$ or $b(y)_{r(x)} = 1$. If $y$
is a parent of $x$ in any of the $\alpha$ forests, then $y=q_j(x)$ for some $j$, so $b(x)_{r(y)} =
b(x)_{r(q_j(x))} = 1$ and the decoder will output 1 with probability 1. Similarly, if $x$ is a
parent of $y$ in any of the $\alpha$ forests, the decoder will output 1 with probability 1. The
decoder fails only when $x,y$ are not adjacent and $r(x) = r(q_j(y))$ or $r(y) = r(q_j(x))$ for some
$j$. By the union bound, this occurs with probability at most $2\alpha \cdot \frac{1}{6\alpha} =
1/3$, as desired. The size of the sketches is $O(\alpha + \log \alpha) = O(\alpha)$.
\end{proof}

\begin{remark}
\label{remark:treewidth}
Graph classes of bounded arboricity include many commonly-studied hereditary graph classes, \eg, classes of graphs of bounded degree, classes of bounded treewidth, and proper minor-closed classes \cite{Mad67}.
\end{remark}

\section{Question I. New Examples of Constant-Cost Communication}
\label{sec:new_examples}

This first main section of results is dedicated to finding new examples of problems with
constant-cost randomized communication, to better understand the variety of problems in this class;
this also leads to some new methods for constructing adjacency labeling schemes.\\

\noindent
\textbf{\cref{section:twin-width}.} We study the problem of computing path distance $\dist(x,y) \leq k$
in graphs, where the players are given vertices $x,y$ of a shared graph.  It has long been known
\cite{Yao03,HSZZ06} that the \textsc{$k$-Hamming Distance} problem has a constant-cost randomized
communication protocol whenever $k$ is constant, which we may think of as computing path distance
$\dist(x,y) \leq k$ for vertices $x,y$ in the hypercube graph, with cost independent of the number
of vertices. One of the questions motivating this paper is, for which other graphs can we compute
path distance $\dist(x,y) \leq k$ in randomized communication cost independent of the number of
vertices? Equivalently, for which hereditary classes of graphs can we construct constant-size
\emph{sketches} for deciding small distances? We generalize this problem to constructing
constant-size sketches for deciding certain first-order formulae $\phi(x,y)$ over vertices $x,y$.\\

\noindent
\textbf{\cref{section:cartesian product}.} We show that constant-size adjacency sketches are
preserved by the Cartesian product operation on graphs. This also gives a new method of constructing
adjacency labeling schemes for Cartesian products.

%%%%%%%%%%%%%%%%%%%%%%%%%%%%%%%%%%%%%%%%%%%
\subsection{Computing Small Distances and First-Order Formulae}
\label{section:twin-width}
%%%%%%%%%%%%%%%%%%%%%%%%%%%%%%%%%%%%%%%%%%%

Graph width parameters like treewidth \cite{robertson1986graph} and clique-width
\cite{courcelle1993handle} are a central tool in structural graph theory. Often, graph classes where
a certain width parameter is bounded possess favorable structural, algorithmic, or combinatorial
properties.  The \emph{twin-width} parameter, introduced recently in \cite{BKTW20}, generalizes
treewidth and clique-width and has attracted a lot of recent attention
\cite{GPT21,schidler2021sat,ahn2021bounds,balaban2021twin}.  Graph classes of
bounded twin-width admit size $O(\log n)$ adjacency labeling schemes \cite{BGK+21a}, making them a
natural choice for studying \cref{question:main}. Building upon recent structural results on stable
classes of bounded twin-width \cite{GPT21}, we prove:

\begin{restatable}{theorem}{thmtwinwidth}
\label{thm:twin width}\RestateRemark
Let $\cF$ be a hereditary graph class of bounded twin-width. Then $\cF$ admits a constant-size PUG if and only if $\cF$ is stable.
\end{restatable}

We also consider a type of sketches that generalizes adjacency and distance-$k$ sketches.  For any
\emph{first-order formula} $\phi(x,y)$ on vertex pairs (see \cref{section:twin-width}), we consider
sketches that allow to decide $\phi(x,y)$ instead of adjacency. An example is the $\dist(x,y) \leq
k$ formula:
\[
  \delta_k(x,y) \define (\exists v_1, \dotsc, v_{k-1} : (E(x,v_1) \vee x=v_1) \wedge
(E(v_1,v_2) \vee v_1=v_2) \wedge
\dotsm \wedge (E(v_{k-1},y) \vee v_{k-1}=y)) \,.
\]
Using results on \emph{first-order transductions} (a graph transformation that arises in model
theory) and their relation to stability and twin-width \cite{BKTW20,NMP+21}, we
obtain the following corollary of \cref{thm:twin width}.

\begin{corollary}
\label{cor:twin-width-fo}
Let $\cF$ be a stable class of bounded twin-width and let $\phi(x,y)$ be a first-order formula.
Then $\cF$ admits a constant-size sketch for deciding $\phi$.
\end{corollary}

This implies
%\cref{thm:intro-twin-width}
\cref{thm:intro-planar} and its generalization stated in the introduction; it gives us constant-size
distance-$k$ sketches for stable classes of bounded twin-width, and, in particular, answers a
question of \cite{Har20}, who asked about distance-$k$ sketches for planar graphs (which are stable,
because they have bounded arboricity, and are of bounded twin-width \cite{BKTW20}).  Our results
have subsequently been extended and improved in \cite{EHK22}.

In \cref{sec:tw-prelim} we provide necessary definitions and notations.
To prove the theorem, we will first reduce the problem to bipartite graphs in
\cref{sec:to-bip-graphs}, and then show in \cref{sec:bip-graphs} how to construct a constant-size
equality-based labeling scheme for any stable class of bipartite graphs of bounded twin-width.
In \cref{section:distance sketching} we will generalize the above theorem to first-order labeling
schemes to show \cref{cor:twin-width-fo}.

%%%%%%%%%%%%%%%%%%%%%%%%%%%%%%%%%%%%%%%%%%%%%%
\subsubsection{Preliminaries}\label{sec:tw-prelim}
%%%%%%%%%%%%%%%%%%%%%%%%%%%%%%%%%%%%%%%%%%%%%%

Let $G=(V,E)$ be a graph. A pair of disjoint vertex sets $X,Y \subseteq V$ is \emph{pure} if
either $\forall x \in X, y \in Y$ it holds that $(x,y) \in E$,
or $\forall x \in X, y \in Y$ it holds that $(x,y) \not\in E$.

\begin{definition}[Twin-Width]
\label{def:twin width}
An \emph{uncontraction sequence of width $d$} of a graph $G=(V,E)$ is a sequence
$\cP_1, \dotsc, \cP_m$ of partitions of $V$ such that:
\begin{itemize}
\item $\cP_1 = \{ V \}$;
\item $\cP_m$ is a partition into singletons;
\item For $i = 1, \dotsc, m-1$, $\cP_{i+1}$ is obtained from $\cP_i$ by splitting exactly one of the
parts into two;
\item For every part $U \in \cP_i$ there are at most $d$ parts $W \in \cP_i$ with $W \neq U$ such
that $(U,W)$ is not pure.
\end{itemize}
The \emph{twin-width} $\tw(G)$ of $G$ is the minimum $d$ such that there is an uncontraction sequence
of width $d$ of $G$.
\end{definition}

\noindent
The following fact, that we need for some of our proofs, uses the notion of first-order (FO) transduction.
We omit the formal definition of FO transductions and refer the interested reader to \eg \cite{GPT21}.
Informally, given a first-order formula $\phi(x,y)$ and a graph $G$, a \emph{first-order (FO) $\phi$-transduction} of $G$ is
a transformation of $G$ to 
a graph that is obtained from $G$ by first taking a constant number of vertex-disjoint copies of $G$, 
then coloring the vertices of the new graph by a constant number of colors, 
then using $\phi$ as the new adjacency relation, 
and finally taking an induced subgraph.
An FO $\phi$-transduction of a graph class $\cF$ is a class of $\phi$-transductions of graphs in $\cF$.
A graph class $\cG$ is an FO-transduction of $\cF$ if there exists an FO formula $\phi$ so that
$\cG$ is an $\phi$-transduction of $\cF$.

\begin{theorem}[\cite{BKTW20,NMP+21}]\label{th:stable-tw-FO}
	Let $\cF$ be a stable class of bounded twin-width. Then any FO transduction of $\cF$ is
	also a stable class of bounded twin-width.
\end{theorem}

%%%%%%%%%%%%%%%%%%%%%%%%%%%%%%%%%%%%%%%%%%%%%%
\subsubsection{From General Graphs to Bipartite Graphs}\label{sec:to-bip-graphs}
%%%%%%%%%%%%%%%%%%%%%%%%%%%%%%%%%%%%%%%%%%%%%%

We begin by defining a natural mapping $\bip(G)$ that transforms a graph $G$ into a bipartite graph.

\newcommand{\cpy}{\triangleleft}
\begin{definition}
For any graph $G = (V,E)$ we define the colored bipartite graph $\bip(G) = (V,V^\cpy,E^\cpy)$, where
$V^\cpy$ is a copy of $V$, as follows. For each $v \in V$ let $\cpy v \in V^\cpy$ denote its copy in
$V^\cpy$.  Then for each $x,y \in V$ we have $E(x,y) = E(y,x) = E^\cpy(x, \cpy y) = E^\cpy(y, \cpy x)$. For
any set $\cF$ of graphs, we define
\[
  \bip(\cF) \define \{ \bip(G) : G \in \cF \} \,.
\]
\end{definition}

We require the following two simple properties of this transformation.
\begin{proposition}\label{prop:Bip-transduction}
	Let $\cF$ be a class of graphs. Then the class $\cl(\bip(\cF))$ is an FO transduction of $\cF$.
\end{proposition}
\begin{proof}[Proof sketch]
	Let $G=(V,E)$ be a graph in $\cF$. To obtain the graph $\bip(G)$, we take a disjoint union $G_1 \cup G_2$, where
	$G_1=(V_1,E_1)$ and $G_2 = (V_2, E_2)$ are two copies of $G$,
	and transform this graph using an FO formula $\phi(x,y)$ that does not hold for any pair of vertices 
	that belong to the same copy of $G$, and holds for any pair of vertices that are in different copies and whose
	preimages are adjacent in $G$:
	$$
		\phi(x,y) \;=\; \bigl[(x \in V_1 \wedge y \in V_2) \vee (y \in V_1 \wedge x \in V_2)\bigr] 
		\wedge \exists y' : M(y,y') \wedge (E_1(x,y') \vee E_2(x,y')),
	$$
	where $M$ is a relation that is true exactly for copies of the same vertex.
	For any induced subgraph of $\bip(G)$, we in addition take an induced subgraph.
\end{proof}

\begin{proposition}
\label{prop:bip-labeling}
Let $\cF$ be a hereditary graph class. 
\begin{enumerate}
	\item[(1)] If $\cl(\bip(\cF))$ has an adjacency labeling scheme of
	size $s(n)$, then $\cF$ has an adjacency labeling scheme of size $O(s(2n))$. 
	\item[(2)] If $\cl(\bip(\cF))$ has a constant-size equality-based labeling schemes, then so does $\cF$.
\end{enumerate}
\end{proposition}
\begin{proof}[Proof sketch]
For $G \in \cF$, we may assign labels to $v \in V(G)$ by concatenating the labels for $v, v^\cpy$
in $\bip(G)$, which are of size at most $s(2n)$. Given the labels for $(u, u^\cpy), (v, v^\cpy)$,
the decoder can determine adjacency from the labels of $u, v^\cpy$. 
%Since $s(\cdot)$ is at most linear, 
Thus, the total label size is at most $2s(2n) = O(s(2n))$.
A similar argument works for the second part of the statement.
\end{proof}

We now see that to prove \cref{thm:twin width}, it is enough to establish its special case for the classes of
bipartite graphs.  Assume \cref{thm:twin width} holds for the bipartite graph classes and let $\cF$ be
an arbitrary hereditary stable class of graphs of bounded twin-width.  Then, by
\cref{prop:Bip-transduction} and \cref{th:stable-tw-FO}, the class $\cl(\bip(\cF))$ is also stable
and has bounded twin-width.  Therefore, by \cref{prop:bip-labeling}, any constant-size equality-based
labeling scheme for $\cl(\bip(\cF))$ can be turned into a constant-size equality-based labeling
scheme for $\cF$.

%%%%%%%%%%%%%%%%%%%%%%%%%%%%%%%%%%%%%%%%%%%%%%
\subsubsection{Bipartite Graphs of Bounded Twin-width}\label{sec:bip-graphs}
%%%%%%%%%%%%%%%%%%%%%%%%%%%%%%%%%%%%%%%%%%%%%%

An \emph{ordered graph} is a graph equipped with a total order on its vertices. We will denote
ordered graphs as $(V,E,\leq)$ and bipartite ordered graphs  as $ (X,Y,E,\leq)$, where $\leq$ 
is a total order on $V$ and on $X \cup Y$ respectively.
A \emph{star forest} is a graph whose every component is a star.
Let $(X,\leq)$ be a totally-ordered set. A subset $S \subseteq X$ is
\emph{convex} if for every $x,y,z\in X$ with $x \leq y \leq z$, such that $x,z \in S$, it holds also that $y \in S$.

\begin{definition}[Division]
A \emph{division} of an ordered bipartite graph $G = (X,Y,E,\leq)$ is a partition $\cD$ of $X \cup
Y$ such that each part $P \in \cD$ is convex and either $P \subseteq X$ or $P \subseteq Y$. We will
write $\cD^X = \{ P \in \cD : P \subseteq X \}, \cD^Y = \{ P \in \cD : P \subseteq Y \}$.
\end{definition}

\begin{definition}[Quotient Graph]
For any ordered bipartite graph $G = (X,Y,E, \leq)$ and any division $\cD$, the \emph{quotient graph}
$G/\cD$ is the bipartite graph $(\cD^X, \cD^Y, \cE)$ where $A \in \cD^X, B \in \cD^Y$ are
adjacent if and only if there exist $x \in A, y \in B$ such that $(x,y)$ are adjacent in $G$.
\end{definition}

\newcommand{\ctw}{\mathsf{ctww}}
\begin{definition}[Convex Twin-Width]
A \emph{convex uncontraction sequence of width $d$} of an ordered bipartite graph $G = (X,Y,E,\leq)$
is a sequence $\cP_1, \dotsc, \cP_m$ of divisions of $X \cup Y$ such that:
\begin{itemize}
\item $\cP_1 = \{ X, Y \}$;
\item $\cP_m$ is a division into singletons;
\item For $i = 1, \dotsc, m-1$, the division $\cP_{i+1}$ is obtained from $\cP_i$ by splitting
exactly one of the parts into two;
\item For every part $U \in \cP_i^X$, there are at most $d$ parts $W \in \cP_i^Y$ such that $(U,W)$
is impure. For every part $W \in \cP_i^Y$, there are at most $d$ parts $U \in \cP_i^X$ such that
$(U,W)$ is impure.
\end{itemize}
The \emph{convex twin-width} $\ctw(G)$ is the minimum $d$ such that there is a convex uncontraction
sequence of width $d$ of $G$.
\end{definition}

\begin{lemma}[\cite{GPT21}, Lemmas 3.14 \& 3.15]
\label{lemma:convex twin width}
For any ordered bipartite graph $G, \tw(G) \leq \ctw(G)$. For any bipartite graph $G = (X,Y,E)$,
there is a total order $\leq$ on $X \cup Y$ such that $\ctw((X,Y,E,{\leq})) \leq \tw(G)+1$.
\end{lemma}

\newcommand{\qch}{\mathsf{qch}}
\begin{definition}[Quasi-Chain Number]
Let $G = (X,Y,E)$ be a bipartite graph. The \emph{quasi-chain number} $\qch(G)$ of $G$ is the largest $k$ for which
there exist two sequences $x_1, \dotsc, x_k \in X$ and $y_1, \dotsc, y_k \in Y$ of not necessarily distinct vertices such that for each $i\in [k]$, one of the following holds:
\begin{enumerate}
\item $x_i$ is adjacent to all of $y_1, \dotsc, y_{i-1}$ and $y_i$ is non-adjacent to all of $x_1,
\dotsc, x_{i-1}$; or,
\item $x_i$ is non-adjacent to all of $y_1, \dotsc, y_{i-1}$ and $y_i$ is adjacent to all of $x_1,
\dotsc, x_{i-1}$.
\end{enumerate}
\end{definition}

\begin{lemma}[\cite{GPT21}, Lemma 3.3]
\label{lemma:quasi chain number}
For every bipartite graph $G$,
\[
  \ch(G) \leq \qch(G) \leq 4 \cdot \ch(G) + 4 \,.
\]
\end{lemma}

\begin{definition}[Flip]
A \emph{flip} of a bipartite graph $G = (X,Y,E)$ is any graph $G' = (X,Y,E')$ obtained by choosing
any $A \subseteq X, B \subseteq Y$ and negating the edge relation for every pair $(a,b) \in A \times
B$.

For $q \in \bN$ and a graph $G$, a graph $G'$ is a \emph{$q$-flip} of $G$ if there is a sequence $G
= G_0, G_1, \dotsc, G_r = G'$ such that $r \leq q$ and for each $i \in [r]$, $G_i$ is a flip of
$G_{i-1}$.
\end{definition}

\begin{lemma}[\cite{GPT21}, Main Lemma]
\label{lemma:gpt21 main}
For all $k,d \in \bN$, $k,d \geq 2$, there are $r,q \in \bN$ satisfying the following. Let $G =
(X,Y,E, \leq)$ be an ordered bipartite graph of convex twin-width at most $d$ and quasi-chain number
at most $k$. Then there is a division $\cD$ of $G$, 
sets $\cU^X_1, \dotsc, \cU^X_r \subseteq \cD^X$,
$\cU^Y_1, \dotsc, \cU^Y_r \subseteq \cD^Y$,
and a $q$-flip $G'$ of $G$ such that the following holds for $H \define G' / \cD$:
\begin{enumerate}
\item For every edge $(A,B)$ of $H$ there exists $i \in [r]$ such that $A\in\cU^X_i$, $B \in \cU^Y_i$;
\item For each $i\in[r]$, the graph $H[\cU^X_i,\cU^Y_i]$ is a star forest. 
	Moreover, for each star in $H[\cU^X_i,\cU^Y_i]$,
	with center $C$ and leaves $K_1, \dotsc, K_m$, 
	the quasi-chain number of $G[C \cup K_1 \cup \dotsm \cup K_m]$ is at most $k-1$.
\end{enumerate}
\end{lemma}

\begin{lemma}
\label{lemma:tw bipartite}
Let $\cF$ be any class of stable bipartite graphs with bounded twin-width. Then $\cF$ admits a
constant-size equality-based labeling scheme.
\end{lemma}
\begin{proof}
Since $\cF$ is stable, we have $\ch(\cF) = k$ and $\tw(\cF)=d$ for some constants $k,d$.\\

\noindent\textbf{Decomposition tree.}
For any $G \in \cF$, we construct a decomposition tree as follows. Let $k^* = \qch(G)$ which
satisfies $k^* \leq 4 \cdot \qch(G) + 4$ by \cref{lemma:quasi chain number}. The root of the tree is
associated with~$G$. Every node of the tree is associated with an induced subgraph $G'$ of $G$,
defined as follows:

For a node $G' = (X',Y',E')$ with $\qch(G') \leq 1$, we have $\ch(G') \leq \qch(G') \leq 1$ by
\cref{lemma:quasi chain number}, so $G'$ is $P_4$.

For a node $G' = (X',Y',E')$ with $\qch(G') = k' > 1$, let $G'_{\leq} = (X',Y',E',\leq)$ be an
ordered bipartite graph with $\ctw(G'_\leq) \leq d+1$, which exists due to \cref{lemma:convex twin
width}. Let $\cD$ be a division of $G'_\leq$, let $\cU^X_1, \dotsc, \cU^X_r \subseteq \cD^X$, $\cU^Y_1, \dotsc, \cU^Y_r \subseteq \cD^Y$, and let
$F$ be a $q$-flip of $G'_\leq$ such that the following holds for $H = F/\cD$:
\begin{enumerate}
\item For every edge $(A,B)$ of $H$ there exists $i \in [r]$ such that $A\in \cU^X_i, B\in \cU^Y_i$; and
\item For every $i\in[r]$, $H[\cU^X_i,\cU^Y_i]$ is a star forest such that for every star in this forest  with
center $C$ and leaves $K_1, \dotsc, K_m$ we have $\qch(G'[C \cup K_1 \cup \dotsm \cup K_m]) \leq
k'-1$.
\end{enumerate}
The above holds for some $r,q$ determined by \cref{lemma:gpt21 main}. Below we will let $l^*$ be
the maximum value of $r$ over all nodes $G'$.  Note that we may assume that each edge $(A,B)$ of
$H$ appears in at most one set $\cU_i\define \cU^X_i\cup\cU^Y_i$. Otherwise we can remove the leaf incident with $(A,B)$ from all but one set $\cU_i$ that induce a star forest containing $(A,B)$; each star forest remains a star forest.

The child nodes of $G'$ are the bipartite graphs $G'[C \cup K_1 \cup \dotsm \cup K_m]$ for each star
$(C, K_1, \dotsc, K_m)$ in the star forests induced by the sets $\cU_i$. 
Observe that each child $G''$ has $\qch(G'') < \qch(G')$, so the depth of the tree is at most
$\qch(G) = k^* \leq 4k + 4$.\\

\noindent\textbf{Labeling scheme.}
We will construct labels as follows. For a vertex $x$ and a graph $G$ containing~$x$, we will write
$L(x,G)$ for the label of $x$ obtained by the following recursive labeling scheme. For each node
$G'$ of the decomposition tree, we assign labels to the vertices of $G'$ inductively as follows.
\begin{enumerate}[topsep=4pt,parsep=0pt,partopsep=0pt]
\item If $G'$ is a leaf node, so that $\qch(G') \leq 1$, each $x$ is assigned a label $L(x,G') =
(0,p(x) \eqLabelSep q(x))$, where $(p(x) \eqLabelSep q(x))$ is the constant-size equality-based label in the graph
$G'$, which is a $P_4$-free bipartite graph. Recall that $P_4$-free bipartite graphs are bipartite
equivalence graphs, so there is a simple equality-based labeling scheme.
\item If $G'$ is an inner node, perform the following. Let $\cU^X_1,\cU^Y_1, \dotsc, \cU^X_r,\cU^Y_r$ be the sets that
induce star forests in $H$. For each $i$, fix an arbitrary numbering $s_i$ to the stars 
in $H[\cU^X_i,\cU^Y_i]$.
\begin{enumerate}[itemsep=0pt,topsep=4pt,partopsep=0pt]
\item Let $G'_\leq = F_0, F_1, \dotsc, F_{q'} = F$ with $q' \leq q$ be the sequence of flips that
take $G'_\leq$ to $F$. To each vertex $x \in V(G')$ assign a binary vector $f(x) \in \zo^q$ so that
$f(x)_i = 1$ if and only if $x$ is in the set that is flipped to get $F_i$ from $F_{i-1}$.
\item For each $i \in [r]$, any vertex $x \in V(G')$ belongs to at most one star $S_i = (C, K_1,
\dotsc, K_m)$ in $H[\cU^X_i,\cU^Y_i]$. Append the tuple $(1, f(x) \mathbin| s_1(S_1), \dotsc, s_r(S_r))$ where
$s_i(S_i)$ is the number of the star $S_i = (C, K_1, \dotsc, K_m)$ in $H[\cU^X_i,\cU^Y_i]$ containing $x$.
Then for each $i \in [r]$, append $L(x, G'_i)$ for $G'_i = G'[C \cup K_1 \cup \dotsc \cup K_m]$
induced by the star $S_i = (C,K_1, \dotsc, K_m)$.
\label{step:star}
\end{enumerate}
\end{enumerate}

\noindent\textbf{Decoder.}
The decoder for this scheme is defined recursively as follows: Given labels $L(x,G)$, $L(y,G)$ for
vertices $x,y$:
\begin{enumerate}[itemsep=0pt,topsep=4pt,partopsep=0pt]
\item If $L(x,G') = (0, p(x) \eqLabelSep q(x))$, $L(y,G') = (0, p(y) \eqLabelSep q(y))$ then output the adjacency of
$x,y$ in $G'$, which is a bipartite equivalence graph, as determined by the labels $(p(x) \eqLabelSep q(x))$, $(p(y) \eqLabelSep q(y))$.
\item If $L(x,G') = (1, f(x) \eqLabelSep s_1(S_1), \dotsc, s_r(S_r))$, $L(y,G') = (1, f(y) \eqLabelSep s_1(S'_1),
\dotsc, s_r(S'_r))$ then let $i$ be the unique value such that $s_i(S_i)=s_i(S'_i)$, if such a
value exists. In this case, output the adjacency of $x$, $y$ as determined from the labels $L(x,G'_i)$,
$L(y,G'_i)$ where $G'_i$ is the child corresponding to star $S_i = S'_i$.
Otherwise, output the parity of
\[
\bigl|\bigl\{i \in [q] : f(x)_i = f(y)_i = 1 \bigr\}\bigr| \,.
\]
\label{step:flip}
\end{enumerate}

\noindent\textbf{Correctness.}
Let $x \in X, y \in Y$ be vertices of $G$.
\begin{claim}
For any node $G' = (X',Y',E') \sqsubset G$
in the decomposition tree, there is at most one child $G'' \sqsubset G'$ such that $x,y \in V(G'')$.
\end{claim}
\begin{proof}[Proof of claim]
Let $F$ be the $q$-flip of $G'_\leq$, let $\cD$ the division of $G'_\leq$, and write $H = F/\cD$.
Let $A \subseteq X', B \subseteq Y'$ be the unique sets with $A,B \in \cD$ such that $x \in A, y \in
B$. Suppose that for some $i \in [r]$ there is a star $\{A, K_1, \dotsc, K_m\} \subseteq \cU_i$
such that $B \in \{K_1, \dotsc, K_m\}$. Then $(A,B)$ is an edge of $H$ so, by assumption, $(A,B)$
appears in exactly one star forest $H[\cU^X_i,\cU^Y_i]$. $(A,B)$ also appears in exactly one star in
the star forest $H[\cU^X_i,\cU^Y_i]$. So there is exactly one child $G'' = G'[C \cup K_1 \cup \dotsm
\cup K_m]$ that contains both $x$ and $y$.
\end{proof}
It follows from the above claim that there is a unique maximal path $G = G_0, \dotsc, G_t$ in the
decomposition tree, starting from the root, satisfying $x,y \in V(G_i)$ for each $i = 0, \dotsc, t$.

First we prove that the labeling scheme outputs the correct value on $G_t = (X',Y',E')$. If $G_t$ is
a leaf node then this follows from the labeling scheme for $P_4$-free bipartite graphs.  Suppose
$G_t$ is an inner node. Let $A \subseteq X', B \subseteq Y'$ be the unique sets $A, B \in V(H)$ such
that $x \in A, y \in B$. Since $G_0, \dotsc, G_t$ is a maximal path, it must be that there is no $i
\in [r]$ and star $S$ in $H[\cU^X_i,\cU^Y_i]$ such that both $A$ and $B$ are nodes of $S$. Then for every $i \in [r]$, let
$S_i,S'_i$ be the (unique) pair of stars in $H[\cU^X_i,\cU^Y_i]$ such that $A \in S_i, B \in S'_i$; since
$A,B$ are not nodes of the same star, we have $S_i \neq S'_i$ so $s_i(S_i) \neq s_i(S'_i)$ in the
labels. So the decoder outputs the parity of
\[
  \bigl|\bigl\{ i \in [q] : f(x)_i = f(y)_i = 1\bigr\}\bigr|.
\]
We show that $x,y$ are not adjacent in $F$. In this case, $x,y$ are adjacent in $G_t$ if and
only if the pair $(x,y)$ is flipped an odd number of times in the sequence $(G_t)_\leq = F_0, F_1,
\dotsc, F_{q'} = F$; this is equivalent to there being an odd number of indices $i \in [q]$ such
that $f(x)_i = f(y)_i$, so the decoder will be correct.
For contradiction, assume that $x,y$ are adjacent in $F$. Then by definition, $(A,B)$ is an edge of
$H$, so $(A,B)$ must belong to some star $S$ in some forest $H[\cU^X_i,\cU^Y_i]$. But then $s_i(S_i) = s_i(S'_i)
= s_i(S)$, a contradiction. So $x,y$ are not adjacent in $F$.

Now consider node $G_i$ for $i = 0, \dotsc, t-1$. By definition, there is a child $G_{i+1}$ that
contains both $x$ and $y$, so there exists $j \in [r]$ and a star $S = (C, K_1, \dotsc, K_m)$
in $H[\cU^X_j,\cU^Y_j]$ such that $x,y \in C \cup K_1 \cup \dotsm \cup K_m$. Then $s_j(S_j)=s_j(S'_j)=s_j(S)$ so
the decoder will recurse on the child~$G_{i+1}$. 
\end{proof}

%%%%%%%%%%%%%%%%%%%%%%%%%%%%%%%%%%%%%%%%%%%
\subsubsection{First-Order Labeling Schemes \& Distance Sketching}
\label{section:distance sketching}
%%%%%%%%%%%%%%%%%%%%%%%%%%%%%%%%%%%%%%%%%%%

In this section we construct sketches for stable graph classes of bounded twin-width that replace
the adjacency relation with a binary relation $\phi(x,y)$ on the vertices, that is defined by a
formula $\phi$ in first-order logic. We will call such a sketch a \emph{$\phi$-sketch}.

\medskip
\noindent
\textbf{First-order logic.}
A \textit{relational vocabulary} $\tau$ is a set of relation symbols. Each relation symbol $R$ 
has an \textit{arity}, denoted $\text{arity}(R) \geq 1$. A \textit{structure} $\mathcal{A}$ 
of vocabulary $\tau$, or $\tau$-structure,  consists of a set $A$, called the \textit{domain}, and an interpretation
$R^{\mathcal{A}} \subseteq A^{\text{arity(R)}}$ of each relation symbol $R \in \tau$.
To briefly recall the syntax and semantics of first-order logic, we fix a countably
infinite set of \textit{variables}, for which we use small letters. 
\textit{Atomic formulas of vocabulary $\tau$} are of the form:
\begin{enumerate} 
	\item $x = y$ or
	\item $R(x_1, \ldots, x_r)$, meaning that $(x_1, \ldots, x_r) \in R$,
\end{enumerate}
where $R \in \tau$ is $r$-ary and $x_1, \ldots , x_r, x, y$ are variables.
\textit{First-order (FO) formulas} of vocabulary $\tau$ are inductively defined as either the atomic formulas,
a Boolean combination $\neg\phi$, $\phi\wedge\psi$, $\phi\vee\psi$, or a quantification $\exists
x.\phi$ or $\forall x.\phi$, where $\phi$ and $\psi$ are FO formulas.
A \textit{free variable} of a formula $\phi$ is a variable $x$ with an occurrence in $\phi$ 
that is not in the scope of a quantifier binding $x$. 
We write $\phi(x_1, x_2, \ldots, x_k)$ to show that the set of free variable of $\phi$
is $\{ x_1, x_2, \ldots, x_k \}$.
By $\phi[t/x]$, we denote the formula that results from substituting $t$ for free variable $x$ in $\phi$.

\medskip
\noindent
\textbf{First-order labeling schemes.}
We fix a relational vocabulary $\tau$ that consists of a unique (symmetric) binary relational symbol
$E$, and $t$ unary relational symbols $R_1, \dotsc, R_t$. A $\tau$-structure $\cG$ with
domain $V$ is a tuple $(V, E^\cG, R_1^\cG, \dotsc, R_t^\cG)$ where each $R_i^\cG \subseteq V$ is an
interpretation of the symbol $R_i$, and $E^\cG \subseteq V \times V$ is an interpretation of the symbol
$E$. 
For a graph class~$\cF$, we will write $\cF^\tau$ for the set of $\tau$-structures $\cG = (V, E^\cG,
R_1^\cG, \dotsc, R_t^\cG)$ where $(V,E^\cG)$ is a graph in $\cF$. Let $\phi(x,y)$ be an FO formula of
vocabulary $\tau$. For a $\tau$-structure $\cG = (V,E^\cG,R_1^\cG,\dotsc,R_t^\cG)$ and $u,v \in V$, we write $\cG \models \phi[u/x,v/y]$, if $\phi[u/x,v/y]$ is a true statement in $\cG$.
When the $\tau$-structure is clear from context, we drop the superscript $\cG$ (and identify the relation symbol and its interpretation in~$\cG$).

\begin{definition}[First-order sketches]
For a symmetric first-order formula $\phi(x,y)$ with vocabulary $\tau$, and a graph class $\cF$, a
\emph{$\phi$-sketch} with cost $c(n)$ and error $\delta$ is a pair of algorithms: a randomized
\emph{encoder} and a deterministic \emph{decoder}. The encoder takes as input any $\tau$-structure
$G = (V,E,R_1,\dotsc,R_t) \in \cF^\tau$ with $|V|=n$ vertices and outputs a (random) function $\sk :
V \to \zo^{c(n)}$. The encoder and (deterministic) decoder $D : \zo^* \times \zo^* \to \zo$ must
satisfy the condition that for all $G \in \cF^\tau$,
\[
\forall u,v \in V(G) : \qquad \Pru{\sk}{ \vphantom{\Big|} D\bigl(\sk(u),\sk(v)\bigr) = \ind{G \models \phi[u/x,v/y]} } \;\geq\;
1-\delta \,.
\]
As usual, if left unspecified, we assume $\delta=1/3$.  We will write $\phi$-$\RL(\cF)$ for the
smallest function $c(n)$ such that there is a randomized  $\phi$-labeling scheme for $\cF$ with cost
$c(n)$ and error $\delta=1/3$. Setting $\delta=0$ we obtain the notion of \emph{(deterministic)
$\phi$-labeling scheme}.
\end{definition}

\begin{theorem}
\label{thm:FO-RL}
Let $\cF$ be a stable, hereditary graph class with bounded twin-width, and let $\phi(x,y)$ be an FO
formula of vocabulary $\tau$. Then $\phi$-$\RL(\cF) = O(1)$.
\end{theorem}
\begin{proof}
Given a $\tau$-structure $G=(V,E,R_1,\dotsc,R_t) \in \cF^\tau$, we denote by $G^{\phi}$ the graph
with vertex set $V$ and the edge set $E^{\phi} = \{ (u,v) : G \models \phi[u/x,v/y] \}$.  We also define
$\cF^{\phi} := \cl( \{ G^{\phi} : G \in \cF^\tau \} )$.  By definition, $\cF^{\phi}$ is hereditary.
Furthermore, we note that $\cF^{\phi}$ is an FO transduction of~$\cF$.  Hence,
\cref{th:stable-tw-FO} implies that $\cF^{\phi}$ is a stable class of bounded twin-width, and
therefore by \cref{thm:twin width}, $\cF^{\phi}$ admits a constant-size adjacency sketch. Any such
sketch can be used as a $\phi$-sketch for $\cF$, since for any two vertices $u,v$ in $G \in
\cF^\tau$ we have $G \models \phi[u/x,v/y]$ if and only if $u$ and $v$ are adjacent in $G^{\phi}$.
\end{proof}

As a corollary we obtain an answer to an open question of
\cite{Har20}, who asked if the planar graphs admit constant-size sketches for $\dist(x,y) \leq k$.

\begin{corollary}
For any $k \in \bN$, there is a constant-size sketch for planar graphs that decides $\dist(x,y)
\leq k$.
\end{corollary}
\begin{proof}
	The class of planar graphs has bounded twin-width \cite{BKTW20} and bounded chain number (because it is of bounded arboricity). Hence, the result follows from \cref{thm:FO-RL} and the fact that 
	the relation $\dist(x,y) \leq k$ is expressible in first-order logic, \eg via the following FO formula
	\[
	\!\!\delta_k(x,y) \define (\exists v_1, \dotsc, v_{k-1} : (E(x,v_1) \vee x=v_1) \wedge
	(E(v_1,v_2) \vee v_1=v_2) \wedge
	\dotsm \wedge (E(v_{k-1},y) \vee v_{k-1}=y)) \,.
	\]
\end{proof}

%%%%%%%%%%%%%%%%%%%%%%%%%%%%%%%%%%%%%%%%%%%%
\subsection{Graph Products}
\label{section:cartesian product}
%%%%%%%%%%%%%%%%%%%%%%%%%%%%%%%%%%%%%%%%%%%%

The Cartesian product of graphs is defined as follows.

\begin{definition}[Cartesian Products]
Let $d \in \bN$ and $G_1,\ldots,G_d$ be any graphs.  The Cartesian product $G_1 \mathbin\square G_2
\mathbin\square \cdots \mathbin\square G_d$ is the graph whose vertices are the tuples $(v_1,
\dotsc, v_d) \in V(G_1)\times\cdots\times V(G_d)$, and two vertices $v$, $w$ are adjacent if and
only if there is exactly one coordinate $i \in [d]$ such that $(v_i,w_i) \in E(G_i)$ and for all $j
\neq i$, $v_j=w_j$.  For any set of graphs $\cF$, we will define the set of graphs $\cF^\square$ as
all graphs obtained by taking a product of graphs in $\cF$:
\[
  \cF^\square \define \left\{ G_1 \mathbin\square G_2 \mathbin\square \dotsm \mathbin\square G_d : d \in \bN,
\forall i \in [d]\,\, G_i \in \cF \right\} \,.
\]
We also write $G^{\square d} \define G \mathbin\square G \mathbin\square \dotsm \mathbin\square G$ for the $d$-wise product
of $G$.  For example, $P_2^{\square d}$ is the $d$-dimensional hypercube.
\end{definition}
Although Cartesian products are extremely well-studied (\eg see \cite{CLR20} for results on
universal graphs for subgraphs of products), it was not previously known whether the number of
unique induced subgraphs of Cartesian products $G^{\square d}$ is at most $2^{O(n \log n)}$, let
alone whether they admit adjacency labeling schemes of size $O(\log n)$. We resolve this question by
proving that Cartesian products preserve constant-size PUGs and distance-$k$ sketches.  We prove the
following theorem. Note that, unlike adjacency, distances are not necessarily preserved by taking
induced subgraphs, so the distance-$k$ result applies to $\cF^\square$ instead of its hereditary
closure. The following theorem implies \cref{thm:intro-products} from the introduction.

\begin{restatable}{theorem}{thmgraphproducts}
\label{thm:cartesian products}\RestateRemark
If $\cF$ is any class that admits a constant-size PUG (including any finite class),
then $\cl(\cF^\square)$ admits a constant-size PUG. For any fixed $k$, if~$\cF$ admits a
constant-size distance-$k$ sketch, then so does $\cF^\square$.
\end{restatable}

We obtain as a corollary the new result that $\cl(\cF^\square)$ is a factorial class with a
$\poly(n)$-size universal graph, when $\cF$ is any class with a constant-size PUG (including any
finite class), which follows by \cref{lemma:derandomization} and \cref{prop:knr}.
\begin{corollary}
For any hereditary class $\cF$ that admits a constant-size PUG (in particular any finite class), the
hereditary class $\cl(\cF^\square)$ has speed $2^{O(n \log n)}$ and admits a $\poly(n)$-size universal
graph.
\end{corollary}

Follow-up work \cite{EHZ23} simplifies and extends our technique from this section to obtain optimal
adjacency labels for \emph{subgraphs} (not only induced subgraphs) of Cartesian products, improving
upon the results of \cite{CLR20}.

\begin{remark}
There are 3 other common types of graph products \cite{HIK11}: the strong product, the direct
product, and the lexicographic product. Early versions of this paper showed that these graph
products admit constant-cost communication protocols only in trivial cases; we omit this discussion
from the current version.
\end{remark}

We now prove \cref{thm:cartesian products}.  The proof is illustrated in \cref{fig:cartesian
products}. 

\begin{proof}[Proof of Theorem~\ref{thm:cartesian products}]
Let $s(k)$ be the size of the randomized distance $k$ labeling scheme for $\cF$, and let $D :
\zo^{s(k)} \times \zo^{s(k)} \to \zo$ be its decoder. We may assume that this scheme has error
probability at most $\frac1{10k}$, at the expense of at most an $O(\log k)$ factor increase by
\cref{prop:boosting}. For any $G = G_1 \mathbin\square \dotsm \mathbin\square G_d$ where each $G_i \in \cF$, we will
construct a distance $k$ labeling for $G$ as follows.

Choose constants $m,t\in\bN$ with $m \geq 9k^2$, $t \geq 9k$ and $mt\ge 27(k+1)^2$. Then:
\begin{enumerate}
\item For each $i \in [d]$, draw a random labeling $\ell_i : V(G_i) \to \zo^{s(k)}$ from the
distance $k$ scheme for $G_i$.
\item  For each $i \in [m]$, $j \in [t]$ and each vertex $x \in
G$, initialize a vector $w^{(i,j)}(x) \in \zo^{s(k)}$ to $\vec 0$ and a bit $W^{(i,j)}(x)$ to 0.
\item For each $i \in [d]$ choose a uniformly random coordinate $b(i) \sim [m]$. For each $i \in
[d]$ and each $v \in V(G_i)$ choose a random coordinate $c(i,v) \sim [t]$.
\item For each $x \in V(G)$, assign the label as follows. For each $i \in [d]$, update
\begin{align*}
  w^{(b(i), c(i,x_i))}(x) &\gets w^{(b(i), c(i,x_i))}(x) \oplus \ell_i(x_i) \\
  W^{(b(i), c(i,x_i))}(x) &\gets W^{(b(i), c(i,x_i))}(x) \oplus 1 \,.
\end{align*}
Then the label for $x$ is $w(x)$, the collection of all vectors $w^{(i,j)}(x)$ and bits $W^{(i,j)}(x)$ 
for $i \in [m]$, $j \in [t]$.  
So the size of the label is $mt \cdot (s(k)+1)$.
\end{enumerate}
The decoder is as follows. On receiving labels $w(x),w(y)$, perform the following:
\begin{enumerate}
\item For each $i \in [m], j \in [t]$, let $z^{(i,j)} = w^{(i,j)}(x) \oplus w^{(i,j)}(y)$ and
$Z^{(i,j)} = W^{(i,j)}(x) \oplus W^{(i,j)}(y)$.
\item If there exists $i \in [m]$ for which the number of coordinates $j \in [t]$ such that
$Z^{(i,j)} = 1$ is not in $\{0,2\}$, output $\bot$.
\item If there are at least $2k+1$ pairs $(i,j) \in [m] \times [t]$ such that
$Z^{(i,j)} = 1$, output $\bot$.
\item Otherwise, there are $2q \leq 2k$ pairs
\[
(i_1, (j_1)_1), (i_1, (j_1)_2),
(i_2, (j_2)_1), (i_2, (j_2)_2), \dotsc,
(i_q, (j_q)_1), (i_q, (j_q)_2) \in [m] \times [t]
\]
such that $Z^{(i^*,j^*_1)}=Z^{(i^*,j^*_2)} = 1$ for each of these pairs $(i^*, j^*) \in \{ (i_1,
j_1), \dotsc, (i_q, j_q)\}$. 
%For a vector $z \in \zo^{t+1}$ we will write $z_{[t]} = (z_1, \dotsc, z_t)$. 
For each $r \in [q]$, define
\[
  k_r = D(z^{(i_r, (j_r)_1)}, z^{(i_r, (j_r)_2)} ) \in [k] \cup \{\bot\} \,.
\]
If any of these values are $\bot$, output $\bot$. Otherwise output $\sum_{r=1}^q k_r$.
\end{enumerate}
\begin{claim}
\label{claim:product good events}
Let $x,y \in G$ such that there are exactly $q \leq k$ coordinates $I \subseteq [d]$, $|I|
= q$, such that $x_i \neq y_i$ for all $i \in I$.  Then, with probability at least $2/3$, all
of the following events occur:
\begin{enumerate}
\item For all $i \in I$, either $\dist(x_i, y_i) \leq k$ and $\dist(x_i, y_i)
= D(\ell_i(x_i), \ell_i(y_i))$, or $\dist(x_i, y_i) > k$ and $D(\ell_i(x_i), \ell_i(y_i)) =
\bot$;
\item The values $b(i)$, $i \in I$, are all distinct;
\item For all $i \in I$, $c(i, x_i) \neq c(i, y_i)$.
\end{enumerate}
\end{claim}
\begin{proof}[Proof of claim]
By assumption, for any $i \in I$, the probability that $D(\ell_i(x_i),\ell_i(y_i))$ fails to output
$\dist(x_i,y_i)$ (if this is at most $k$) or $\bot$ (if the distance is greater than $k$) is at most
$1/9k$, so by the union bound, the probability that the first event fails to occur is at most
$1/9$.

The probability that event 2 fails to occur is the probability that there exists a distinct pair
$i,j \in I$ such that $b(i)=b(j)$; by the union bound, this is at most $|I|^2/m \leq k^2/m \leq
1/9$.

The probability that there exists $i \in I$ such that $c(i,x_i) = c(i,y_i)$ is at most $|I|/t \leq
k/t \leq 1/9$. Therefore the probability that any one of these events fails to occur is at most $3/9
= 1/3$.
\end{proof}

\begin{claim}
For any $x,y \in G$ such that there are exactly $q \leq k$ coordinates $I \subseteq [d]$, $|I|
= q$, such that $x_i \neq y_i$ for all $i \in I$. Then, if all 3 events in \cref{claim:product
good events} occur, the decoder outputs $\dist(x,y)$ unless there is a coordinate $j$ such that
$\dist(x_j,y_j) > k$, in which case it outputs $\bot$.
\end{claim}
\begin{proof}[Proof of claim]
We first observe that the $2q$ pairs $P = \{(b(i), c(i,x_i)), (b(i), c(i,y_i)) \}_{i \in I}$ are
distinct, because each $b(i)$ is distinct and $c(i,x_i) \neq c(i,y_i)$. First, consider pairs $(b(i), c(i,x_i))$ and $(b(i), c(i,y_i))$ for $i\in I$. Then, since $b(j) \neq b(i)$ for each $j \in I$ with $j \neq i$, and
$c(i,y_i) \neq c(i,x_i)$, we have
\begin{align*}
  Z^{(b(i), c(i,x_i))} &= W^{(b(i), c(i,x_i))}(x) \oplus W^{(b(i), c(i,x_i))}(y) \\
    &= 1 \oplus \ind{c(i,y_i)=c(i,x_i)}
    \oplus \left(\bigoplus_{j \neq i, b(j) = b(i)} \ind{ c(j,x_j) = c(i,x_i) } \right)
\\	&\phantom{{}={}}\qquad
    \oplus \left(\bigoplus_{j \neq i, b(j) = b(i)} \ind{ c(j,y_j) = c(i,x_i) } \right) \\
    &= 1 \oplus \left(\bigoplus_{j \notin I, b(j) = b(i)} \ind{ c(j,x_j) = c(i,x_i) } 
                                                   \oplus \ind{ c(j,y_j) = c(i,x_i) } \right) \\
    &= 1 \oplus \left(\bigoplus_{j \notin I, b(j) = b(i)} \ind{ c(j,x_j) = c(i,x_i) } 
                                                   \oplus \ind{ c(j,x_j) = c(i,x_i) } \right) \\
    &= 1 \,.
\end{align*}
Multiplying each indicator $\ind{(b(j),c(j,x_j)) =
(b(i),c(i,x_i))}$ by the vector $\ell_i(x_i)$, and each indicator $\ind{(b(j),c(j,y_j)) =
(b(i),c(i,x_i))}$ by the vector $\ell_i(y_i)$, we obtain the following:
\[
  z^{(b(i), c(i,x_i))} = w^{(b(i),c(i,x_i))} = \ell_i(x_i) \,.
\]
By similar reasoning, $Z^{(b(i), c(i,y_i))} = 1$, and
\[
  z^{(b(i), c(i,y_i))} = w^{(b(i),c(i,y_i))} = \ell_i(y_i) \,.
\]
Now consider any $i'\in[m]$, $j'\in[t]$ such that $(i',j')\notin P$.
If there exists no $i\in[d]$, such that $i' = b(i)$, then
clearly $W^{(i',j')}(x)$ and $W^{(i',j')}(y)$ remain $0$ and $Z^{(i',j')}=0$, so these entries
do not contribute.
Otherwise $i' = b(i)$ for some $i\in[d]$. 
If $c(i,x_i) \ne j' \ne c(i,y_i)$, again the label entries are not touched and $Z^{(i',j')}=0$.
Finally, assume $(i',j') = (b(i),c(i,x_i))$.
Clearly $(b(j),c(j,x_j)) \neq (b(i),c(i,x_i))$ for any $j \in I$ since $(i',j') \notin P$;
in particular $i \notin I$. Thus $y_i = x_i$ by assumption.
%Similarly, any $(b(j),c(j,y_j)) \neq (b(i),c(i,x_i))$ for $j \in I$. 
Then
\begin{align*}
  Z^{(b(i), c(i,x_i))}
  &= W^{(b(i), c(i,x_i))}(x) \oplus W^{(b(i), c(i,x_i))}(y) \\
  &= \left( \bigoplus_{j : b(j)=b(i)} \ind{c(j,x_j)=c(i,x_i)} \right) \oplus
     \left( \bigoplus_{j : b(j)=b(i)} \ind{c(j,y_j)=c(i,x_i)} \right) \\
  &= \left( \bigoplus_{j \notin I: b(j)=b(i)} \ind{c(j,x_j)=c(i,x_i)}
      \oplus \ind{c(j,y_j)=c(i,x_i)} \right) \\
  &= \left( \bigoplus_{j \notin I: b(j)=b(i)} \ind{c(j,x_j)=c(i,x_i)}
      \oplus \ind{c(j,x_j)=c(i,x_i)} \right) = 0 \,.
\end{align*}
Since $x_i=y_i$, also $Z^{(b(i), c(i,y_i))}=0$. We may then conclude that the $2q$ distinct
pairs $P$ are the only pairs $(i,j)$ such that $Z^{(i,j)}=1$. Therefore the decoder will output
\[
  \sum_{i \in I} D\left( w^{(b(i),c(i,x_i))}, w^{(b(i),c(i,y_i))} \right)
  = \sum_{i \in I} D\left( \ell_i(x_i), \ell_i(y_i) \right) \,,
\]
which is $\dist(x,y)$ when each $\dist(x_i,y_i) \leq k$ and $\bot$ otherwise.
\end{proof}
Suppose that $x,y \in V(G)$ have distance at most $k$. Then the above two claims suffice to show
that the decoder will output $\dist(x,y)$ with probability at least $2/3$.

Now suppose that $x,y \in V(G)$ have distance greater than $k$. There are three cases:
\begin{enumerate}
\item There are at least $k+1$ coordinates $i \in [d]$ such that $x_i \neq y_i$. In this case, using
\cref{prop:hamming distance} below (with $\delta=\frac13$, $u=mt$, $n=k+1$, $k=k$) yields
that vector $Z$ has at least $k+1$ 1-valued coordinates
$Z^{(i,j)}$ with probability at least $2/3$, and the decoder will correctly output $\bot$.
\item There are at most $k$ coordinates $i \in [d]$ such that $x_i \neq y_i$, and there is some
coordinate $i$ such that $\dist(x_i,y_i) > k$. Then the above two claims suffice to show that the
decoder will output $\bot$ with probability at least $2/3$, as desired.
\item There are at most $k$ coordinates $i \in [d]$ such that $x_i \neq y_i$, and all coordinates
satisfy $\dist(x_i,y_i) \leq k$. Then the above two claims suffices to show that the decoder will
output $\dist(x,y)$ with probability at least $2/3$.
\end{enumerate}
This concludes the proof.
\end{proof}

\begin{proposition}
	\label{prop:hamming distance}
	For any $0 < \delta < 1$ and
	any $u,k,n \in \bN$, where $u \ge \frac{9(k+1)^2}{\delta}$
	and $n > k$ the following holds.  Write $e_i \in \bF_2^u$ for the $i$th standard basis vector. 
	Let $R_1,
	\dotsc, R_n \sim [u]$ be uniformly and independently random.  
	Then $\Pr{|e_{R_1} + \dotsm + e_{R_n}| \leq k} < \delta$, 
	where $|v|$ is the number of 1-valued coordinates in $v$ and addition is in $\bF_2^u$.
\end{proposition}
\begin{proof}
	Set $t = \lfloor \frac{1}{3} \sqrt{\delta u} \rfloor$. We first prove the following claim.
	
	\smallskip
	\noindent
	\textbf{Claim.}
	\textit{For any vector $e\in \bF_2^u$
	and $S_1, \ldots, S_t \sim [u]$, we have
	$\Pr{|e + e_{S_1} + \dotsm + e_{S_t}| \leq k} < \delta$.}
	
	\smallskip
	\noindent
	\textit{Proof.} Observe that $t > \frac{1}{3} \sqrt{\delta u} -1 \ge k$, and $t\le \frac{1}{3} \sqrt{\delta u}$.
	Suppose first that $|e| > k + t$. Then $\Pr{|e + e_{S_1} + \dotsm + e_{S_t}| \leq k}  = 0 < \delta$,
	as each of the vectors $e_{S_i}$ can flip at most one coordinate in $e$ from 1 to 0.
	So assume now that $|e| \leq k + t$. For a fixed $i \in [t]$ we have
	$\Pr{|e + e_{S_i}| < |e|} \leq \frac{k + t}{u} \le \frac{2t}u \le \frac{2\sqrt{\delta u}/3}{u} = \frac23\frac{\sqrt{\delta}}{\sqrt{u}}$.
	Then for $A$, the event that $\exists i \in [t] : |e + e_{S_i}| < |e|$, we have $\Pr{A} \leq  
	\frac{1}{3}\sqrt{\delta u} \cdot \frac23\frac{\sqrt{\delta}}{\sqrt{u}} < \delta/2$.
	Let further $B$ be the event that $\exists i,j \in [t] : S_i=S_j$; then also 
	$\Pr{B} \le \binom t2 \frac{1}{u} \leq \frac{t^2}{2u} < \delta/2$.
	If $A$ and $B$ both do \emph{not} occur, we have 
	$|e + e_{S_1} + \dotsm + e_{S_t}| = |e| + |e_{S_1}| + \dotsm + |e_{S_t}| \ge t>k$,
	so $\Pr{|e + e_{S_1} + \dotsm + e_{S_t}| \leq k} \le \Pr{A\cup B} < \delta$,
	which proves the claim. \qed
	
	\medskip
	\noindent
	To prove the proposition we consider two cases. The first case is when $n \le t$.
	In this case, as above, the probability that there are two distinct $i,j \in [n]$ such
	that $R_i=R_j$ is at most $\binom n2\frac{1}{u} \le \binom{t}{2}\frac{1}{u} < \delta/2$.
	If all $R_i$ are pairwise different, then $|e_{R_1} + \dotsm + e_{R_n}| = n > k$. 
	Hence, $\Pr{|e_{R_1} + \dotsm + e_{R_n}| \leq k} < \delta/2 < \delta$.
	
	The second case is when $n > t$. In this case we apply the above claim with
	$e = e_{R_{1}} + \cdots + e_{R_{n-t}}$ and $S_1 = R_{n-t+1}, \ldots, S_t = R_n$.
\end{proof}

\begin{figure}[tb]
\centering%\input{cartesian_product_figure.tex}

		\scalebox{1.5}{
	\begin{tikzpicture}[scale=0.7]
		
		% Square Cell style
		\tikzstyle{sq-cell}=[draw, minimum width=0.3cm, minimum height=0.3cm, inner sep=0pt]
		
		% Cell style
		\tikzstyle{cell}=[draw, minimum width=0.3cm, minimum height=1.5cm, inner sep=0pt]
		
		% Custom colors
		\definecolor{mygreen}{RGB}{0,150,0}
		\definecolor{myred}{RGB}{255, 77, 77}
		\definecolor{mygray}{RGB}{240, 240, 240}
		
		% Layout
		\def\bucketskip{3.5}
		\def\topY{3.4}
		\def\midY{2}
		\def\botY{-0.6}
		\def\bottomY{-3.2}
		
		% Labels on left
		\node at (-0.9, \topY-3) {\scriptsize $Z$:};
		\node at (-0.9, \midY-3) {\scriptsize $z$:};
		\node at (-1.8, \botY-3.6) {\scriptsize distance};
		\node at (-2.4, \botY-4.1) {\scriptsize sketches for $\mathbf{x}$:};
		\node at (-1.8, \bottomY-3.2) {\scriptsize distance};
		\node at (-2.4, \bottomY-3.8) {\scriptsize sketches for $\mathbf{y}$:};
		
		% Bucket titles
		\foreach \i/\label in {0/{\scriptsize bucket 1}, 1/{\scriptsize bucket 2}, 2/{\scriptsize bucket 3}, 3/{\scriptsize bucket 4}} {
			\node at (\i*\bucketskip+1, \topY-1.5) {\label};
		}
		
		% BUCKET 1
		\def\xbase{0}
		% Z
		\foreach \i/\c in {0/white,1/black,2/white,3/white,4/white,5/black} {
			\node[sq-cell, fill=\c] at (\xbase+0.5*\i, \topY-3) {};
		}
		% z
		\foreach \i/\c in {0/white,1/myred,2/white,3/white,4/white,5/mygreen} {
			\node[cell, fill=\c] at (\xbase+0.5*\i, \midY -3) {};
		}
		% x
		\foreach \i/\c in {0/mygray,1/mygray,2/mygray,3/myred,4/mygray,5/mygray,6/mygray} {
			\node[cell, fill=\c] at (\xbase+0.5*\i, \botY - 4) {};
		}
		% y
		\foreach \i/\c in {0/mygray,1/mygray,2/mygray,3/mygreen,4/mygray,5/mygray, 6/mygray} {
			\node[cell, fill=\c] at (\xbase+0.5*\i, \bottomY - 3.5) {};
		}
		\node at (\xbase+0.5*1.6, \topY-2.5) {\scriptsize $Z^{(1,2)}$};
		\node at (\xbase+0.5*5.6, \topY-2.5) {\scriptsize $Z^{(1,6)}$};
		
		% BUCKET 2
		\def\xbase{1*\bucketskip}
		\foreach \i/\c in {0/black,1/white,2/white,3/black,4/white,5/white} {
			\node[sq-cell, fill=\c] at (\xbase+0.5*\i, \topY - 3) {};
		}
		\foreach \i/\c in {0/mygreen,1/white,2/white,3/myred,4/white,5/white} {
			\node[cell, fill=\c] at (\xbase+0.5*\i, \midY - 3) {};
		}
		\foreach \i/\c in {0/mygray,1/mygray,2/mygray,3/mygray,4/mygray,5/mygray,6/mygray} {
			\node[cell, fill=\c] at (\xbase+0.5*\i, \botY - 4) {};
		}
		\foreach \i/\c in {0/mygray,1/mygray,2/mygray,3/mygray,4/mygray,5/mygray,6/mygray} {
			\node[cell, fill=\c] at (\xbase+0.5*\i, \bottomY - 3.5) {};
		}
		\node at (\xbase+0.5*0.6, \topY-2.5) {\scriptsize $Z^{(2,1)}$};
		\node at (\xbase+0.5*3.6, \topY-2.5) {\scriptsize $Z^{(2,4)}$};
		
		% BUCKET 3
		\def\xbase{2*\bucketskip}
		\foreach \i/\c in {0/white,1/white,2/white,3/white,4/white,5/white} {
			\node[sq-cell, fill=\c] at (\xbase+0.5*\i, \topY - 3) {};
			\node[cell, fill=\c] at (\xbase+0.5*\i, \midY - 3) {};
		}
		\foreach \i/\c in {0/mygray,1/mygray,2/myred,3/mygray,4/mygray,5/mygray,6/mygray} {
			\node[cell, fill=\c] at (\xbase+0.5*\i, \botY - 4) {};
		}
		\foreach \i/\c in {0/mygray,1/mygray,2/mygreen,3/mygray,4/mygray,5/mygray,6/mygray} {
			\node[cell, fill=\c] at (\xbase+0.5*\i, \bottomY - 3.5) {};
		}
		
		% BUCKET 4
		\def\xbase{3*\bucketskip}
		\foreach \i/\c in {0/black,1/white,2/black,3/white,4/white,5/white} {
			\node[sq-cell, fill=\c] at (\xbase+0.5*\i, \topY - 3) {};
		}
		\foreach \i/\c in {0/myred,1/white,2/mygreen,3/white,4/white,5/white} {
			\node[cell, fill=\c] at (\xbase+0.5*\i, \midY - 3) {};
		}
		\foreach \i/\c in {0/mygray,1/mygray,2/myred,3/mygray,4/mygray,5/mygray} {
			\node[cell, fill=\c] at (\xbase+0.5*\i, \botY - 4) {};
		}
		\foreach \i/\c in {0/mygray,1/mygray,2/mygreen,3/mygray,4/mygray,5/mygray} {
			\node[cell, fill=\c] at (\xbase+0.5*\i, \bottomY - 3.5) {};
		}
		\node at (\xbase+0.5*0.6, \topY-2.5)  {\scriptsize $Z^{(4,1)}$};
		\node at (\xbase+0.5*2.6, \topY-2.5)  {\scriptsize $Z^{(4,3)}$};
		
		% Arrows (only for selected buckets)
		\draw[-stealth] (0*\bucketskip + 0.5*3, \botY - 4) -- (0*\bucketskip  + 0.5*1, \midY - 4.1);
		\draw[-stealth] (0*\bucketskip + 0.5*3, \bottomY - 3.5) -- (0*\bucketskip  + 0.5*5, \midY - 4.1);
		
		\draw[-stealth] (2*\bucketskip + 0.5*2, \botY - 4) -- (1*\bucketskip  + 0.5*3, \midY - 4.1);
		\draw[-stealth] (2*\bucketskip + 0.5*2, \bottomY - 3.5) -- (1*\bucketskip  + 0.5*0, \midY - 4.1);		
		
		\draw[-stealth] (2*\bucketskip + 0.5*4, \botY - 4) -- (2*\bucketskip  + 0.5*3, \midY - 4.1);
		\draw[-stealth] (2*\bucketskip + 0.5*4, \bottomY - 3.5) -- (2*\bucketskip  + 0.5*3, \midY - 4.1);		
		
		\draw[-stealth] (3*\bucketskip + 0.5*2, \botY - 4) -- (3*\bucketskip  + 0.5*0, \midY - 4.1);
		\draw[-stealth] (3*\bucketskip + 0.5*2, \bottomY - 3.5) -- (3*\bucketskip  + 0.5*2, \midY - 4.1);		
	\end{tikzpicture}
}

\smallskip
\caption{Small-distance sketch for Cartesian products. Along the bottom are the distances sketches
for $x_i$ (top) and $y_i$ (bottom) for $i=1$ to $d$. Where $x_i=y_i$, the sketches for $x_i$ and
$y_i$ are equal and are colored grey; they cancel out in $Z$ and $z$. Where $x_i \neq y_i$, the
sketch for $x_i$ is in red and the sketch for $y_i$ is in green. For $i \neq j$ where $x_i \neq y_i$
and $x_j \neq y_j$, the sketches for $x_i,y_i$ and $x_j,y_j$ are mapped to different buckets with
high probability (\ie, $b(i) \neq b(j)$) and the sketches for $x_i$ and $y_i$ are mapped to
different locations in the same bucket with high probability (\ie, $c(i,x_i) \neq c(i,y_i)$).}
\label{fig:cartesian products}
\end{figure}

\section{Question II. Structure: Stability is Sometimes Sufficient} 
\label{section:stability}

Recall Question II: What structures of a problem explain the existence or non-existence of a
constant-cost protocol? Or equivalently, which structures of a hereditary graph class explain the
existence or non-existence of a constant-size PUG?

In this section, we approach the question by looking at hereditary graph classes $\cF$ which satisfy
the necessary condition of having factorial speed, but \emph{do not} have a constant-size PUG, and
asking: What structure of a subclass $\cH \subset \cF$ determines whether it \emph{does} have a
constant-size PUG? We hope to find statements of the form: ``A subclass $\cH \subset \cF$ admits a
constant-size PUG if and only if it has structure $X$''.

We consider examples of two types of hereditary graph classes:\\

\noindent
\textbf{\cref{section:intersection graphs}.} Geometric intersection graphs are an important type of
hereditary graph classes. These are classes graphs obtained by representing vertices as geometric
objects (like lines in $\bR^2$) and connecting vertices by an edge if their objects intersect. These
types of graphs are important for both randomized communication complexity (\eg
\cite{PS86,HHPTZ22,HHM23,HZ23}) and adjacency labeling schemes (\eg \cite{Spin03,Fit19}) and are the
subject of open problems in both fields. Keeping in mind \cref{question:igq-to-pug}, we consider two
types of geometric intersection graphs -- interval graphs and permutation graphs -- which easily
admit polynomial-size universal graphs, and ask what further structure is necessary and sufficient
to get constant-size PUGs instead. We find that \emph{stability} is the answer in both cases.\\

\noindent
\textbf{\cref{section:bipartite graphs}.}
Any hereditary subclass of bipartite graphs can be defined via a (possibly infinite) set $\cH$
of \emph{bipartite} forbidden induced subgraphs. The ones obtained by finite sets $\cH$ are called
\emph{finitely-defined}, with the simplest example being the \emph{monogenic} classes of bipartite
graphs, obtained by forbidding exactly one bipartite graph $H$. It is known exactly for which graphs
$H$ the family of $H$-free bipartite graphs has the required condition of factorial speed, but it is
not yet known whether all the $H$-free bipartite graph classes with factorial speed have
polynomial-size universal graphs. We show that, for monogenic bipartite graph classes having
factorial speed, stability is once again the necessary and sufficient condition to guarantee a
constant-size PUG.

%%%%%%%%%%%%%%%%%%%%%%%%%%%%%%%%%%%%%%%%%%%%%
\subsection{Interval \& Permutation Graphs}
\label{section:intersection graphs}
%%%%%%%%%%%%%%%%%%%%%%%%%%%%%%%%%%%%%%%%%%%%%

\begin{definition}[Interval graph]
A graph $G$ is an \emph{interval graph} if there
exists an \emph{interval realization} $\ell, r : V(G) \to \bR$ with $\ell(v) \le r(v)$ for all $v\in
V(G)$ so that $u,v\in V(G)$ are adjacent in $G$ if and only if $[\ell(u),r(u)]\cap[\ell(v),r(v)] \ne
\emptyset$.
\end{definition}

The class of interval graphs admits a simple $O(\log n)$-bit adjacency labeling scheme: Fix an
interval realization of a given $n$-vertex interval graph $G$ where all endpoints of the intervals
are distinct integers in $[2n]$ and assign to each vertex a label consisting of the two endpoints of
the corresponding interval.

Permutation graphs are another important class of geometric intersection graphs (like interval
graphs, they are a subclass of segment intersection graphs). They also admit a straightforward
$O(\log n)$-bit adjacency labeling scheme that follows from their definition.

\begin{definition}[Permutation Graph]
	A graph $G$ is a \emph{permutation graph} on $n$ vertices if each vertex can be identified with a
	number $i \in [n]$, such that there is a permutation $\pi$ of $[n]$ where $i,j$ are adjacent if and
	only if $i < j$ and $\pi(i) > \pi(j)$.
\end{definition}

In this section we prove \cref{thm:intro-intersection} from the introduction, restated more formally:
\begin{restatable}{theorem}{thmintersection}
\label{thm:intersection}\RestateRemark
Let $\cF$ be any hereditary subclass of interval or permutation graphs. Then $\cF$ admits a
constant-size PUG if and only if $\cF$ is stable.
\end{restatable}

This proof requires new structural results for interval and permutation graphs\footnote{A weaker
version of our permutation graph result (without an explicit bound on the sketch size)
also follows by our result for bounded twin-width combined with \cite{BKTW20}. We keep our
direct proof because: it gives an explicit bound on the sketch size; it is much clearer than the
twin-width result; the decomposition scheme may be of independent interest; and it was proved
independently of \cite{GPT21}, which we use for our twin-width result.}. We show that any stable
subclass of interval graphs can be reduced to a class with bounded treewidth, and therefore has a
constant-size PUG (implied by \cref{lemma:arboricity}).  A consequence of our proof is that stable
subclasses of interval graphs have bounded twin-width (which is not true for general interval
graphs \cite{BKTW20}).  For permutation graphs, we give a new decomposition scheme whose depth can
be controlled by the chain number.

We prove \cref{thm:intersection} separately for interval graphs (\cref{thm:interval graphs})
and for permutation graphs (\cref{thm:permutation graphs}) in the sections below.

\subsubsection{Interval Graphs}

The proof will rely on bounding the clique number of interval graphs with bounded chain number.

\begin{lemma}
\label{lem:bounded-clique-number-implies-RL-const}
Let $\cF$ be a class of interval graphs with bounded clique number, \ie, there is a
constant $c$ so that for any graph $G\in\cF$, the maximal clique size of $G$ is at most $c$.  Then $\cF$
admits a constant-size equality-based labeling scheme.
\end{lemma}
\begin{proof}
Any interval graph is \emph{chordal} and the treewidth of a chordal graph is one less its clique
number.  It follows that any interval graph $G$ with clique number at most $c$ has treewidth at most
$c-1$.  Graphs of treewidth $c-1$ have arboricity at most $O(c)$, and therefore, by
\cref{lemma:arboricity}, $\cF$ admits a constant-size equality-based labeling scheme and an
adjacency sketch of size $O(c)$.
\end{proof}

It is not possible in general to bound the clique number of interval graphs with bounded chain
number, because there may be an arbitrarily large set of vertices realized by identical intervals,
which forms an arbitrarily large clique. Our first step is to observe that, for the purpose of
designing an equality-based labeling scheme, we can ignore these duplicate vertices (called
\emph{true twins} in the literature).

\begin{definition}
For a graph $G = (V,E)$, two distinct vertices $x,y$ are called \emph{twins} if $N(x) \setminus
\{y\} = N(y) \setminus \{x\}$, where $N(x)$, $N(y)$ are the neighbourhoods of $x$ and $y$ in $G$.
Twins $x$ and $y$ are \emph{true twins} if they are adjacent, and \emph{false twins} if they are not
adjacent. The false-twin relation and true-twin relation are equivalence relations.

A graph is \emph{true-twin-free} if it does not contain any vertices $x$, $y$ that are true twins, and
it is \emph{false-twin-free} if it does not contain any vertices $x$, $y$ that are false twins. It
is \emph{twin-free} if it is both true-twin-free and false-twin-free.
\end{definition}

\begin{lemma}
\label{lemma:true-twin reduction}
Let $\cF$ be any hereditary graph class and let $\cF'$ be either the set of true-twin-free members
of $\cF$, or the set of false-twin-free members of $\cF$. If $\cF'$ admits an $(s,k)$-equality based
labeling scheme, then $\cF$ admits an $(s,k+1)$-equality based labeling scheme.
\end{lemma}
\begin{proof}
We prove the lemma for $\cF'$ being the true-twin-free members of $\cF$; the proof for the
false-twin-free members is similar.  Let $G \in \cF$. We construct a true-twin-free graph $G' \in
\cF'$ as follows. Let $T_1, \dotsc, T_m \subseteq V(G)$ be the equivalence classes under the
true-twin relation, so that for any $i$, any two vertices $x,y \in T_i$ are true twins. For each $i \in
[m]$, let $t_i \in T_i$ be an arbitrary element, and let $T = \{t_1, \dotsc, t_m\}$. We claim that
$G[T]$ is true-twin free.

Suppose for contradiction that $t_i$, $t_j$ are true twins in $G[T]$. Let $x \in T_i$, $y \in T_j$.
Since $t_i$, $t_j$ are adjacent in $G[T]$, they are adjacent in $G$. Then $x$ is adjacent to $t_j$
since $x,t_i$ are twins. Since $t_j,y$ are twins, $x$ is adjacent to $y$. So $G[T_i,T_j]$ is a
biclique. Now suppose that $z \in T_k$ for some $k \notin \{i,j\}$, and assume $z$ is adjacent to
$x$. Then $z$ is adjacent to $t_i$ since $x, t_i$ are twins, and $t_i$ is adjacent to $t_k$ since
$z, t_k$ are twins. Since $t_i$, $t_j$ are twins, it also holds that $t_j$ is adjacent to $t_k$ and
to $z$. So $y$ is adjacent to $z$. Then for any $z$ it holds that $x,z$ are adjacent if and only if
$y,z$ are adjacent. So $x,y$ are true twins, for any $x \in T_i, y \in T_j$. But then $T_i \cup T_j$
is an equivalence class under the true-twin relation, which is a contradiction.

Therefore, $G[T]$ is true-twin free and a member of $\cF$, so $G[T] \in \cF'$. For any $x \in V(G)$,
assign the label $( p(t_i) \eqLabelSep q(t_i), i )$ where $( p(t_i) \eqLabelSep q(t_i) )$ is the label of $t_i$ in
the equality-based labeling scheme for $\cF'$, and $i \in [m]$ is the unique index such that $x \in
T_i$. The decoder for $\cF$ performs the following on labels $( p(t_i) \eqLabelSep q(t_i), i)$ and
$( p(t_j) \eqLabelSep q(t_j), j)$. If $i = j$ output 1; otherwise simulate the decoder for $\cF'$ on
labels $(p(t_i) \eqLabelSep q(t_j))$ and $(p(t_j) \eqLabelSep q(t_j) )$. For vertices $x,y$ in $G$, if $x,y$ are
true twins then $i=j$ and the decoder outputs 1. Otherwise, the adjacency between $x$ and $y$ is
equivalent to the adjacency between $t_i,t_j$ in $G[T]$, which is computed by the decoder for
$\cF'$, as desired.
\end{proof}

The true-twin-free interval graphs with bounded chain number also have bounded clique number.
\begin{lemma}
\label{lem:interval-graphs-clique-bounds-chain}
Let $G$ be a true-twin-free interval graph and let $G$ contain a clique of $c$ vertices.
Then $G$ has chain number at least $\lfloor \sqrt c/2\rfloor$.
\end{lemma}
\begin{proof}
Since $G$ is interval, there is an interval realization $\ell, r : V(G) \to \bR$ with $\ell(v)\le r(v)$
for all $v\in V(G)$ so that $u,v\in V(G)$ are adjacent if and only if
$[\ell(u),r(u)]\cap[\ell(v),r(v)] \ne \emptyset$.  We can assume without loss of generality
that no two endpoints are the same. We abbreviate $i(v)= [\ell(v),r(v)]$.
	%	so that 
	%	$V=\{v_1,\ldots,v_n\} = \{[\ell_1,r_1],\ldots,[\ell_n,r_n]\}$ with
	%	$\ell_i \le r_i$ and $v_i$ adjacent to $v_j$ if and only if $[\ell_i,r_i]\cap[\ell_j,r_j] \ne\emptyset$.
	Let $X = \{x_1,\ldots,x_c\}$ be the $c$ vertices that form a $c$-clique in $G$,
	arranged so that $\ell(x_1)< \ell(x_2)<\cdots< \ell(x_c)$.
	Consider the sequence $r = (r(x_1),\ldots,r(x_c))$ of right endpoints.
	By the Erd{\H
o}s-Szekeres theorem \cite{ES35}, any sequence of at least $R(k) = (k-1)^2 + 1$ distinct numbers contains 
	either an increasing or a decreasing subsequence of length at least $k$.
	Setting $k=\lfloor \sqrt c\rfloor$, we have $R(k) \le c$,
	so there is a clique over $k$ vertices $y_1,\ldots,y_k$ with $\ell(y_1)<\cdots<\ell(y_k)$
	and either $r(y_1)<\cdots<r(y_k)$ or $r(y_1)>\cdots>r(y_k)$.
	Graphically speaking, the intervals for $y_1,\ldots,y_k$ either form a staircase or a (step) pyramid.
	In either case, $\ell(y_k) < \min\{r(y_1),r(y_k)\}$, so
	$m = (\min\{r(y_1),r(y_k)\} + \ell(y_k))/2$ is contained in all $i(y_j)$.
	We will assume the staircase case, $r(y_1)<\cdots<r(y_k)$; 
	see \cref{fig:interval-graphs-clique-bounds-chain}. 
	The pyramid case is similar.
	
	\begin{figure}[tbhp]
		\centering
		\begin{tikzpicture}[
			xscale=.95,yscale=1.6,
			noty/.style={blue!80!black,dashed,opacity=.4,|-|},
			yline/.style={blue!80!black,|-|},
			aline/.style={blue!80!black,very thick,|-|},
			zline/.style={red!60!black,|-|},
			bline/.style={red!60!black,very thick,|-|},
			]
			\foreach [count=\i] \r/\s in {5/yline,6/aline,4/noty,7/yline,1/noty,2/noty,8/aline,3/noty,9/aline} {
				\draw[\s] (\i-10,-\i/2) coordinate[label={left,scale=.55:$\ell(x_{\i})$}] (l\i) 
				-- (\r,-\i/2) coordinate[label={right,scale=.55:$r(x_{\i})$}] (r\i) ;
			}
			\draw[ultra thick,dotted,gray] (0,0) -- ++(0,-5) node[below] {$m$};
			\foreach [count=\j] \r/\i/\a in {5/1/,6/2/{=a_1},7/4/,8/7/{=a_2},9/9/{=a_3}} {
				\node[scale=.75,yline,anchor=south west] at (0.25,-\i/2) {$y_\j\a$} ;
			}
			\foreach [count=\j] \zl/\zr/\y/\s/\b in 
			{5.5/8.5/1/bline/=b_1,%
				-10+2.75/-10+2.25/2/zline/,%
				7.5/11/6/bline/=b_2,%
				8.5/10/7/bline/=b_3%
			} {
				\draw[\s] (\zl,-\y/2-.25) -- node[above,scale=.75] {$z_{\j}\b$} (\zr,-\y/2-.25) ;
			}
		\end{tikzpicture}
		\caption{%
			Example illustrating the notation from the proof of \cref{lem:interval-graphs-clique-bounds-chain}.
			The fat blue and red intervals, $a_1,\ldots,a_3$ resp.\ $b_1,\ldots,b_3$,
			form an induced subgraph with \chainNum 3.
		}
		\label{fig:interval-graphs-clique-bounds-chain}
	\end{figure}

	Now, since $G$ is true-twin free, for every $v,v'\in X$, there must be $u\in V(G)$ with
	$i(v)\cap i(u)=\emptyset$ and $i(v')\cap i(u)\ne \emptyset$ or vice versa, so
	$i(u)$ must lie entirely to the left or entirely to the right of $i(v)$ or $i(v')$.
	In particular, for pair $y_j,y_{j+1}$ with $j\in[k-1]$, there must be $z_j\in V(G)$
	adjacent to exactly one of these vertices. 
	So the endpoints of $i(z_j)$ are on the same side of $m$ (``left'' or ``right'' of $m$)
	and the endpoint closer to $m$ must be
	either between $\ell(y_j)$ and $\ell(y_{j+1})$
	or between $r(y_j)$ and $r(y_{j+1})$.

	Among the $k-1$ intervals $i(z_1),\ldots,i(z_{k-1})$, at least $h=\lceil(k-1)/2\rceil$
	are on the same side of $m$.
	Assume the majority is on the right; the other case is similar.
	So for $1\le j_1<\cdots<j_h \le k-1$ intervals $i(z_{j_1}),\ldots,i(z_{j_h})$ are all to the right of $m$.
	Define $B=(b_1,\ldots,b_h) = (z_{j_1},\ldots,z_{j_h})$
	and $A = (a_1,\ldots,a_h) = (y_{j_1+1},\ldots,y_{j_h+1})$.
	By definition, 
	$\ell(b_1)<r(a_1)<\ell(b_2)<r(a_2)<\cdots<\ell(b_h)<r(a_h)$,
	so $a_i$ is adjacent to $b_j$ if and only if $j \le i$.
	Hence $G[A,B]$ is isomorphic to $H^{\circ\circ}_h$.
	It is easy to check that 
	$h = \lceil (\lfloor\sqrt c\rfloor-1)/2\rceil = \lfloor \sqrt c/2\rfloor$.
\end{proof}

With these preparations, the proof of the main result of this section becomes easy.

\begin{theorem}
\label{thm:interval graphs}
Let $\cF$  be a stable class of interval graphs.
Then $\cF$ admits a constant-size equality-based adjacency labeling scheme, and hence $\RL(\cF_n) = O(1)$.
\end{theorem}
\begin{proof}
Since $\cF$ is stable, we have $\ch(\cF) = k$ for some constant $k$. Let $\cF'$ be the set of
true-twin-free members of $\cF$, and let $G' \in \cF'$. Then $\ch(G')\leq k$, and hence the clique
number of $G'$ is at most $4(k+1)^2$ by (contraposition of)
\cref{lem:interval-graphs-clique-bounds-chain}. So $\cF'$ is a class of interval graphs with clique
number bounded by $4(k+1)^2$, and hence by \cref{lem:bounded-clique-number-implies-RL-const}, it
admits a constant-size equality-based labeling scheme (and a size $O(k^2)$ adjacency sketch). By
\cref{lemma:true-twin reduction}, so does $\cF$.
\end{proof}

\begin{remark}
We obtain an adjacency sketch of size $O(k^2)$ for interval graphs with chain number $k$.  There is
another, less direct proof of the above theorem that uses \emph{twin-width} instead of
\cref{lemma:true-twin reduction}, but does not recover this explicit bound on the sketch size. We
prove in \cref{section:twin-width} that any stable class $\cF$ with bounded twin-width admits a
constant-size equality-based labeling scheme. Although interval graphs do not have bounded
twin-width, some subclasses of interval graphs (\eg unit interval graphs) are known to have
bounded twin-width~\cite{BKTW20}. Below we show that any stable class of interval graphs has bounded
twin-width. 
\end{remark}

\begin{proposition}
\label{prop:stable interval twin-width}
Let $\cF$ be a stable class of interval graphs. Then $\cF$ has bounded twin-width.
\end{proposition}
\begin{proof}
It is known and can easily be deduced from the definition of twin-width, that if $x$ and $y$ are twins (true or false) in a graph $G$, then $\tw(G) = \tw(G-x)$, where $G-x$ is the graph obtained by removing $x$ from $G$.

Let $G$ be an arbitrary graph in $\cF$ and $\ch(\cF) = k$. Let $G'$ be a twin-free induced subgraph of $G$ that is obtained from $G$ by iteratively removing twins. By the above, $\tw(G) = \tw(G')$, and, from \cref{lem:interval-graphs-clique-bounds-chain}, the clique number of $G'$ is at most $4(k+1)^2$. 
Then, by the argument in the proof of
\cref{lem:bounded-clique-number-implies-RL-const}, $G'$ has treewidth at most $4(k+1)^2-1$, and therefore twin-width at most $2^{O(k^2)}$ \cite{BD23}. Since $G$ was chosen arbitrarily, the twin-width of any graph in $\cF$ is bounded by~$2^{O(k^2)}$.
\end{proof}

\subsubsection{Permutation Graphs}
\label{section:permutation graphs}

We will denote by $\prec$ the standard partial order on $\bR^2$, where $(x_1,x_2) \prec (y_1,y_2)$ if $x_1 \leq y_1$ and $x_2 \leq y_2$ and $(x_1,x_2) \neq (y_1,y_2)$. 

The following alternative representation of permutation graphs is well-known (although one should
note that adjacency is sometimes defined as between \emph{incomparable} pairs, instead of comparable
ones -- this is equivalent since the complement of a permutation graph is again a permutation
graph).
\begin{proposition}
	For any permutation graph $G$ there is an injective mapping $\phi : V(G) \to \bR^2$ such that $x,y
	\in V(G)$ are adjacent if and only if $\phi(x),\phi(y)$ are comparable in the partial order $\prec$.
	This mapping also satisfies the property that no two vertices $x,y$ have $\phi(x)_i=\phi(y)_i$ for
	either $i \in [2]$.
\end{proposition}

We will call this the $\bR^2$-representation of $G$. From now on we will identify vertices of $G$
with their $\bR^2$-representation, so that a vertex $x$ of $G$ is a pair $(x_1,x_2) \in \bR^2$.  For
a permutation graph $G$ with fixed $\bR^2$-representation, any $i \in [2]$, and any $t_1 < t_2$, we
define
\[
V_i(t_1,t_2) \define \{ x \in V(G) : t_1 < x_i < t_2 \} \,.
\]

We need the following lemma, which gives a condition that allows us to
increment the chain number.

\begin{lemma}
	\label{lemma:increment chain number}
	For a graph $G$ and a set $A \subset V(G)$, suppose $u,v \in V(G) \setminus A$ are vertices such
	that $u$ has no neighbors in $A$, while $v$ is adjacent to $u$ and to every vertex in $A$. Then
	\[\ch(G[A \cup \{u,v\}]) > \ch(G[A]).\]
\end{lemma}
\begin{proof}
	Suppose $\ch(G[A])=k$ and
	let $\{ a_1, \dotsc, a_k, b_1, \dotsc, b_k\}$ be the vertices such that $a_i,b_j$ are adjacent if
	and only if $i \leq j$. Then set $a_{k+1} = u$ and $b_{k+1} = v$, and verify that
	vertices $\{a_1, \dotsc, a_{k+1}, b_1, \dotsc, b_{k+1}\}$ satisfy \cref{def:chain number}, so
	$\ch(G[A \cup \{u,v\}]) \geq k+1 > \ch(G[A])$.
\end{proof}

\newcommand{\twH}{H^{\mathrm{tw}}}
A bipartite graph $G = (X,Y,E)$ is called a \emph{chain graph} if it is an induced subgraph of a half-graph.  Chain graphs are exactly the $2K_2$-free bipartite graphs,
where $2K_2$ is the disjoint union of two edges.

\begin{proposition}
	\label{prop:twinned half-graph}
	For any $t \in \bR$, any $\bR^2$-representation of a permutation graph $G$, and for each $i \in [2]$, $G[V_i(-\infty,t),V_i(t,\infty)]$ is a chain graph.
\end{proposition}
\begin{proof}
	Let $V_1(-\infty,t) = \{a^{(1)}, \dotsc, a^{(s)} \}$ and $V_1(t, \infty) = \{b^{(1)}, \dotsc,
	b^{(t)}\}$ where the vertices $\{a^{(i)}\}$ and $\{b^{(i)}\}$ are sorted in increasing order in the
	second coordinate. Since $a^{(i)}_1 < t < b^{(j)}_1$ for every $i,j$, it holds that
	$a^{(i)},b^{(j)}$ are adjacent if and only if $a^{(i)}_2 \leq b^{(j)}_2$. 
	To prove the statement we will show that $G[V_i(-\infty,t),V_i(t,\infty)]$ is $2K_2$-free.
	Suppose, towards a contradiction, that $a^{(i_1)}, a^{(i_2)}, b^{(j_1)}, b^{(j_2)}$ induce a $2K_2$ in the graph,
	where $a^{(i_1)}$ is adjacent to $b^{(j_1)}$ and $a^{(i_2)}$ is adjacent to $b^{(j_2)}$.
	Assume, without loss of generality, that $a_2^{(i_1)} < a_2^{(i_2)}$. Since $a^{(i_2)}$ is adjacent to $b^{(j_2)}$,
	we have that $a_2^{(i_2)} \leq b_2^{(j_2)}$, which together with the previous inequality imply that 
	$a_2^{(i_1)} < b_2^{(j_2)}$, and hence $a^{(i_1)}$ is adjacent to $b^{(j_2)}$, a contradiction.
\end{proof}

Any \emph{proper} subclass of the class of chain graphs has a constant-size adjacency labeling scheme, because
the class of chain graphs is a minimal factorial class (see \cref{section:graph theory}). We give an explicit bound on the size of the labeling
scheme, so that we can get an explicit bound on the size of the labels for permutation graphs.

\begin{figure}
	\centering
	\begin{tikzpicture}[
		scale=1.5,
		empty/.style={fill=black!30},
		]
		\def\mypoints{%
			{}/right/b0/10/1,%
			{a^{(1)}}/below/a1/6/0,%
			{b^{(1)}}/above/b1/8/2,%
			{a^{(2)}}/left/a2/4.25/1,%
			{b^{(2)}}/above/b2/5/3,%
			{a^{(3)}}/left/a3/2/2.3,%
			{b^{(3)}}/above/b3/3/4.8,%
			{a^{(4)}}/left/a4/0/4.3,%
			{b^{(4)}}/above/b4/1/6,%
			{a^{(5)}}/left/a5/-1/5.5%,%
		}
		\foreach \label/\dir/\name/\x/\y in \mypoints {
			\coordinate (\name) at (\x,\y) ;
		}
		
		\foreach \lower/\upper/\left/\right/\label/\filling in {%
			a1/b4/a5/b0//,%
			a1/a2/a1/b1/A_0/,%
			a1/a2/b1/b0/B_0/,%
			a2/b1/a1/b0/B_1/,%
			a1/b1/a2/a1/A_1/,%
			b1/b2/a2/a1/B_2/,%
			b1/b2/a3/a2/A_2/,%
			b2/b3/a3/a2/B_3/,%
			b2/b3/a4/a3/A_3/,%
			b3/b4/a4/a3/B_4/,%
			b3/b4/a5/a4/A_4/,%
			b1/b4/a1/b0//empty,%
			a1/b1/a5/a2//empty,%
			b2/b4/a2/a1//empty,%
			b1/b2/a5/a3//empty,%
			b3/b4/a3/a2//empty,%
			b2/b3/a5/a4//empty%
		} {
			\draw[thick,\filling] let \p1 = (\lower), \p2 = (\upper), \p3 = (\left), \p4 = (\right) 
			in (\x3,\y1) rectangle node {$\label$} (\x4,\y2);
		}
		
		\foreach \label/\dir/\name/\x/\y in \mypoints {
			 \ifnum\pdfstrcmp{\name}{b0}=0
				% Skip this iteration, do nothing for \name = b0
			\else
				\node[vertex,label={[rounded corners=3pt,fill opacity=0,text opacity=1,fill=white,inner sep=2pt,\dir=3pt]$\label$}] at (\x,\y) {};
			\fi
		}
		
	\end{tikzpicture}
	\caption{The permutation graph decomposition.}
	\label{fig:permutation graph}
\end{figure}

\begin{proposition}
\label{prop:chain-labeling-scheme}
Let $\cF$ be a hereditary class of chain graphs of chain number at most $k$.
Then $\cF$ admits an adjacency labeling scheme of size $O(\log k)$.
\end{proposition}
\begin{proof}
Let $G \in \cF$, so that $G \sqsubset H_r^{\dcirc}$ for some $r \in \bN$. Then we can partition
$V(G)$ into independent sets $A$ and $B$, such that
the following holds. There is a total order $\prec$ defined on $V(G) = A \cup B$ such that for $a
\in A$ and $b \in B$, $a,b$ are adjacent if and only if $a \prec b$. Then we may identify each $a
\in A$ and each $b \in B$ with a number in $[n]$, such that the ordering $\prec$ is the natural
ordering on $[n]$.

Let $A_1, \dotsc, A_p \subset [n]$ be the set of (non-empty) maximal intervals such that each $A_i
\subseteq A$, and let $B_1, \dotsc, B_q \subseteq [n]$ be the set of (non-empty) maximal intervals
such that each $B_i \subseteq B$. We claim that $p,q \leq k+1$. Suppose for contradiction that $p
\geq k+2$. Since $A_1, \dotsc, A_p$ are maximal, there exist $b_1, \dotsc, b_{p-1} \in B$ such that
$a_1 < b_1 < a_2 < b_2 < \dotsm < b_{p-1} < a_p$, where we choose $a_i \in A_i$ arbitrarily. But
then $\{a_1, \dotsc, a_{p-1}\}, \{b_1, \dotsc, b_{p-1}\}$ is a witness that $\ch(G) \geq p-1 \geq
k+1$, a contradiction. A similar proof shows that $q \leq k+1$.

We construct adjacency labels for $G$ as follows. To each $x \in A$, assign 1 bit to indicate that
$x \in A$, and append the unique number $i \in [k+1]$ such that $x$ belongs to interval $A_i$. To
each $y \in B$, assign 1 bit to indicate that $y \in B$, and append the unique number $j \in [k+1]$
such that $y \in B_j$. It holds that for $x \in A, y \in B$, $x,y$ are adjacent if and only if
$i \leq j$. Therefore, on seeing the labels for $x$ and $y$, the decoder simply checks that $x \in
A$ and $y \in B$ (or vice versa) and outputs 1 if $i \leq j$.
\end{proof}

\begin{definition}[Permutation Graph Decomposition; see \cref{fig:permutation graph}]
	\label{def:permutation decomposition}
	For a permutation graph $G$ with a fixed $\bR^2$-representation, where $G,\overline G$ are both
	connected, we define the following partition.  Let $a^{(1)}$ be the vertex with minimum $a^{(1)}_2$
	coordinate, and let $b$ be the vertex with maximum $b_2$ coordinate. If $b_1 < a^{(1)}_1$, perform
	the following. Starting at $i=1$, construct the following sequence:
	\begin{enumerate}[label=\textup{(\arabic*)}]
		\item Let $b^{(i)}$ be the vertex with maximum $b^{(i)}_2$ coordinate among vertices with $b^{(i)}_1
		> a^{(i)}_1$.
		\item For $i > 1$, let $a^{(i)}$ be the vertex with minimum $a^{(i)}_1$ coordinate among
		vertices with $a^{(i)}_2 < b^{(i-1)}_2$.
	\end{enumerate}
	Let $\beta$ be the smallest number such that $b^{(\beta+1)}=b^{(\beta)}$ and $\alpha$ the smallest number
	such that $a^{(\alpha+1)}=a^{(\alpha)}$. Define these sets:
	\begin{align*}
		\text{For $2 \leq i \leq \alpha$, define } A_i &\define \{a^{(i+1)}\} \cup
		\left( V_1(a^{(i+1)}_1, a^{(i)}_1) \cap V_2(b^{(i-1)}_2, b^{(i)}_2) \right) \\
		\text{For $2 \leq i \leq \beta$, define } B_i &\define \{b^{(i)}\} \cup
		\left( V_1(a^{(i)}_1, a^{(i-1)}_1) \cap V_2(b^{(i-1)}_2, b^{(i)}_2) \right) \\
		A_1 &\define \{a^{(2)}\} \cup
		\left( V_1(a^{(2)}_1, a^{(1)}_1) \cap V_2(a^{(1)}_2, b^{(1)}_2)\right) \\
		A_0 &\define \{a^{(1)}\} \cup
		\left( V_1(a^{(1)}_1, b^{(1)}_1) \cap V_2(a^{(1)}_2, a^{(2)}_2) \right) \\
		B_1 &\define \{b^{(1)}\} \cup
		\left( V_1(a^{(1)}_1, \infty) \cap V_2(a^{(2)}_2, b^{(1)}_2) \right) \\
		B_0 &\define 
		\left( V_1(b^{(1)}_1, \infty) \cap V_2(a^{(1)}_2, a^{(2)}_2) \right) \,.
	\end{align*}
	If $b_1 > a^{(1)}_1$, define the map $\phi : \bR^2 \to \bR^2$ as $\phi(x) = (-x_1,x_2)$ and apply
	$\phi$ to each vertex in the $\bR^2$-representation of $G$; it is easily seen that the result is an
	$\bR^2$-representation of $\overline G$. Then apply the above process to $\overline G$.
\end{definition}

It is necessary to ensure that $b^{(1)}$ is well-defined, \ie, that the set of points $x$ with $x_1
> a^{(1)}_1$ is non-empty, so that the maximum is taken over a non-empty set.
\begin{proposition}
	If $G$ is connected then there exists $x \in V(G)$ such that $x_1 > a^{(1)}_1$.
\end{proposition}
\begin{proof}
	Suppose otherwise. Then every $x \in V(G)$ distinct from $a^{(1)}$ has $x_2 > a^{(1)}_2$ by
	definition, and $x_1 < a^{(1)}_1$. But then $x$ is not adjacent to $a^{(1)}$. So $a^{(1)}$ has no
	neighbors, contradicting the assumption that $G$ is connected.
\end{proof}

\begin{proposition}
	If $G$ is connected and $b_1 < a^{(1)}_1$, then $b^{(1)} \neq b^{(2)}$.
\end{proposition}
\begin{proof}
	Suppose $b^{(2)}=b^{(1)}$. Then $b^{(1)}_2=b^{(2)}_2$ is maximum among all vertices $x$ with $x_1 >
	a^{(2)}_1$, so $b_1 < a^{(2)}_1$. But all vertices $x$ with $x_1 < a^{(2)}_1$ satisfy $x_2 >
	b^{(1)}_2=b^{(2)}_2$, so they cannot have an edge to $V_1(a^{(2)}_1,\infty)$. Both
	$V_1(a^{(2)}_1,\infty)$ and $V_1(-\infty, a^{(2)}_1)$ are non-empty, so the graph is not connected.
\end{proof}

\begin{proposition}
	\label{prop:permutation partition}
	If $G$ is connected, the sets $\{A_i\}_{i=0}^{\alpha}, \{B_i\}_{i=0}^{\beta}$ form a partition of
	$V(G)$.
\end{proposition}
\begin{proof}
	Let $C = \{ x : x_1 \geq a^{(\alpha)}_1, x_2 \leq b^{(\beta)}_2 \}$.  There are no vertices $x$ with
	$x_1 > a^{(\alpha)}_1$ and $x_2 > b^{(\beta)}_2$, since this would contradict the definition of
	$b^{(\beta)}$; likewise, there are no vertices $x$ with $x_1 < a^{(\alpha)}_1$ and $x_2 <
	b^{(\beta)}_2$, since this would contradict the definition of $a^{(\alpha)}$.  Now suppose $x_1 <
	a^{(\alpha)}_1, x_2 > b^{(\beta)}_2$. Then $x$ has no edge to any vertex $y \in C$. Then the set of
	vertices with $x_1 < a^{(\alpha)}_1, x_2 > b^{(\beta)}_2$ must be empty, otherwise $V(G)$ is
	partitioned into $C, V(G) \setminus C$ where $V(G) \setminus C \neq \emptyset$ has no edges to $C$.
	
	Then we may assume that every vertex $x$ is in $C$; we will show that it belongs to some $A_i$ or
	$B_i$.  We may assume that $x$ has distinct $x_1,x_2$ coordinates from all $a^{(i)}, b^{(i)}$,
	otherwise we would have $x = b^{(i)}$ or $x = a^{(i)}$, so $x$ is an element of some $A_i$ or $B_i$.
	
	Let $i$ be the smallest number such that $x_2 < b^{(i)}_2$. Suppose $i \geq 2$. By definition it
	must be that $x_1 > a^{(i+1)}_1$. If $x_1 > a^{(i-1)}_1$, then $x_2 < b^{(i-1)}_2$ by definition,
	which contradicts the choice of $i$. So it must be that $a^{(i+1)}_1 < x_1 < a^{(i-1)}_1$ and
	$b^{(i-1)}_2 < x_2 < b^{(i)}_2$. The set of points that satisfy this condition is contained in $A_i
	\cup B_i$. Now suppose $i=1$. Again, it must be that $x_1 > a^{(2)}_1$ by definition, and also
	$a^{(1)}_2 < x_2 < b^{(1)}_2$. The points satisfying these conditions are easily seen to be
	partitioned by $A_0, B_0, A_1, B_1$.
\end{proof}

For any subset $A \subset V(G)$ and any two vertices $u,v \in V(G) \setminus A$, we will say that
$u,v$ \emph{cover} $A$ if $u$ has no edge into $A$ and $v$ is adjacent to $u$ and every vertex in
$A$. Then by \cref{lemma:increment chain number}, if $u,v$ cover $A$, then $\ch(G) > \ch(G[A])$.

\begin{proposition}
	\label{prop:permutation main chain}
	If $G$ is connected and $b_1 < a^{(1)}_1$, then
	for each $D \in \{A_i\}^{\alpha}_{i=0} \cup \{B_i\}_{i=0}^{\beta}$, $\ch(G[D]) < \ch(G)$.
\end{proposition}
\begin{proof}
	Each $x \in B_0$ satisfies $x_1 > b^{(1)}_1 > a^{(1)}_1$ and $a^{(1)}_2 < x_2 < a^{(2)}_2 <
	b^{(1)}_2$, so $b^{(1)}$ has no neighbors in $B_0$ while $a^{(1)}$ is adjacent to $b^{(1)}$ and all
	vertices in $B_0$, so $a^{(1)},b^{(1)}$ cover $B_0$.
	
	Each $x \in B_1$ satisfies $x_1 > a^{(1)}_1 > b^{(2)}_1 > a^{(2)}_1$ and $a^{(1)}_2 < x_2 <
	b^{(1)}_2 < b^{(2)}$, so $b^{(2)}$ has no neighbors in $B_1$ and $a^{(2)}$ is adjacent $b^{(2)}$ and
	all vertices in $B_1$, so $a^{(2)}, b^{(2)}$ cover $B_1$.
	
	Each $x \in A_0$ satisfies $b^{(1)}_1 > x_1 \geq a^{(1)}_1 > a^{(2)}_1$ and $x_2 < a^{(2)}_2 <
	b^{(1)}_2$, so $a^{(2)}$ has no neighbors in $A_0$ and $b^{(1)}$ is adjacent to $a^{(2)}$ and all
	vertices in $A_0$, so $a^{(2)}, b^{(1)}$ cover $A_0$.
	
	Next we show that for any $1 \leq i \leq \alpha$, $A_i$ is covered by $a^{(i)}, b^{(i)}$. By
	definition each $x \in A_i$ satisfies $x_1 < a^{(i)}_1 < b^{(i)}_1$ and $a^{(i)}_2 < b^{(i-1)}_2 <
	x_2 < b^{(i)}_2$, so $a^{(i)}$ has no neighbors in $A_i$ while $b^{(i)}$ is adjacent to $a^{(i)}$
	and all vertices in $A_i$.
	
	Finally, we show that for any $2 \leq i \leq \beta$, $B_i$ is covered by $a^{(i)}, b^{(i-1)}$. By
	definition each $x \in B_i$ satisfies $a^{(i)}_1 < x_1 < a^{(i-1)}_1 < b^{(i-1)}_1$ and $a^{(i)}_2 <
	b^{(i-1)}_2 < x_2$, so $b^{(i-1)}$ has no neighbors in $B_i$ while $a^{(i)}$ is adjacent
	to $b^{(i-1)}$ and all vertices in $B_i$.
\end{proof}

\begin{lemma}
	\label{lemma:permutation decomposition}
	Let $G \in \cP$ be any permutation graph. Then one of the following holds:
	\begin{enumerate}[label=\textup{(\arabic*)}]
		\item $G$ is disconnected;
		\item $\overline G$ is disconnected;
		\item There is a partition $V(G) = V_1 \cup \dotsm \cup V_m$ such that:
		\begin{itemize}
			\item $\ch(G[V_i]) < \ch(G)$ for each $i \in [m]$, or $\ch(\overline{G[V_i]}) < \ch(\overline G)$ for each $i \in [m]$;
			\item For each $i \in [m]$, there is a set $J(i) \subset \{ V_t \}_{t \in [m]}$ of at most 4 parts such that
			for each $W \in J(i)$, $G[V_i,W]$ is a chain graph; and
			\item One of the following holds:
			\begin{itemize}
				\item For all $i \in [m]$ and $W \in \{V_t\}_{t \in [m]} \setminus J(i)$, $G[V_i,W]$ is a
				co-biclique; or,
				\item For all $i \in [m]$ and $W \in \{V_t\}_{t \in [m]} \setminus J(i)$, $G[V_i,W]$ is a
				biclique.
			\end{itemize}
		\end{itemize}
	\end{enumerate}
\end{lemma}
\begin{proof}
	Assume $G,\overline G$ are connected. Perform the decomposition of \cref{def:permutation
		decomposition}. We will let $m = \alpha+\beta+2$ and let $V_1, \dotsc, V_m$ be the sets
	$\{A_i\}_{i=0}^\alpha \cup \{B_i\}_{i=0}^\beta$.
	
	\textbf{Case 1:} $b_1 < a^{(1)}_1$. Then $V_1, \dotsc, V_m$ is a partition due to
	\cref{prop:permutation partition}, and $\ch(G[V_i]) < \ch(G), i \in [m]$ holds by \cref{prop:permutation main
		chain}. For $V_i = A_1$ we define the corresponding set $J(i) = \splitaftercomma{\{ A_0, B_0, B_1, B_2\}}$. Since all
	sets $V_i,V_j$ with $i \neq j$ are separated by a horizontal line or a vertical line, it holds by
	\cref{prop:twinned half-graph} that $G[V_i,V_j]$ is a chain graph. Now let $W \notin J(i)$.
	Observe that all $x \in W$ must satisfy $x_1 < a^{(2)}_1$ and $x_2 > b^{(1)}_2$, so $x$ is not
	adjacent to any vertex in $A_1$. So $G[A_1,W]$ is a co-biclique.
	
	Now for $V_i \in \{A_0, B_0, B_1\}$, we let $J(i) = \{ A_0, B_0, A_1, B_1\} \setminus \{V_i\}$. Similar
	arguments as above hold in this case to show that $G[V_i,W]$ is a co-biclique for each
	$W \notin J(i)$.
	
	For $V_i = A_j$ for some $j > 1$, we define $J(i) = \{ B_j, B_{j+1} \}$. For any $W \notin J(i)$
	with $W \neq A_j$, it holds either that all $x \in W$ satisfy $x_1 < a^{(i+1)}_1$ and $x_2 >
	b^{(j)}_2$, or that all $x \in W$ satisfy $x_1 \geq a^{(j)}_1$ and $x_2 \leq b^{(j-1)}_2$; in either
	case $x$ is not adjacent to any vertex in $A_j$, so $G[A_j,W]$ is a co-biclique.
	
	For $V_i = B_j$ for some $j > 1$, we define $J(i) = \{ A_j, A_{j-1} \}$. Similar arguments to the
	previous case show that $G[B_j, W]$ is a co-biclique for each $W \notin J(i), W \neq B_j$. This
	concludes the proof for Case 1.
	
	\textbf{Case 2:} $b_1 > a^{(1)}_1$. In this case we transform the $\bR^2$-representation of $G$
	using $\phi$ to obtain an $\bR^2$-representation of $\overline G$ and apply the arguments above to
	obtain $V_1, \dotsc, V_m$ such that $\ch(\overline G[V_i]) < \ch(\overline G)$ for each $i \in [m]$,
	and each $V_j \in \{V_t\}_{t \in [m]} \setminus J(i)$ satisfies that $\overline{G}[V_i,V_j]$ is a
	co-biclique; then $G[V_i,V_j]$ is a biclique as desired.
\end{proof}

\begin{theorem}
\label{thm:permutation graphs}
Let $\cF$ be a stable subclass of permutation graphs. Then $\cF$ admits a constant-size
equality-based labeling scheme, and hence $\RL(\cF) = O(1)$.
\end{theorem}
\begin{proof}
Since $\cF$ is stable, we have $\ch(\cF) = k$ for some constant $k$.

We apply an argument similar to \cref{lemma:bipartite decomposition}. For any $G \in \cF$, we
construct a decomposition tree where each node is associated with either an induced subgraph of
$G$, or a bipartite induced subgraph of $G$, with the root node being $G$ itself. For each node
$G'$, we decompose $G'$ into children as follows,
\begin{enumerate}
  \item If $G'$ is a chain graph, the node is a leaf node.
  \item If $G'$ is disconnected, call the current node a $D$-node, and let the children $G_1,
    \dotsc, G_t$ be the connected components of $G'$.
  \item If $\overline G'$ is disconnected, call the current node a $\overline D$-node, and let
    $C_1, \dotsc, C_t \subseteq V(G')$ be such that $\overline{G}'[C_i], i \in [t]$ are the connected components of 
    $\overline G'$. Define the children to be $G_i = G[C_i], i \in [t]$.
  \item Otherwise construct $V_1, \dotsc, V_m$ as in \cref{lemma:permutation decomposition} and
    let the children be $G[V_i]$ for each $i \in [m]$ and $G[V_i,V_j]$ for each $i,j$ such that $i \in [m]$ and
    $V_j \in J(i)$. Call this node a $P$-node.
\end{enumerate}

We will show that this decomposition tree has bounded depth.  As in the decomposition for bipartite
graphs, on any leaf-to-root path there cannot be two adjacent $D$-nodes or $\overline D$-nodes. As
in the proof of \cref{cl: Qk-properties}, if $G''$ is associated with a $D$-node and its parent $G'$
is associated with a $\overline D$-node, and $G'''$ is any child of $G''$, then $\ch(G') >
\ch(G''')$. On the other hand, if $G''$ is associated with a $\overline D$-node and its parent is
associated with a $D$-node, then $\ch(\overline{G'}) > \ch(\overline{G'''})$.

Now consider any $P$-node associated with $G'$, with child $G''$. By \cref{lemma:permutation
decomposition}, it holds that either $G''$ is a bipartite induced subgraph of $G'$ that is a chain
graph, or $G''$ has $\ch(G'') < \ch(G')$ or $\ch(\overline{G''}) < \ch(\overline{G'})$.  It is easy
to verify that $\ch(\overline{G}) \leq \ch(G)+1$ for any graph $G$. Now, since every sequence
$G''',G'',G'$ of inner nodes along the leaf-to-root path in the decomposition tree must satisfy
$\ch(G''') < \ch(G')$ or $\ch(\overline{G'''}) < \ch(\overline{G'})$ and $\ch(\overline G) \leq
k+1$, it must be that the depth of the decomposition tree is at most $2(2k+1)$.

Now we construct an equality-based labeling scheme. For a vertex $x$, we construct a label at each
node $G'$ inductively as follows.
\begin{enumerate}
  \item If $G'$ is a leaf node, it is a chain graph with chain number at most $k$. We may assign a
    label of size $O(\log k)$ due to \cref{prop:chain-labeling-scheme}.
  \item If $G'$ is a $D$-node with children $G_1, \dotsc, G_t$, append the pair $(D \eqLabelSep i)$ where
    the equality code $i$ is the index of the child $G_i$ that contains $x$, and recurse on $G_i$.
  \item If $G'$ is a $\overline D$-node with children $G_1, \dotsc, G_t$, append the pair
    $(\overline D \eqLabelSep i)$ where the equality code $i$ is the index of the child $G_i$ that
    contains $x$, and recurse on $G_i$.
  \item If $G'$ is a $P$-node, let $V_1, \dotsc, V_m$ be partition of $V(G')$ as in
    \cref{lemma:permutation decomposition}, and for each $i$ let $J(i)$ be the (at most 4) indices
    such that $G'[V_i,V_j]$ is a chain graph when $j \in J(i)$. Append the tuple
    \[
      (P, b, \ell_1(x),\ell_2(x), \ell_3(x), \ell_4(x) \eqLabelSep i, j_1,  j_2,  j_3, j_4)
    \]
    where $b$ indicates whether all $G'[V_i,V_j]$, $j \notin J(i)$ are bicliques or co-bicliques; the equality code $i$ is the index such that $x \in V_i$, the equality codes $j_1,
    \dotsc, j_4$ are the elements of $J(i)$, and $\ell_s(x)$ is the $O(\log k)$-bits adjacency
    label for $x$ in the chain graph $G'[V_i, V_{j_s}]$. Then, recurse on the child
    $G'[V_i]$.
\end{enumerate}
Given labels for $x$ and $y$, which are sequences of the tuples above, the decoder iterates
through the pairs and performs the following. On pairs $(D, i), (D,j)$ the decoder outputs 0 if $i
\neq j$, otherwise it continues. On pairs $(\overline D, i), (\overline D, j)$, the decoder
outputs 1 if $i \neq j$, otherwise it continues. On tuples
\begin{align*}
      (P, b, \ell_1(x),\ell_2(x), \ell_3(x), \ell_4(x) &\eqLabelSep i, j_1,  j_2,  j_3, j_4) \\
      (P, b, \ell_1(y),\ell_2(y), \ell_3(y), \ell_4(y) &\eqLabelSep i', j'_1,  j'_2,  j'_3, j'_4) \,,
\end{align*}
the decoder continues to the next tuple if $i=i'$. Otherwise,
the decoder outputs 1 if $i \notin \{j'_1, \dotsc, j'_4\}$ and $i' \notin \{j_1, \dotsc, j_4\}$
and $b$ indicates that $G'[V_i,V_j]$ are bicliques for $j \notin J(i)$; it outputs 0 if $b$
indicates otherwise. If $i = j'_s$ and $i' = j_t$ then the decoder outputs the adjacency of $x,y$
using the labels $\ell_t(x), \ell_s(y)$. On any tuple that does not match any of the above
patterns, the decoder outputs 0.

Since the decomposition tree has depth at most $2(2k+1)$, each label consists of $O(k)$ tuples.
Each tuple contains at most $O(\log k)$ prefix bits (since adjacency labels for the chain graph with
chain number at most $k$ have size at most $O(\log k)$) and at most $5$ equality codes. So this is
an $( O(k\log k), O(k) )$-equality-based labeling scheme.

The correctness of the labeling scheme follows from the fact that at any node $G'$, if $x,y$
belong to the same child of $G'$, the decoder will continue to the next tuple. If $G'$ is the
lowest common ancestor of $x,y$ in the decomposition tree, then $x$ and $y$ are adjacent in $G$
if and only if they are adjacent in $G'$. If $G'$ is a $D$- or $\overline D$-node then adjacency
is determined by the equality of $i,j$ in the tuples $(D \eqLabelSep i), (D \eqLabelSep j)$ or $(\overline D \eqLabelSep
i), (\overline D \eqLabelSep j)$. If $G'$ is a $P$-node and $i \notin J(i')$ (equivalently, $i' \notin
J(i)$) then adjacency is determined by $b$. If $i \in J(i')$ (equivalently, $i' \in J(i)$) then $i =
j'_s$ and $i' = j_t$ for some $s,t$, and the adjacency of $x,y$ is equivalent to their adjacency in
$G[V_i, V_{i'}] = G[V_{j'_s}, V_{j_t}]$, which is a chain graph, and it is determined by the
labels $\ell_t(x), \ell_s(y)$.
\end{proof}

\begin{remark}
We get an explicit $O(k \log k)$ bound on the size of the adjacency sketch in terms of the chain
number $k$, due to \cref{prop:eq-label-to-sk};
this explicit bound would not arise from the alternate proof that goes through the twin-width
(proper subclasses of permutation graphs have bounded twin-width \cite{BKTW20}, so we could apply
\cref{thm:twin width}).
\end{remark}

%%%%%%%%%%%%%%%%%%%%%%%%%%%%%%%%%%%%%%%%%%%%
\subsection{Monogenic Bipartite Graphs}
\label{section:bipartite graphs}
%%%%%%%%%%%%%%%%%%%%%%%%%%%%%%%%%%%%%%%%%%%%

As explained in \cref{section:correspondence}, any constant-cost communication problem is equivalent
to the problem of deciding adjacency in a hereditary class of \emph{bipartite} graphs. In this
section we answer \cref{question:main} for the hereditary classes of bipartite graphs which have at
most factorial speed and are defined by a single forbidden induced bipartite subgraph.
The following theorem is a formal restatement of \cref{thm:intro-monogenic} from the introduction:

\begin{restatable}{theorem}{thmbipartite}
\label{thm:bipartite monogenic}\RestateRemark
Let $H$ be a bipartite graph such that the class of $H$-free bipartite graphs is factorial. Then
any hereditary subclass $\cF$ of the $H$-free bipartite graphs has a constant-size PUG if and only
if $\cF$ is stable.
\end{restatable}

To prove this theorem, we require new structural results for some classes of bipartite graphs.
Previous work \cite{All09,LZ17} has shown that a class of $H$-free bipartite graphs is factorial
only when $H$ is an induced subgraph of $P_7, S_{1,2,3}$, or one of the infinite set
$\{F^*_{p,q}\}_{p,q \in \bN}$ (defined in \cref{sec:mono-bip-graph-families}). We construct a new
decomposition scheme for the $F^*_{p,q}$-free graphs whose depth is controlled by the chain number,
and we show that the chain number controls the depth of the decomposition scheme from  \cite{LZ17} for $P_7$-free graphs.

As a result, we get a $\poly(n)$-size universal graph for any stable subclass of the $P_7$-free
bipartite graphs. The $P_7$-free bipartite graphs form a factorial class, but existence of a
$\poly(n)$-size universal graph for this class is not known.
%constructing a $\poly(n)$-size universal graph (as predicted by the Implicit Graph Conjecture) is an open problem \cite{LZ17}. 
We take this as evidence that, independent of randomized communication, the study of stable graph classes 
might allow progress on the IGQ.
%are an interesting special case of the IGQ that might allow progress.

We remark that, under a conjecture of \cite{LZ17}, if our theorem was proved for the classes of
bipartite graphs obtained by excluding only \emph{two} graphs $H_1,H_2$, it would establish that
stability is characteristic of the constant-PUG factorial classes of bipartite graphs that are
obtained by excluding any \emph{finite} set of induced subgraphs.
%%-----------------------------------------

%%%%%%%%%%%%%%%%%%%%%%%%%%%%%%%%%%%%%%%%%%%%
\subsubsection{Decomposition Scheme for Bipartite Graphs}
%%%%%%%%%%%%%%%%%%%%%%%%%%%%%%%%%%%%%%%%%%%%

In this section we define a decomposition scheme for bipartite graphs that we will use to establish
constant-size adjacency sketches for the stable subclasses of factorial monogenic classes of bipartite graph.

\begin{definition}[$(\cQ,k)$-decomposition tree]
\label{Qkd decomposition tree}
Let $G=(X,Y,E)$ be a bipartite graph, $k \geq 2$, and let $\cQ$ be a hereditary class of bipartite
graphs.  A graph $G$ admits a \emph{$(\cQ,k)$-decomposition tree of depth $d$} if there is a tree of
depth $d$ of the following form, with $G$ as the root. Each node of the tree is a bipartite graph
$G' = G[X',Y']$ for some $X' \subseteq X, Y' \subseteq Y$, labelled with either $L, D, \overline D$,
or $P$ as follows
	
	\begin{enumerate}[label=\textup{(\arabic*)}]
		\item $L$ (\emph{leaf node}): The graph $G'$ belongs to $\cQ$.
		
		\item $D$ (\emph{$D$-node}): The graph $G'$ is disconnected. 
		There are sets $X_1', \dotsc, X_t' \subseteq X'$ and $Y_1', \dotsc, Y_t' \subseteq Y'$ 
		such that $G[X_1',Y_1'], \dotsc, G[X_t',Y_t']$ are the connected components of $G'$. 
		The children of this decomposition tree node are $G[X_1',Y_1'], \dotsc, G[X_t',Y_t']$.
		
		\item $\overline D$ (\emph{$\overline{D}$-node}): The graph $\bc{G'}$ is disconnected. 
		There are sets $X_1', \dotsc, X_t' \subseteq X'$ and $Y_1', \dotsc, Y_t' \subseteq Y'$ 
		such that $\bc{G[X_1',Y_1']}, \dotsc,  \bc{G[X_t',Y_t']}$ 
		are the connected components of $\bc{G'}$. 
		The children of this decomposition tree node are $G[X_1',Y_1'], \dotsc, G[X_t',Y_t']$.
		
		\item $P$ (\emph{$P$-node}):
		The vertex set of $G'$ is partitioned into at most $2k$ \emph{non-empty} sets $X_1', X_2', \ldots, X_p' \subseteq X'$ and $Y_1', Y_2', \ldots, Y_q' \subseteq Y'$, 
		where $p \leq k$, $q \leq k$. The children of this decomposition tree node are $G[X_i', Y_j']$, for all $i \in[p]$,  $j \in [q]$.
		We say that the $P$-node $G'$ is \emph{specified} by the partitions $X_1', X_2', \ldots, X_p'$ and 
		$Y_1', Y_2', \ldots, Y_q'$.
	\end{enumerate}
\end{definition}

\begin{lemma}
\label{lemma:bipartite decomposition}
Let $k \geq 2$ and $d \geq 1$ be natural constants, and let $\cQ$ be a class of bipartite graphs that
admits a constant-size equality-based adjacency labeling scheme.
Let $\cF$ be a class of bipartite graphs such that each $G \in \cF$ admits a $(\cQ, k)$-decomposition
tree of depth at most $d$. Then $\cF$ admits a constant-size equality-based adjacency labeling scheme. 
\end{lemma}
\begin{proof}
Let $G=(X,Y,E) \in \cF$. We fix a $(\cQ, k)$-decomposition tree of depth at most $d$ for $G$.
For each node $v$ in the decomposition tree we write $G_v$ for the induced subgraph
of $G$ associated with node $v$.  Each leaf node $v$ has $G_v \in \cQ$. 
For some constants $s$ and $r$, we fix an $(s,r)$-equality-based adjacency labeling scheme for $\cQ$, and for each leaf node $v$, we denote by $\ell'_v$ the function that assigns labels to the vertices of $G_v$ under this scheme.

For each vertex $x$ we will construct a label $\ell(x)$ that consists of a constant number of tuples (as in \cref{rm:conv-form}),
where each tuple contains one prefix of at most two bits, and at most two equality codes.
First, we add to $\ell(x)$ a tuple  $(\alpha(x) \eqLabelSep \mathnormal- )$, where $\alpha(x)=0$ if $x \in X$, and $\alpha(x)=1$
if $x \in Y$. Then we append to $\ell(x)$ tuples defined inductively.
Starting at the root of the decomposition tree, for each node $v$ of the tree where $G_v$ 
contains $x$, we add tuples $\ell_v(x)$ defined as follows.  
Write $X' \subseteq X, Y' \subseteq Y$ for the vertices of $G_v$.
\begin{itemize}
\item If $v$ is a leaf node, then $G_v \in \cQ$, and we define $\ell_v(x) = (L \eqLabelSep \mathnormal-), \ell'_v(x)$.

\item If $v$ is a $D$-node then $G_v$ is disconnected, with sets $X'_1, \dotsc, X'_t \subseteq X',
Y'_1, \dotsc, Y'_t \subseteq Y$ such that the children $v_1, \dotsc, v_t$ are the connected
components $G_v[X'_1, Y'_1], \dotsc, G_v[X'_t, Y'_t]$ of $G_v$. 
We define $\ell_v(x) = (D \eqLabelSep j), \ell_{v_j}(x)$, where $j \in [t]$ is the unique index such that $x$ belongs to the connected component $G_v[X'_j, Y'_j]$, and $\ell_{v_j}(x)$ is the inductively defined label for the child node $v_j$.

\item If $v$ is a $\overline D$-node then $\bc{G_v}$ is disconnected, with sets $X'_1, \dotsc, X'_t \subseteq X',
Y'_1, \dotsc, Y'_t \subseteq Y$ such that $\bc{G_v[X_1',Y_1']}, \dotsc,  \bc{G_v[X_t',Y_t']}$ 
are the connected components of $\bc{G_v}$, and the children $v_1, \dotsc, v_t$ of $v$ are the graphs $G_v[X'_1, Y'_1], \dotsc, G_v[X'_t, Y'_t]$. 
We define $\ell_v(x) = (\overline D \eqLabelSep j), \ell_{v_j}(x)$, where $j \in [t]$ is the unique index such that $x$ belongs to $G_v[X'_j, Y'_j]$, and $\ell_{v_j}(x)$ is the inductively defined label for the child node $v_j$.

\item If $v$ is a $P$-node then let $X'_1, \dotsc, X'_p \subseteq X', Y'_1, \dotsc,
Y'_q \subseteq Y'$ be the partitions of $X',Y'$ with $p,q\le k$. 
For each $(i,j) \in [p]\times[q]$, let $v_{i,j}$ be the child
node of $v$ corresponding to the subgraph $G_v[X'_i,Y'_j]$. 
If $x \in X$, then there is a unique $i \in [p]$ such that $x \in X'_i$, and
we define $\ell_v(x) = (P \eqLabelSep i, q), \ell_{v_{i,1}}(x), \dotsc, \ell_{v_{i,q}}(x)$,
where $\ell_{v_{i,j}}(x)$ is the label assigned to $x$ at node $v_{i,j}$. 
If $x \in Y$, then we define $\ell_v(x) = (P \eqLabelSep i, p), \ell_{v_{1,i}}(x), \dotsc, \ell_{v_{p,i}}(x)$, 
where $i \in [q]$ is the unique index such that $x \in Y'_i$.
\end{itemize}
First, we will estimate the size of the label $\ell(x)$ produced by the above procedure.
For every leaf node $v$, the label $\ell_v(x)$ of $x$ is a tuple consisting of an $s$-bit prefix and $r$ equality codes.
Let $f(i)$ be the maximum number of tuples added to $\ell(x)$ by a node $v$ at level $i$ of the decomposition tree,
where the root node belongs to level $0$.
Then, by construction, $f(i) \leq 1 + k \cdot f(i+1)$ and $f(d-1) = 1$, which implies that the total
number of tuples in $\ell(x)$ does not exceed $f(0) \leq k^d$. 
Since every tuple contains a prefix
with at most $s' = \max\{2, s\}$ bits, and at most $r' = \max\{ 2, r \}$ equality codes, we have that
the label $\ell(x)$ contains a prefix with at most $s' k^d$ bits, and at most $r' k^d$ equality codes.

We will now show how to use the labels to define an equality-based adjacency decoder. 
Let $x$ and $y$ be two arbitrary vertices of $G$. The decoder first checks the first tuples 
$(\alpha(x) \eqLabelSep -)$ and $(\alpha(y) \eqLabelSep -)$ of the labels $\ell(x)$ and $\ell(y)$ respectively, to ensure that
$x,y$ are in different parts of $G$ and outputs 0 if they are not. We may now assume
$x \in X, y \in Y$. The remainder of the labels are of the form $\ell_v(x)$ and $\ell_v(y)$, where
$v$ is the root of the decomposition tree. 
\begin{itemize}
\item If the labels $\ell_v(x),\ell_v(y)$ are of the form $(L \eqLabelSep -), \ell'_v(x)$ and $(L \eqLabelSep -), \ell'_v(y)$, then the decoder simulates the decoder for the labeling scheme for $\cQ$, on inputs
$\ell'_v(x),\ell'_v(y)$, and outputs the correct adjacency value.

\item If the labels $\ell_v(x),\ell_v(y)$ are of the form $(D \eqLabelSep i), \ell_{v_i}(x)$ and
$(D \eqLabelSep j), \ell_{v_j}(y)$, the decoder outputs 0 when $i \neq j$ (\ie, $x,y$ are in different connected
components of $G_v$), and otherwise it recurses on $\ell_{v_i}(x), \ell_{v_i}(y)$.

\item If the labels $\ell_v(x),\ell_v(y)$ are of the form $(\overline D \eqLabelSep i), \ell_{v_i}(x)$ and
$(\overline D \eqLabelSep j), \ell_{v_j}(y)$, the decoder outputs 1 when $i \neq j$ (\ie, $x,y$ are in different connected
components of $\bc{G_v}$ and therefore they are adjacent in $G_v$), and otherwise it
recurses on $\ell_{v_i}(x), \ell_{v_i}(y)$.

\item If the labels $\ell_v(x)$, $\ell_v(y)$ are of the form $(P \eqLabelSep i, q), \ell_{v_{i,1}}(x), \dotsc, \ell_{v_{i,q}}(x)$
and $\splitaftercomma{(P \eqLabelSep j, p), \ell_{v_{1,j}}(y), \dotsc, \ell_{v_{p,j}}(y)}$ the decoder recurses on 
$\ell_{v_{i,j}}(x)$ and $\ell_{v_{i,j}}(y)$.
\end{itemize}
It is routine to verify that the decoder will output the correct adjacency value for $x,y$.
\end{proof}

\begin{remark}[$(\cQ,k)$-tree for general graphs]
A similar decomposition scheme can be used for non-bipartite graph classes; we do this for
permutation graphs in \cref{section:permutation graphs}.
\end{remark}

%%%%%%%%%%%%%%%%%%%%%%%%%%%%%%%%%%%%%%%%%%%%
\subsubsection{Monogenic Classes of Bipartite Graphs}
\label{sec:mono-bip-graph-families}
%%%%%%%%%%%%%%%%%%%%%%%%%%%%%%%%%%%%%%%%%%%%

Let $\cH$ be a \emph{finite} set of bipartite graphs. It is known \cite{All09} that if the class of
$\cH$-free bipartite graphs is at most factorial, then $\cH$ contains a forest and a graph whose bipartite complement is a forest. The converse was conjectured in \cite{LZ17}, where it was verified for monogenic classes of bipartite graphs. More specifically, 
it was shown that, for a colored bipartite graph $H$, the class of $H$-free bipartite graphs is at most factorial if and only if both $H$ and its bipartite complement is a forest.  
It is not hard to show that a colored bipartite graph $H$ is a forest and its bipartite complement is
a forest if and only if $H$ is an induced subgraph of $S_{1,2,3}, P_7$, or one of the graphs $F_{p,q}^*$, $p,q \in \mathbb{N}$ defined below.

\begin{figure}[tbh]
\centering
	\begin{tikzpicture}[xscale=.9,yscale=1.4,semithick]
		\useasboundingbox (-3,-.5) rectangle (18,1.75);
		\begin{scope}
			\foreach \i/\x/\y in {%
					1/0/0,2/0/1,3/1/1,4/1.5/0,5/-1/1,6/-1.5/0,7/-2/1%
				}{
				\node[vertex] (S\i) at (\x,\y) {};
			}
			\foreach \u/\v in {1/2,1/3,1/5,3/4,5/6,6/7} {
				\draw (S\u) to (S\v) ;
			}
			\node at (0,-.5) {$S_{1,2,3}$} ;
		\end{scope}
		\begin{scope}[shift={(6,0)}]
			\foreach \i/\x/\y in {%
					1/-1.5/0,2/-1/1,3/-.5/0,4/0/1,5/.5/0,6/1/1,7/1.5/0
				}{
				\node[vertex] (P\i) at (\x,\y) {};
			}
			\foreach \u/\v in {1/2,2/3,3/4,4/5,5/6,6/7} {
				\draw (P\u) to (P\v) ;
			}
			\node at (0,-.5) {$P_7$} ;
		\end{scope}
		\begin{scope}[shift={(13,0)}]
			\def\p{3}\def\q{5}
			\node[vertex] (c) at (0,1) {} ;
			\begin{scope}[shift={(-.5,0)}]
				\node[vertex] (a) at (0,0) {} ;
				\foreach \i in {1,...,\p}{
					\node[vertex] (a\i) at (-0.5*\i,1) {};
				}
				\foreach \i in {1,...,\p} {
					\draw (a) to (a\i) ;
				}
				\draw (a) to (c);
			\end{scope}
			\begin{scope}[shift={(.5,0)}]
				\node[vertex] (b) at (0,0) {} ;
				\foreach \i in {1,...,\q}{
					\node[vertex] (b\i) at (0.5*\i,1) {};
				}
				\foreach \i in {1,...,\q} {
					\draw (b) to (b\i) ;
				}
				\draw (b) to (c);
			\end{scope}
			\node[vertex] at (4,1) {} ;
			\node at (0,-.5) {$F_{\p,\q}^*$} ;
		\end{scope}
	\end{tikzpicture}
	\caption{The bipartite graphs from \cref{def:S123-P7-Fpq}}
  \label{fig:maximal monogenic bipartite graphs}
\end{figure}

\begin{definition}[$S_{1,2,3}$, $P_7$, $F_{p,q}^*$]
\label{def:S123-P7-Fpq}
See \cref{fig:maximal monogenic bipartite graphs} for an illustration.
\begin{enumerate}[label=\textup{(\arabic*)}]
\item $S_{1,2,3}$ is the (colored) bipartite graph obtained from a star with three leaves by
subdividing one of its edges once and subdividing another edge twice.
\item $P_7$ is the (colored) path on 7 vertices.
\item $F^*_{p,q}$ is the colored bipartite graph with vertex color classes $\{a,b\}$ and
$\{a_{1}, \dotsc, a_{p}\splitaftercomma{, c,} b_{1}, \dotsc, b_{q}, d\}$. The edges are $\{(a, a_{i}) \eqLabelSep i \in[p]\}$, $\{(b,b_{j}) \eqLabelSep j \in [q]\}$, and $(a,c),(b,c)$.
\end{enumerate}
\end{definition}

Combining results due to Allen \cite{All09} (for the $S_{1,2,3}$ and $F^*_{p,q}$ cases) and a result of Lozin \& Zamaraev \cite{LZ17} (for the $P_7$ case), we formally state
\begin{theorem}[\cite{All09,LZ17}]
\label{lem:mongenic-factorial-excludes-P7-S123-Fpq}
Let $H$ be a colored bipartite graph, and let 
$\cF$ be the class of $H$-free bipartite graphs.
If $\cF$ has at most factorial speed, then $\cF$ is a subclass of either the $S_{1,2,3}$-free
bipartite graphs, the $P_7$-free bipartite graphs, or the $F^*_{p,q}$-free bipartite graphs, for some $p,q \in \mathbb{N}$.
\end{theorem}

By the above result, in order to establish \cref{thm:bipartite monogenic}, it suffices to consider
the \emph{maximal} monogenic factorial classes of bipartite graphs defined by the forbidden induced
subgraphs $S_{1,2,3}$, $P_7$, $F_{p,q}^*$.

%%%%%%%%%%%%%%%%%%%%%%%%%%%%%%%%%%%%%%%%%%%%
\subsubsection{\texorpdfstring{$S_{1,2,3}$}{S 1,2,3}-Free Bipartite Graphs}
%%%%%%%%%%%%%%%%%%%%%%%%%%%%%%%%%%%%%%%%%%%%

In this section, we derive \cref{thm:bipartite monogenic} for the class of $S_{1,2,3}$-free bipartite graphs.
It is known that the class of $S_{1,2,3}$-free bipartite graphs has bounded clique-width \cite{LV08},
and hence it has also bounded twin-width \cite{BKTW20}. Therefore, the following theorem
follows immediately from our result for graph classes of bounded twin-width (\cref{thm:twin width}). 

\begin{theorem}
	\label{th:S123 labeling}
	Let $\cF$ be a stable class of $S_{1,2,3}$-free bipartite graphs.
	Then $\cF$ admits a constant-size equality-based adjacency labeling scheme, and hence $\RL(\cF_n) = O(1)$.
\end{theorem}

%%%%%%%%%%%%%%%%%%%%%%%%%%%%%%%%%%%%%%%
\subsubsection{\texorpdfstring{$F^*_{p,q}$}{F p,q}-Free Bipartite Graphs}
%%%%%%%%%%%%%%%%%%%%%%%%%%%%%%%%%%%%%%%
\label{sec:Fpq}

\begin{figure}[!h]
\centering
	\begin{tikzpicture}[xscale=.9,yscale=1.4,every label/.style={font=\scriptsize,label distance=3pt,text=black!50,anchor=base}]
	\def\p{3}\def\q{5}
		\begin{scope}[shift={(0,0)}]
			\node[vertex,label=$c$] (c) at (0,1) {} ;
			\begin{scope}[shift={(-0.5,0)}]
				\node[vertex,label={[yshift=-5pt]below:$a$}] (a) at (0,0) {} ;
				\foreach \i in {1,...,\p}{
					\node[vertex,label={$a_\i$}] (a\i) at (-0.5*\i,1) {};
				}
				\foreach \i in {1,...,\p} {
					\draw (a) to (a\i) ;
				}
				\draw (a) to (c);
			\end{scope}
			\begin{scope}[shift={(0.5,0)}]
				\node[vertex,label={[yshift=-5pt]below:$b$}] (b) at (0,0) {} ;
				\foreach \i in {1,...,\q}{
					\node[vertex,label={$b_\i$}] (b\i) at (0.5*\i,1) {};
				}
				\foreach \i in {1,...,\q} {
					\draw (b) to (b\i) ;
				}
				\draw (b) to (c);
			\end{scope}
			\node[vertex,label={$d$}] at (\p+.75,1) {} ;
			\node at (0,-1) {$F_{\p,\q}^*$} ;
		\end{scope}
		\begin{scope}[shift={(9,0)}]
			\node[vertex,label=$c$] (c) at (0,1) {} ;
			\begin{scope}[shift={(-0.5,0)}]
				\node[vertex,label={[yshift=-5pt]below:$a$}] (a) at (0,0) {} ;
				\foreach \i in {1,...,\p}{
					\node[vertex,label={$a_\i$}] (a\i) at (-0.5*\i,1) {};
				}
				\foreach \i in {1,...,\p} {
					\draw (a) to (a\i) ;
				}
				\draw (a) to (c);
			\end{scope}
			\begin{scope}[shift={(0.5,0)}]
				\node[vertex,label={[yshift=-5pt]below:$b$}] (b) at (0,0) {} ;
				\foreach \i in {1,...,\q}{
					\node[vertex,label={$b_\i$}] (b\i) at (0.5*\i,1) {};
				}
				\foreach \i in {1,...,\q} {
					\draw (b) to (b\i) ;
				}
				\draw (b) to (c);
			\end{scope}
			\node at (0,-1) {$F_{\p,\q}$} ;
		\end{scope}
		\begin{scope}[shift={(18,0)}]
			\def\q{\p}
			\begin{scope}[shift={(-0.25,0)}]
				\node[vertex,label={[yshift=-5pt]below:$a$}] (a) at (0,0) {} ;
				\foreach \i in {1,...,\p}{
					\node[vertex,label={$a_\i$}] (a\i) at (-0.5*\i,1) {};
				}
				\foreach \i in {1,...,\p} {
					\draw (a) to (a\i) ;
				}
			\end{scope}
			\begin{scope}[shift={(0.25,0)}]
				\node[vertex,label={[yshift=-5pt]below:$b$}] (b) at (0,0) {} ;
				\foreach \i in {1,...,\q}{
					\node[vertex,label={$b_\i$}] (b\i) at (0.5*\i,1) {};
				}
				\foreach \i in {1,...,\q} {
					\draw (b) to (b\i) ;
				}
			\end{scope}
			\node at (0,-1) {$T_{\p}$} ;
		\end{scope}
	\end{tikzpicture}
	\caption{The bipartite graphs considered in \cref{sec:Fpq}.}
\end{figure}

In this section, we prove \cref{thm:bipartite monogenic} for classes of $F^*_{p,q}$-free bipartite
graphs by developing a constant-size equality-based adjacency labeling scheme for stable classes of
$F^*_{p,q}$-free bipartite graphs via a sequence of labeling schemes for special subclasses each
generalizing the previous one.

We denote by $F_{p,q}$ the bipartite graph with parts $\{a,b\}$ and $\{c, a_1, \dotsc, a_p, b_1,
\dotsc, b_q\}$, and with edges $(a,c), (b,c), \{(a,a_i) \mid i \in [p] \}, \{(b,b_i) \mid j \in
[q]\}$.  We also denote by $T_p$ the bipartite graph on vertex sets $\{a,b\}, \{a_1, \dotsc, a_p,
b_1, \dotsc, b_p\}$, where $(a,a_i)$ and $(b,b_i)$ are edges for each $i \in [p]$. So $T_p$ is the
disjoint union of two stars with $p+1$ vertices.

\begin{definition}
For $q,s \in \bN$ we denote by $Z_{q,s}$ the bipartite graph $(X,Y,E)$ with 
$|X|=q, |Y|=qs$, where $X = \{x_1, \ldots, x_q \}$, $Y$ is partitioned into $q$
sets $Y = Y_1 \cup \ldots \cup Y_q$ each of size $s$, and for every $i \in [q]$:
\begin{enumerate}[label=\textup{(\arabic*)}]
\item $x_i$ is adjacent to all vertices in $Y_j$ for all $1 \leq j \leq i$, and
\item $x_i$ is adjacent to no vertices in $Y_j$ for all $i < j \leq q$.
\end{enumerate}
\end{definition}
\noindent
Note that $Z_{q,s}$ is obtained from $H^{\circ\circ}_q$ by duplicating every vertex in one of the parts $s-1$ times. In particular, $H^{\circ\circ}_q$ is and induced subgraph of $Z_{q,s}$

\medskip
\noindent
We start with structural results and an equality-based labeling scheme for \emph{one-sided} $T_p$-free bipartite graphs.
A colored bipartite graph $G = (X,Y,E)$ is one-sided $T_p$-free if it does not contain $T_p$
as an induced subgraph such that the centers of both stars belong to $X$. Note that
any $T_p$-free bipartite graph is also a one-sided $T_p$-free graph.

\begin{proposition}
\label{prop:degree order}
Let $G = (X,Y,E)$ be any one-sided $T_p$-free bipartite graph and let $u,v \in X$ satisfy $\deg(u) \leq
\deg(v)$. Then $|N(u) \cap N(v)| > |N(u)|-p$.
\end{proposition}
\begin{proof}
For contradiction, assume $|N(u) \cap N(v)| \leq |N(u)|-p$ so that $|N(u) \setminus N(v)| \geq p$.
Then since $\deg(v) \geq \deg(u)$ it follows that $|N(v) \setminus N(u)| \geq p$. But then $T_p$ is
induced by $\{u,v\}$ and $(N(u)\setminus N(v)) \cup (N(v) \setminus N(u))$.
\end{proof}

\begin{proposition}
\label{prop:intersections}
Suppose $S_1, \dotsc, S_t \subseteq [n]$ each have $|S_i| \geq n-p$ where $n > pt$. Then
\[\Bigg|\bigcap_{j =1}^t S_j\Bigg| \geq n-pt.\]
\end{proposition}
\begin{proof}
Let $R$ be the set of all $i \in [n]$ such that for some $S_j, i \notin S_j$. Then
\[
  |R| \leq \sum_{j=1}^t (n - |S_j|) \leq \sum_{j=1}^t p = pt \,,
\]
so $\left|\bigcap_{j=1}^t S_j\right| \geq n-|R| \geq n-pt$.
\end{proof}

\begin{lemma}
\label{lemma:tp free structure}
Fix any constants $k,q,p$ such that $k \geq qp+1$ and let $G = (X,Y,E)$ be any one-sided $T_p$-free bipartite graph. % where all vertices in $X$ have degree at least $k$.
Then there exists $m \geq 0$ and partitions $X =A_0 \cup  A_1 \cup \ldots \cup A_m$ and
$Y = B_1 \cup \ldots \cup B_m \cup B_{m+1}$, where $A_i \neq \emptyset$, $B_i \neq \emptyset$
for every $i \in [m]$, such that the following hold
\begin{enumerate}[label=\textup{(\arabic*)}]
	\item $|B_i| \geq k$, for all $i \in [m]$.
	\item For every $j \in \{0,1, \ldots, m\}$, every $x \in A_j$ has less than $k$ neighbours in $\bigcup_{i \geq j+1} B_i$.
	\item For every $i,j, 1 \leq i \leq j \leq m$, every $x \in A_j$ has more than $|B_i| - p$ neighbours in $B_i$.
	\item If $m \geq q$, then $Z_{q,k-qp}$ is an induced subgraph of $G$.
\end{enumerate}
\end{lemma}
\begin{proof}
	Let $A_0$ be the set of vertices in $X$ that have less than $k$ neighbours. If $A_0 = X$, then
	$m=0$, $A_0$, and $B_1 = Y$ satisfy the conditions of the lemma.
	Otherwise, we construct the remaining parts of partitions using the following procedure.
	Initialize $X' = X \setminus A_0, Y'=Y$, and $i=1$.
	\begin{enumerate}
		\item Let $a_i$ be a vertex in $X'$ with the least number of neighbours in $Y'$.
		\item Let $B_i$ be the set of all neighbors of $a_i$ in $G[X',Y']$.
		\item Let $A_i$ be the set of vertices in $X'$ with degree less than $k$ in $G[X',Y'\setminus B_i]$.
		Note that $A_i$ contains $a_i$.
		\item $X' \gets X' \setminus A_i$, $Y' \gets Y' \setminus B_i$.
		\item If $X' = \emptyset$, then $B_{i+1} = Y'$, let $m=i$, and terminate the procedure;
		Otherwise increment $i$ and return to step 1.
	\end{enumerate}
	
	\noindent
	Conditions (1) and (2) follow by definition. Next we will prove condition (3) by showing 
	that for every $1 \leq i \leq j \leq m$, every $x \in A_j$ has more than $|B_i| - p$ neighbours in $B_i$.
	Suppose, towards a contradiction, that $|N(x) \cap B_i| \leq |B_i|-p$. Consider $X',Y'$ as in round $i$ of
	the construction procedure, so $B_i$ is the neighbourhood of $a_i$ in $G[X',Y']$. 
	Then $x$ has degree at least that of $a_i$ in $G[X',Y']$, and hence the conclusion holds by 
	Proposition \ref{prop:degree order}.

	Finally, to prove condition (4) we will show that for any $q \leq m$ there exist sets $B'_1 \subseteq B_1, \dotsc, B'_q \subseteq B_q$ so that the vertices $\{a_1, \dotsc, a_q\}$ and the sets $B'_1, \dotsc, B'_q$ induce $Z_{q,k-pq}$.
	First, observe that by construction for every $1 \leq i < j \leq m$, $a_i$ has no neighbours in $B_j$.
	Now, let $i \in [m]$, then by condition (3), for all $i \leq j \leq m$ it holds that $|N(a_j) \cap B_i| > |B_i|-p$.
	Since $|B_i| \geq k > pq$, it holds by Proposition \ref{prop:intersections} that
	\[
		\left|B_i \cap \bigcap_{j=i}^q N(a_j)\right| \geq |B_i| - pq \geq k-pq \,.
	\]
	We define $B'_i = B_i \cap \bigcap_{j=i}^q N(a_j)$.
	Then for each $i \in[m]$ it holds that $a_i$ is adjacent to all vertices in $B'_j$ for all $1 \leq j \leq i$, but $a_i$ is adjacent to no vertices in $B'_j$ for $i < j \leq m$. Hence the vertices $\{a_1, \dotsc, a_q\}$ and the sets $B'_1, \dotsc, B'_q$ induce $Z_{q,k-pq}$, which proves condition (4) and 
	concludes the proof of the lemma.
\end{proof}

\begin{lemma}
\label{lemma:tp free labeling}
Let $p \in \mathbb{N}$ and let $\cT$ be a stable class of one-sided $T_{p}$-free bipartite graphs.
Then $\cT$ admits a constant-size equality-based adjacency labeling scheme, and hence 
$\RL(\cT_n) = O(1)$.
\end{lemma}
\begin{proof}
	Since $\cT$ is stable, it does not contain $\cC^\dcirc$ as a subclass.
	Let $q$ be the minimum number such that $H^{\circ\circ}_q \not\in \cT$,
	and let $G=(X,Y,E)$ be an arbitrary graph from $\cT$.
	
	Let $k = qp+1$ and let $X =A_0 \cup  A_1 \cup \ldots \cup A_m$ and $Y = B_1 \cup \ldots \cup B_m \cup B_{m+1}$ be partitions satisfying the conditions of \cref{lemma:tp free structure}.
	Since $G$ does not contain $H^{\circ\circ}_q$ as an induced subgraph, it holds that $m < q$.
	
	We construct the labels for the vertices of $G$ as follows. 
	For a vertex $x \in X$ we define $\ell(x)$ as a label consisting of several tuples. The first tuple is
	$(0, i \eqLabelSep -)$, where $i \in \{ 0, 1, \ldots, m \}$ is the unique index such that $x \in A_i$.
	This tuple follows by $i$ tuples $(- \eqLabelSep y^j_1, y^j_2, \ldots, y^j_{p_j}), j \in [i]$, where
	$p_j < p$ and $\{y^j_1, y^j_2, \ldots, y^j_{p_j}\}$ are the non-neighbours of $x$ in $B_j$.
	The last tuple of $\ell(x)$ is $(- \eqLabelSep y^{i+1}_1, y^{i+1}_2, \ldots, y^{i+1}_{k'})$, where $k' < k$
	and $y^{i+1}_1, y^{i+1}_2, \ldots, y^{i+1}_{k'}$ are the neighbours of $x$ in $\bigcup_{i \geq j+1} B_i$.
	For a vertex $y \in Y$ we define $\ell(y) = ( 1, i \eqLabelSep y)$, where $i \in [m+1]$ is the unique index such that $y \in B_i$.
	
	Note that, in every label, the total length of prefixes is at most $1 + \ceil{\log m} \leq 1 + \ceil{\log q}$, and
	the total number of equality codes depends only on $p,q$, and $k$, which are constants.
	Therefore, it remains to show that the labels can be used to define an equality-based adjacency decoder.
		
	Given two vertices $x,y$ in $G$ the decoder operates as follows. 
	First, it checks the first prefixes in the first tuples of $\ell(x)$ and $\ell(y)$.
	If they are the same, then $x,y$ belong to the same part in $G$ and the decoder outputs 0.
	Hence, we can assume that they are different. Without loss of generality, let
	$\ell(x) = (0, i \eqLabelSep -)$ and $\ell(y) = ( 1, j \eqLabelSep y)$, so $x \in A_i \subseteq X$ and $y \in B_j \subseteq Y$.
	
	If $j \leq i$, then the decoder compares $y$ with the equality codes $y^j_1, y^j_2, \ldots, y^j_{p_j}$ of the $(j+1)$-th tuple of $\ell(x)$. If $y$ is equal to at least one of them, then $y$ is among the non-neighbours of $x$ in $B_j$ and the decoder outputs 0; otherwise, $x$ and $y$ are adjacent and the decoder outputs 1.
	If $j > i$, then the decoder compares $y$ with the equality codes $y^{i+1}_1, y^{i+1}_2, \ldots, y^{i+1}_{k'}$ of the last tuples of $\ell(x)$, and
	if $y$ is equal to at least one of them, then $y$ is among the neighbours of $x$ in 
	$\bigcup_{i \geq j+1} B_i$ and the decoder outputs 1; otherwise, $x$ and $y$ are not adjacent and 
	the decoder outputs 0.
\end{proof}

\medskip
\noindent
Next, we develop an equality-based labeling scheme for stable classes of one-sided $F_{p,p}$-free bipartite graphs.
A colored bipartite graph $G = (X,Y,E)$ is one-sided $F_{p,p}$-free if it does not contain $F_{p,p}$
as an induced subgraph such that the part of $F_{p,p}$ of size 2 is a subset of $X$.

\begin{proposition}
\label{prop:degree intersection bound}
Let $G = (X,Y,E)$ be any one-sided $F_{p,p}$-free bipartite graph and let $u,v \in X$ satisfy $\deg(u) \leq \deg(v)$. Then either $N(u) \cap N(v) = \emptyset$ or $|N(u) \cap N(v)| > |N(u)|-p$.
\end{proposition}
\begin{proof}
Suppose that $N(u) \cap N(v) \neq \emptyset$, and for contradiction assume that $|N(u) \setminus
N(v)| \geq p$. Since $\deg(u) \leq \deg(u)$, this means $|N(v) \setminus N(u)| \geq |N(u)
\setminus N(v)| \geq p$. Let $w \in N(u) \cap N(v)$. Then $\{u,v\}$ with $\{w\} \cup (N(v) \setminus
N(u)) \cup (N(u) \setminus N(v))$ induces a graph containing $F_{p,p}$, a contradiction.
\end{proof}

\begin{proposition}
\label{prop:transitive intersections}
Let $G = (X,Y,E)$ be any one-sided $F_{p,p}$-free bipartite graph and let $x,y,z \in X$ satisfy $\deg(x) \geq \deg(y) \geq \deg(z) \geq 2p$. Suppose that $N(y) \cap N(z) \neq \emptyset$. Then
\[
  N(x) \cap N(y) = \emptyset \iff N(x) \cap N(z) = \emptyset \,.
\]
\end{proposition}
\begin{proof}
Since $N(y) \cap N(z) \neq \emptyset$, it holds that $|N(y) \cap N(z)| > |N(z)|-p \geq p$ by
Proposition~\ref{prop:degree intersection bound}.

Suppose that $N(x) \cap N(y) \neq \emptyset$. For contradiction, assume that $N(x) \cap N(y) \cap
N(z) = \emptyset$. Then $|N(y) \setminus N(x)| \geq |N(y) \cap N(z)| > |N(z)|-p \geq p$, which
contradicts $|N(y) \cap N(x)| > |N(y)|-p$.

Now suppose that $N(x) \cap N(y) = \emptyset$. For contradiction, assume that $N(x) \cap N(z) \neq
\emptyset$. Then $|N(x) \cap N(z)| \leq |N(z) \setminus N(y)| < p \leq |N(z)|-p
< |N(x) \cap N(z)|$, a contradiction.
\end{proof}

We will say that a bipartite graph $G = (X,Y,E)$ is \emph{left-disconnected} if there are two
vertices $x,y \in X$ that are in different connected components of $G$. It is \emph{left-connected} otherwise.

\begin{proposition}
\label{prop:highest degree connected}
Let $G = (X,Y,E)$ be any one-sided $F_{p,p}$-free bipartite graph where every vertex in $X$ has degree at least $2p$. Let $x \in X$ have maximum degree of all vertices in $X$. If $G$ is left-connected, then for any $y \in X$ it holds that $|N(y) \cap N(x)| > |N(y)|-p$.
\end{proposition}
\begin{proof}
Let $y \in X$. Since $G$ is left-connected, there is a path from $y$ to $x$. Let $y_0,y_1,\dotsc,y_t$ be
the path vertices in $X$, where $y=y_0,x=y_t$, and $N(y_{i-1}) \cap N(y_i) \neq \emptyset$ for each
$i \in [t]$. By Propositions \ref{prop:transitive intersections} and \ref{prop:degree intersection
bound}, it holds that if $N(y_i) \cap N(x) \neq \emptyset$ then $|N(y_i) \cap N(x)| > |N(y_i)|-p$
and $|N(y_{i-1}) \cap N(x)| > |N(y_{i-1})|-p$. Therefore the conclusion holds, because $N(y_{t-1})
\cap N(x) = N(y_{t-1}) \cap N(y_t) \neq \emptyset$.
\end{proof}

\begin{lemma}
\label{lemma:Fpq structure}
Fix any constants $p,q \geq 1$, let $k = (q+1)p$, and let $G = (X,Y,E)$ be any
connected one-sided $F_{p,p}$-free bipartite graph. 
Then there exists a partition $X = X_0 \cup X_1 \cup X_2$ (where some of the sets can be empty) such that the following hold:
\begin{enumerate}[label=\textup{(\arabic*)}]
	\item $X_0$ is the set of vertices in $X$ that have degree less than $k$.
	\item The induced subgraph $G[X_1,Y]$ is one-sided $T_p$-free.
	\item The induced subgraph $G[X_2,Y]$ is left-disconnected.
	\item For any $r, s$ such that $r < q$ and $p < s \leq k$, if $X_1 \neq \emptyset$ and $Z_{r,s} \sqsubset G[X_2,Y]$, then $Z_{r+1,s-p} \sqsubset G$.
\end{enumerate}
\end{lemma}
\begin{proof}
	Let $X_0$ be the set of vertices in $X$ that have degree less than $k$,
	and let $X' = X \setminus X_0$. If $G[X', Y]$ is left-disconnected, then we define
	$X_1 = \emptyset$ and $X_2 = X'$.
	
	Assume now that $G[X', Y]$ is left-connected. By Proposition \ref{prop:highest degree connected}, the highest-degree vertex $x \in X'$ satisfies $|N(x) \cap N(y)| > |N(y)|-p$ for every $y \in X'$. Define $X_1$ as follows: add the highest-degree vertex $x$ to $X_1$, and repeat until $G[X' \setminus X_1,Y]$ is left-disconnected. Then set $X_2 = X' \setminus X_1$.  Condition 3 holds by definition, so it remains to prove conditions 2 and 4.

	For every $a,b \in X_1$, note that $N(a) \cap N(b) \neq \emptyset$. Suppose for contradiction that
	$T_p \sqsubset G[X_1,Y]$, then there are $a,b \in X_1$ such that $T_p$ is contained in the subgraph induced by the vertices $\{a,b\}$ and $(N(a) \setminus N(b)) \cup (N(b) \setminus N(a))$.  But then adding any $c \in N(a) \cap N(b)$ results in a forbidden copy of induced $F_{p,p}$, a contradiction. This proves condition 2.

Now for any $r, s$ such that $r < q$ and $p < s \leq k$, suppose that $X_1 \neq \emptyset$ and $Z_{r,s} \sqsubset G[X_2,Y]$. Then there are $u_1, \dotsc, u_r \in X_2$ and pairwise disjoint sets $V_1 \subseteq N(u_1), \dotsc, V_r \subseteq N(u_r)$ such that for
each $i$, $|V_i| = s$, for every $1 \leq j \leq i$, $v_i$ is adjacent to all vertices in $V_j$, and
for every $i < j \leq r$, $v_i$ is adjacent to no vertices in $V_j$.

Let $x$ be the vertex in $X_1$ with least degree, so that $x$ was the last vertex to be added to
$X_1$. Then $G[X_2 \cup \{x\},Y]$ is left-connected but $G[X_2,Y]$ is left-disconnected, and $x$ is
the highest-degree vertex of $G[X_2 \cup \{x\},Y]$ in $X_2 \cup \{x\}$. Since $u_1, \dotsc, u_r$ are
in the same connected component of $G[X_2,Y]$, but the graph $G[X_2,Y]$ is disconnected, it must be that there is $z \in X_2$ such that $N(z) \cap N(u_i) = \emptyset$ for all $u_i$. 
It is also the case that $|N(x) \cap N(z)| > |N(z)|-p \geq k-p \geq s-p$ by Proposition \ref{prop:highest degree connected}, since $x$ has the highest degree in $X_2 \cup \{x\}$.

Observe that for each $V_i \subseteq N(u_i)$ it holds that $|N(x) \cap V_i| \geq s - p$ also by
Proposition \ref{prop:highest degree connected}. Set $V'_i = V_i \cap N(x)$ for each $i \in [r]$,
and set $V'_{r+1} = N(x) \cap N(z)$.
Clearly, the graph induced by
$\{u_1, \dotsc, u_r, x\} \cup V'_1 \cup V'_2 \cup \ldots \cup V'_r \cup V'_{r+1}$ contains 
$Z_{r+1,s-p}$ as an induced subgraph.
\end{proof}

We will now use the above structural result to construct a suitable decomposition
scheme for stable one-sided $F_{p,p}$-free bipartite graphs.
Let $p,q \geq 1$ be fixed constants, let $k = (q+1)p$, and let $\cF_{p,q}$ 
be the class of one-sided $F_{p,p}$-free bipartite graphs that do not contain $H^{\circ\circ}_q$
as an induced subgraph. 
Let $G=(X,Y,E) \in \cF_{p,q}$.
Using \cref{lemma:Fpq structure}, we define a decomposition tree $\cT$ for $G$ inductively as follows. 
Let $G_v$ be the induced subgraph of $G$ associated with node $v$ of the decomposition tree and write $X' \subseteq X$, $Y' \subseteq Y$ for its sets of vertices, so $G_v = G[X',Y']$.
Graph $G$ is associated with the root node of $\cT$.
\begin{itemize}
	\item If $G_v$ is one-sided $T_k$-free, terminate the decomposition, so $v$ is a leaf node ($L$-node) of the decomposition tree.
	\item If $G_v$ is disconnected (in particular, if it is left-disconnected), then $v$ is a $D$-node such
	that the children are the connected components of $G_v$.
	\item If $G_v$ is connected and not one-sided $T_k$-free, then $X'$ admits a partition 
	$X' = X_0' \cup X_1' \cup X_2'$ satisfying the condition of \cref{lemma:Fpq structure}.
	Since $G_v$ is connected, $X_0' \cup X_1' \neq \emptyset$. Furthermore, since 
	$G_v$ is not one-sided $T_k$-free, $X_2' \neq \emptyset$. Hence, $v$ is a $P$-node
	with exactly two children $v_1$ and $v_2$, where $G_{v_1} = G[X_0' \cup X_1', Y']$ and 
	$G_{v_2} = G[X_2', Y']$. Observe that 
	\begin{enumerate}
		\item[(1)] $G_{v_1}$ is one-sided $T_k$-free, and therefore $v_1$ is a leaf;
		\item[(2)] $G_{v_2}$ is left-disconnected, and therefore $v_2$ is a $D$-node; furthermore, every vertex $x \in X_2'$ has degree at least $k$ in $G_{v_2}$ (otherwise it would be included in the set $X_0'$).
	\end{enumerate}
\end{itemize}

\begin{proposition}\label{prop:Fpp-decomposition}
	Let $\cQ$ be the class of one-sided $T_k$-free bipartite graphs.
	Then the graphs in $\cF_{p,q}$ admit $(\cQ,2)$-decomposition trees of depth at most $2q$.
\end{proposition}
\begin{proof}
	By definition, the above decomposition scheme produces $(\cQ,2)$-decomposition trees. 
	In the rest of the proof we will establish the claimed bound on the depth of any such tree.
	Suppose, towards a contradiction, that there exists a graph $G=(X,Y,E) \in \cF_{p,q}$
	such that the decomposition tree $\cT$ for $G$ has depth at least $2q+1$.
	Let $\cP = (v_0,v_1, v_2, \ldots, v_s)$ be a leaf-to-root path in $\cT$ of length $s \geq 2q+1$, where
	$v_0$ is a leaf and $v_s$ is the root. Denote by $G_{v_i} = G[X^{i}, Y^{i}]$ the graph
	corresponding to a node $v_i$ in $\cP$.
	By construction, all internal nodes of $\cP$ are either $D$-nodes or $P$-nodes.
	Clearly, the path cannot contain two consecutive $D$-nodes, as any child of a $D$-node
	is a connected graph.
	Furthermore, a unique non-leaf child $v_i$ of a $P$-node is a $D$-node, 
	and every $x \in X^{i}$ has degree at least $k$ in $G_{v_i}$. 
	Consequently, $P$-nodes and $D$-nodes alternated along (the internal part of) $\cP$. 
	
	Let $v_{i-1}, v_i, v_{i+1}, v_{i+2}$ be four internal nodes of $\cP$,
	where $v_{i-1}$ and $v_{i+1}$ are $D$-nodes, and $v_i$ and $v_{i+2}$ are $P$-nodes.
	Recall that, since the parent $v_{i+2}$ of $v_{i+1}$ is a $P$-node, every vertex in $X_{i+1}$ has degree at least $k$ in $G_{v_{i+1}}$.
	Hence, since $G_{v_i}$ is a connected component of $G_{v_{i+1}}$, every vertex in $X^{i} \subseteq X^{i+1}$ also has degree at least $k$.
	Let $X^{i} = X^{i}_0 \cup X^{i}_1 \cup X^{i}_2$ be the partition of $X^{i}$ according to the decomposition rules. Since $X^{i}_0 \cup X^{i}_1 \neq \emptyset$ and $X^i_0 = \emptyset$, we conclude that $X^{i}_1 \neq \emptyset$. Therefore, by \cref{lemma:Fpq structure}, if the $G_{v_{i-1}} = G[X^{i}_2, Y^i]$ contains $Z_{r,s}$ for some $r < q$ and $p < s \leq k$, then $G_{v_i}$ contains $Z_{r+1, s-p}$.
	
	Let $v_t$ be the first $D$-node in $\cP$. Note that $t \leq 2$. Every vertex in $X^t$ has degree at least $k$
	in $G_{v_t}$, and therefore $Z_{1,k} \sqsubset G_{v_t}$. By induction, the above discussion implies that
	for $1 \leq i \leq q-1$, the graph $Z_{1+i,k-ip}$ is an induced subgraph of $G_{v_{t+2i-1}}$.
	Hence, since the length of $\cP$ is at least $2q+1$, we have $H^{\circ\circ}_q = Z_{q,1} \sqsubset Z_{q,k-(q-1)p} \sqsubset G_{v_{t+2q-3}} \sqsubset G$, a contradiction.
\end{proof}

\begin{lemma}
\label{lemma:Fpq labeling scheme}
Let $p \in \mathbb{N}$ and let $\cF$ be a stable class of one-sided $F_{p,p}$-free bipartite graphs.
Then $\cF$ admits a constant-size equality-based adjacency labeling scheme, and hence $\RL(\cF_n) = O(1)$.
\end{lemma}
\begin{proof}
	Since $\cF$ is stable, it does not contain $\cC^\dcirc$ as a subclass.
	Let $q$ be the minimum number such that $H^{\circ\circ}_q \not\in \cF$.
	Let $k = (q+1)p$ and let $\cQ$ be the class of one-sided $T_k$-free bipartite graphs.
	We have that $\cF \subseteq \cF_{p,q}$, and therefore, by \cref{prop:Fpp-decomposition}, the graphs in $\cF$ admit $(\cQ,2)$-decomposition trees of depth at most $2q$.
	Hence, by \cref{lemma:tp free labeling} and \cref{lemma:bipartite decomposition},
	$\cF$ admits a constant-size equality-based adjacency labeling scheme.
\end{proof}

We conclude this section by showing that stable classes of $F^*_{p,p'}$-free graphs admit
constant-size equality-based adjacency labeling schemes. For this we will use the above result for
one-sided $F_{p,p}$-free graphs and the following

\begin{proposition}[\cite{All09}, Corollary 9]
\label{prop:allen F*}
Let $G = (X,Y,E)$ be a $F^*_{p,p}$-free bipartite graph. Then there is a partition $X = X_1 \cup
X_2$ and $Y = Y_1 \cup Y_2$, where $|Y_2| \leq 1$, such that both $G[X_1,Y_1]$ and $\bc{G[X_2,Y_1]}$ are one-sided $F_{p,p}$-free.
\end{proposition}

\begin{theorem}
\label{th:Fpq labeling}
For any constants $p,p' \geq 1$, a stable class $\cF$ of $F^*_{p,p'}$-free bipartite graphs 
admits a constant-size equality-based adjacency labeling scheme, and hence $\RL(\cF_n) = O(1)$.
\end{theorem}
\begin{proof}
	As before, since $\cF$ is stable, it does not contain $\cC^\dcirc$ as a subclass.
	Let $q$ be the minimum number such that $H^{\circ\circ}_q \not\in \cF$, and 
	assume without loss of generality that $p \geq p'$.
	It follows that $\cF$ is a subclass of $(F^*_{p,p}, H^{\circ\circ}_q)$-free bipartite
	graphs. Let $G = (X,Y,E)$ be a member of this class. Let $X=X_1\cup X_2, Y = Y_1 \cup Y_2$ be the partition given by \cref{prop:allen F*}. We assign labels as follows.

We start the label for each vertex with a one-bit prefix indicating whether it is in $X$ or $Y$. We
then append the following labels. For $x \in X$, we use another one-bit prefix that is equal to 1 if $x$ is adjacent to the unique vertex $Y_2$, and 0 otherwise. Then, we use one more one-bit prefix to indicate whether $x \in X_1$ or $x \in X_2$. If $x \in X_1$, complete the label by using the labeling scheme of \cref{lemma:Fpq labeling scheme} for $G[X_1,Y_1]$. If $x \in X_2$, complete the label by using the labeling scheme of \cref{lemma:Fpq labeling scheme} for $\bc{G[X_2,Y_1]}$.

For $y \in Y$, use a one-bit prefix to indicate whether $y \in Y_2$. If $y \in Y_1$ then concatenate the two labels for $y$ obtained from the labeling scheme of \cref{lemma:Fpq labeling scheme} for
$G[X_1,Y_1]$ and $\bc{G[X_2,Y_1]}$.

The decoder first checks if $x,y$ are in opposite parts. Now assume $x \in X, y \in Y$. The decoder
checks if $y \in Y_2$ and outputs the appropriate value using the appropriate prefix from the label of
$x$. Then if $x \in X_1$, it uses the labels of $x$ and $y$ in $G[X_1,Y_1]$; otherwise it uses the
labels of $x$ and $y$ in $\bc{G[X_2,Y_1]}$ and flips the output.
\end{proof}

%%%%%%%%%%%%%%%%%%%%%%%%%%%%%%%%%%%%%%%
\subsubsection{\texorpdfstring{$P_7$}{P7}-Free Bipartite Graphs}
\label{subsection:P7}
%%%%%%%%%%%%%%%%%%%%%%%%%%%%%%%%%%%%%%%

In this section, we prove \cref{thm:bipartite monogenic} for $P_7$-free bipartite graphs
by developing a constant-size equality-based adjacency labeling scheme for stable classes of $P_7$-free bipartite graphs 

In the below definition, for two disjoint sets of vertices $A$ and $B$ we say 
that $A$ is \emph{complete} to $B$ if every vertex in $A$ is adjacent to every vertex in $B$; we also say that $A$ is  \emph{anticomplete} to $B$ if there are no edges between $A$ and $B$.

\begin{definition}[Chain Decomposition]
\label{def: chain decomp}
See \cref{fig:chain decomposition} for an illustration of the chain decomposition.
Let $G=(X,Y,E)$ be a bipartite graph and $k \in \bN$.  We say that $G$ admits a $k$-\textit{\decomp
decomposition} if one of the parts, say $X$, can be partitioned into subsets $A_1, \ldots, A_k, C_1,
\ldots, C_k$ and the other part $Y$ can be partitioned into subsets $B_1, \ldots, B_k, D_1, \ldots,
D_k$ in such a way that:
	\begin{itemize}
		\item For every $i \leq k-1$, the sets $A_i, B_i,C_i,D_i$ are non-empty. For $i=k$, at least one of the sets $A_i, B_i,C_i,D_i$ must be non-empty.
		\item For each $i =1, \ldots, k$, 
		\begin{itemize}
			\item every vertex of $B_i$ has a neighbour in $A_i$; 
			\item every vertex of $D_i$ has a neighbour in $C_i$;
		\end{itemize}
		\item For each $i = 2, \ldots, k-1$, 
		\begin{itemize}
			\item every vertex of  $A_i$ has a non-neighbour in $B_{i-1}$;
			\item	every vertex of $C_i$ has a non-neighbour in $D_{i-1}$;
		\end{itemize}		
		\item For each $i = 1, \ldots, k$, 
		\begin{itemize}
			\item the set $A_i$ is anticomplete to $B_j$ for $j > i$ and
			is complete to $B_j$ for $j < i-1$; 
			\item the set $C_i$ is anticomplete to $D_j$ for $j > i$ and
			is complete to $D_j$ for $j < i-1$;
		\end{itemize}
		\item For each $i = 1, \ldots, k$, 
		\begin{itemize}
			\item the set $A_i$ is complete to $D_j$ for $j < i$, and is anticomplete to
			$D_j$ for $j \geq i$; 
			\item	the set $C_i$ is complete to $B_j$ for $j < i$, and is anticomplete to $B_j$ for $j \geq i$.
		\end{itemize}
	\end{itemize}
\end{definition}

\begin{figure}[ht]
	\begin{center}
		\tikzstyle{small_vertex}=[circle,fill=blue!100,text=white,inner sep=0.4mm,font=\scriptsize]
		\tikzstyle{med_vertex}=[circle,fill=blue!100,text=white,inner sep=0.5mm,font=\small]
		\tikzstyle{med_vertex_b}=[circle,fill=black!100,text=white,inner sep=0.5mm,font=\small]
		\tikzstyle{med_vertex_r}=[circle,fill=BrickReplacement!100,text=white,inner sep=0.5mm,font=\small]
		\tikzstyle{med_vertex_g}=[circle,fill=ForestGreen!100,text=white,inner sep=0.5mm,font=\small]
		\tikzstyle{vertex}=[circle, draw=black, text=white,inner sep=0.8mm]
		\tikzstyle{blue_vertex}=[circle,fill=blue!100,text=white,inner sep=0.8mm]
		\tikzstyle{red_vertex}=[circle,fill=red!100,text=white,inner sep=0.8mm]
		\tikzstyle{point}=[circle,fill=black,inner sep=0.3mm]
		\tikzstyle{spoint}=[circle,fill=black,inner sep=0.2mm]
		\tikzstyle{first part vertex} = [vertex, fill=red!40]
		\tikzstyle{second part vertex} = [vertex, fill=green!80]
		\tikzstyle{hidden vertex} = [vertex, fill=black!10]
		\tikzstyle{minClass}=[circle,draw=black!50,thin,inner sep=1.6mm]
		\tikzstyle{smallMinClass}=[circle,draw=black!50,thin,inner sep=1.4mm]
		\tikzstyle{tinyMinClass}=[circle,draw=black!50,thin,inner sep=1.2mm]
		\tikzstyle{veryTinyMinClass}=[circle,draw=black!50,thin,inner sep=0.4mm]
		\tikzstyle{w_vertex}=[circle,fill=white!100,text=black,inner sep=0.8mm,draw]
		
		\tikzstyle{b_set}=[circle, minimum size=1cm, draw=NotionColor]
		\tikzstyle{r_set}=[circle, minimum size=1cm, draw=BrickReplacement]
		
		\tikzstyle{sb_set}=[circle, minimum size=0.5cm, draw=NotionColor]
		\tikzstyle{sr_set}=[circle, minimum size=0.5cm, draw=BrickReplacement]
		
		\tikzstyle{pset}=[rounded rectangle, thin, draw, minimum width=11.5cm, minimum height=2.5cm]

		\tikzstyle{vset}=[ellipse, thin, draw, minimum width=4.5cm, minimum height=1.5cm]
		\tikzstyle{neigh_set}=[ellipse, thin, draw, minimum width=1.5cm, minimum height=0.5cm]
		
		\tikzstyle{defbox} = [
		draw=blue, very thick,
		rectangle, rounded corners,
		inner sep=10pt
		]
		
		\tikzstyle{selected edge} = [draw,line width=2pt,-,green!80]
		\tikzstyle{contraction edge} = [draw,line width=1.5pt]
		\tikzstyle{hidden edge} = [draw,line width=0.5pt,black!10]
		\tikzstyle{strong edge} = [draw,line width=1pt,-]
		\tikzstyle{path edge} = [draw,gray!90,dashed]
		\tikzstyle{eat edge} = [draw,BrickReplacement,very thick]
		
		\tikzstyle{text label}=[text=black]
		
		%\vskip-5ex
		\def\lshift{-1.8}
		\def\rshift{1.2}
		\def\dshift{-3}
		\def\pshift{2.2} % interval between the parts
		\begin{tikzpicture}[scale=1.2]
			\node[pset, color=black, opacity=0] (U) at (3*\rshift+\pshift/2,0) {};
			\node[pset, color=black, opacity=0] (W) at (3*\rshift+\pshift/2,\dshift) {};
			
			\node[b_set, label=center:\bb{$A_1$}] (A1) at (0,0) {};
			\node[b_set, label=center:\bb{$A_2$}] (A2) at (0+\rshift,0) {};
			\node[b_set, label=center:\bb{$A_3$}] (A3) at (0+2*\rshift,0) {};
			\node[b_set, label=center:\bb{$A_4$}] (A4) at (0+3*\rshift,0) {};
			
			\node[b_set, label=center:\bb{$B_1$}] (B1) at (0,\dshift) {};
			\node[b_set, label=center:\bb{$B_2$}] (B2) at (0+\rshift,\dshift) {};
			\node[b_set, label=center:\bb{$B_3$}] (B3) at (0+2*\rshift,\dshift) {};
			\node[b_set, label=center:\bb{$B_4$}] (B4) at (0+3*\rshift,\dshift) {};
			
			\node[r_set, label=center:\rr{$C_1$}] (C1) at (3*\rshift+\pshift+3*\rshift, 0) {};
			\node[r_set, label=center:\rr{$C_2$}] (C2) at (3*\rshift+\pshift+2*\rshift, 0) {};
			\node[r_set, label=center:\rr{$C_3$}] (C3) at (3*\rshift+\pshift+\rshift, 0) {};
			\node[r_set, label=center:\rr{$C_4$}] (C4) at (3*\rshift+\pshift, 0) {};
			
			\node[r_set, label=center:\rr{$D_1$}] (D1) at (3*\rshift+\pshift+3*\rshift, \dshift) {};
			\node[r_set, label=center:\rr{$D_2$}] (D2) at (3*\rshift+\pshift+2*\rshift, \dshift) {};
			\node[r_set, label=center:\rr{$D_3$}] (D3) at (3*\rshift+\pshift+\rshift, \dshift) {};
			\node[r_set, label=center:\rr{$D_4$}] (D4) at (3*\rshift+\pshift, \dshift) {};

			% neighbour connections
			\foreach \from/\to in {B1/A1,B2/A2,B3/A3,B4/A4,  D1/C1,D2/C2,D3/C3,D4/C4}
			\draw[color=violet,snake=coil, line after snake=9pt, segment aspect=0, segment length=5pt, -{Latex[length=3mm]}, opacity=0.7] 
			(\from) -- (\to);
			
			\node[sb_set, label=center:\footnotesize\bb{$V$}] (X) at (1.5*\rshift,1.6*\dshift) {};
			\node[sb_set, label=center:\footnotesize\bb{$U$}] (Y) at (0,1.6*\dshift) {};
			\foreach \from/\to in {Y/X}
			\draw[color=violet,snake=coil, line after snake=9pt, segment aspect=0, segment length=5pt, -{Latex[length=3mm]}] 
			(\from) -- (\to);
			
			\coordinate [label=right:every \bb{$u \in U$} \rr{has a neighbour in} \bb{$V$}] (def) at (\rshift + 1, 1.6*\dshift-0.02);

			% non-neighbour connections
			\foreach \from/\to in {A2/B1,A3/B2,  C2/D1,C3/D2}
			\draw[dashed, color=ForestGreen, -{Latex[length=3mm]}, opacity=0.7] 
			(\from) -- (\to);
			
			\node[sb_set, label=center:\footnotesize\bb{$V$}] (X) at (1.5*\rshift,1.85*\dshift) {};
			\node[sb_set, label=center:\footnotesize\bb{$U$}] (Y) at (0,1.85*\dshift) {};
			\foreach \from/\to in {Y/X}
			\draw[dashed, color=ForestGreen, -{Latex[length=3mm]}] (\from) -- (\to);
			
			\coordinate [label=right:every \bb{$u \in U$} \rr{has a non-neighbour in} \bb{$V$}] (def) at (\rshift + 1, 1.85*\dshift-0.02);

			% complete connections
			\foreach \from/\to in {A3/B1,A4/B1,A4/B2, C3/D1,C4/D1,C4/D2}
			\draw[color=Magenta, opacity=0.7] (\from) -- (\to);
			
			\node[sb_set, label=center:\footnotesize\bb{$V$}] (X) at (1.5*\rshift,2.1*\dshift) {};
			\node[sb_set, label=center:\footnotesize\bb{$U$}] (Y) at (0,2.1*\dshift) {};
			\foreach \from/\to in {Y/X}
			\draw[color=Magenta] (\from) -- (\to);
			
			\coordinate [label=right:\bb{$U$} \rr{is complete to} \bb{$V$}] (def) at (\rshift + 1, 2.1*\dshift-0.02);

			% complete connections between the parts
			\draw[color=Magenta, opacity=0.7] ([xshift=10pt,yshift=-10pt]A2.center) -- ([xshift=-10pt,yshift=10pt]D1.center);
			
			\foreach \from/\to in {A3/D1, A3/D2}
			\draw[color=Magenta, opacity=0.7] ([xshift=10pt,yshift=-10pt]\from.center) -- ([xshift=-10pt,yshift=10pt]\to.center);
			
			\foreach \from/\to in {A4/D1, A4/D2, A4/D3}
			\draw[color=Magenta, opacity=0.7] ([xshift=10pt,yshift=-10pt]\from.center) -- ([xshift=-10pt,yshift=10pt]\to.center);
			
			\draw[color=Magenta, opacity=0.7] ([xshift=-10pt,yshift=-10pt]C2.center) -- ([xshift=10pt,yshift=10pt]B1.center);
			
			\foreach \from/\to in {C3/B1, C3/B2}
			\draw[color=Magenta, opacity=0.7] ([xshift=-10pt,yshift=-10pt]\from.center) -- ([xshift=10pt,yshift=10pt]\to.center);
			
			\foreach \from/\to in {C4/B1, C4/B2, C4/B3}
			\draw[color=Magenta, opacity=0.7] ([xshift=-10pt,yshift=-10pt]\from.center) -- ([xshift=10pt,yshift=10pt]\to.center);
			
		\end{tikzpicture}
	\end{center}
	\caption{Example of a 4-chain decomposition.}
	\label{fig:chain decomposition}
\end{figure}
\begin{remark}
\label{rem:2-chain decomp}
In the case of a $2$-chain decomposition of a connected $P_7$-free bipartite graphs, we will also
need the fact that every vertex in $A_2$ and every vertex in $A_1$ have a neighbour in common; and
every vertex in $C_2$ and every vertex in $C_1$ have a neighbour in common. This is not stated
explicitly in \cite{LZ17}, but easily follows from a proof in \cite{LZ17}.
Since the neighbourhood of every vertex in $A_1$ lies entirely in $B_1$, the above fact implies
that every vertex in $A_2$ has a neighbour in $B_1$. Similarly, the neighbourhood of every vertex in
$C_1$ lies entirely in $D_1$, and therefore every vertex in $C_2$ has a neighbour in $D_1$.
\end{remark}

\begin{theorem}[\cite{LZ17}]\label{thm:P7_chain_dec}
	Let $G=(X,Y,E)$ be a $P_7$-free bipartite graph such that both $G$ and $\bc{G}$ are connected.
	Then $G$ or $\bc{G}$ admits a $k$-chain decomposition for some $k \geq 2$. 
\end{theorem}

\begin{lemma}
	\label{lemma: P7 decomp G}
	Let $G=G(X,Y,E)$ be a connected $P_7$-free bipartite graph of chain number $c$ that
	admits a $k$-chain decomposition for some $k \geq 2$. Then there exists 
	a partition of $X$ into $p \leq 2(c+1)$ sets $X_1, X_2, \ldots, X_p$, and
	a partition of $Y$ into $q \leq 2(c+1)$ sets $Y_1, Y_2, \ldots, Y_q$ such that, for any $i
  \in [p], j \in [q]$,
	\[
	\ch(G[X_i, Y_j]) <  \ch(G) \,.
	\]
\end{lemma}
\begin{proof}
	Assume, without loss of generality, that $X$ is partitioned into sets $A_1, \ldots, A_k, C_1, \ldots, C_k$ and $Y$ is partitioned into sets $B_1, \ldots, B_k, D_1, \ldots, D_k$ satisfying \cref{def: chain decomp}.
	Since at least one of the sets $A_k, B_k,C_k,D_k$ is non-empty, and
	every vertex in $B_k$ has a neighbour in $A_k$ and every vertex in $D_k$ has a neighbour in $C_k$,
	at least one of $A_k$ and $C_k$ is non-empty. Without loss of generality we assume that $A_k$ is not empty.
	It is straightforward to check by definition that for any vertices $a_2 \in A_2, a_3 \in A_3, \ldots, a_k \in A_k$,
	and $d_1 \in D_1, d_2 \in D_2, \ldots, d_{k-1} \in D_{k-1}$ the subgraph of $G$ induced by 
	$\{ a_2, a_3, \ldots, a_k, d_1, d_2, \ldots, d_{k-1} \}$ is isomorphic to $H^{\circ\circ}_{k-1}$, which implies that $k$ is at most $c+1$.
	We also observe that any path from a vertex in $A_1$ to a vertex in $D_1$
	contains at least 4 vertices, and hence $G$ contains $H^{\circ\circ}_2$ and 
	$\ch(G) \geq 2$. We split our analysis in two cases.
	
	\textbf{Case 1.} $k \geq 3$. We will show that for any 
	$X' \in \{ A_1, \ldots, A_k, C_1, \ldots, C_k \}$ and $Y' \in \{ B_1, \ldots, B_k\splitaftercomma{,} D_1, \ldots, D_k \}$,
	$\ch(G[X',Y']) < \ch(G)$. Since $\ch(G) \geq 2$, the chain number of a biclique is 1, and the chain number of a co-biclique is 0, we need only to verify pairs of sets that can induce a graph which is neither a biclique nor a co-biclique. By \cref{def: chain decomp}, these are the pairs 
	$(A_i, B_i)$, $(C_i, D_i)$  for $i \in [k]$ and $(A_i, B_{i-1})$, $(C_i, D_{i-1})$ for $i \in \{2, \ldots, k\}$. 
	
	We start with the pair $(A_1, B_1)$. Since $D_2$ is anticomplete to $A_1$, and $C_2$ is complete to $B_1$, for any vertex $d_2 \in D_2$ and its neighbour $c_2 \in C_2$ we have that $\ch(G[A_1,B_1]) < \ch(G[A_1 \cup \{ c_2 \}, B_1 \cup \{d_2\}]) \leq \ch(G)$. Similarly, since $D_1$ is complete to all $A_2, A_3, \ldots, A_k$, and $C_1$ is anticomplete to all $B_1, B_2, \ldots, B_k$, addition of a vertex $d_1 \in D_1$ and its neighbour $c_1 \in C_1$ to
	any of the graphs $G[A_i, B_i]$ or $G[A_i, B_{i-1}]$ for $i \in \{2, \ldots, k\}$ strictly increases the chain number of that graph. Symmetric arguments establish the desired conclusion for the pairs of sets
	$(C_i, D_i)$, $i \in [k]$, and $(C_i, D_{i-1})$, $i \in \{2, \ldots, k\}$.
	
	In this case, $A_1, \ldots, A_k, C_1, \ldots, C_k$ and $B_1, \ldots, B_k, D_1, \ldots, D_k$ are the desired partitions of $X$ and $Y$ respectively.

	\textbf{Case 2.} $k = 2$. Assume first that both $A_2$ and $C_2$ are non-empty.
	Let $c_2$ be a vertex in $C_2$, $d_1$ be a neighbour of $c_2$ in $D_1$ (which exists by \cref{rem:2-chain decomp}), and $c_1$
	be a neighbour of $d_1$ in $C_1$. Since $C_2$ is complete to $B_1$ and $D_1$ is
	anticomplete to $A_1$, $\ch(G[A_1,B_1]) < \ch(G[A_1 \cup \{ c_2 \}, B_1 \cup \{d_1\}]) \leq \ch(G)$.
	Similarly, because $C_1$ is anticomplete to both $B_1$ and $B_2$ and $D_1$ is complete to $A_2$,
	we have that $\ch(G[A_2,B_1]) < \ch(G[A_2 \cup \{ c_1 \}, B_1 \cup \{d_1\}]) \leq \ch(G)$
	and $\ch(G[A_2,B_2]) < \ch(G[A_2 \cup \{ c_1 \}, B_2 \cup \{d_1\}]) \leq \ch(G)$.
	Using symmetric arguments we can show that the chain number of each of 
	$G[C_1,D_1]$, $G[C_2,D_1]$, and $G[C_2,D_2]$ is strictly less than the chain number of
	$G$. All other pairs of sets $(X',Y')$, where $X' \in \{ A_1, A_2, C_1, C_2 \}$ and $Y' \in \{ B_1, B_2, D_1, D_2 \}$
	induce either a biclique or a co-biclique, and therefore $\ch(G[X',Y']) < \ch(G)$. In this case, $A_1, A_2, C_1, C_2$ and $B_1, B_2, D_1, D_2$ are the desired partitions of $X$ and $Y$ respectively.
	
	The case when one of $A_2$ and $C_2$ is empty requires a separate analysis. Assume that $A_2 \neq \emptyset$ 
	and $C_2 = \emptyset$. The  case when $A_2 = \emptyset$ and $C_2 \neq \emptyset$ is symmetric 
	and we omit the details.
	Since $C_2$ is empty, $D_2$ is also empty and therefore $A_1, A_2, C_1$ is a partition of $X$,
	and $B_1, B_2, D_1$ is a partition of $Y$. Let $a_2$ be a vertex in $A_2$, $d_1$ be a vertex in $D_1$,
	and $c_1$ be a neighbour of $d_1$ in $C_1$. Let $B_1'$ be the neighbourhood of $a_2$ in $B_1$ and
	let $B_1'' = B_1 \setminus B_1'$. We claim that  $A_1, A_2, C_1$ and $B_1', B_1'', B_2, D_1$
	are the desired partitions of $X$ and $Y$ respectively. 
	All the pairs of sets,
	except $(A_1, B_1')$ and $(A_1, B_1'')$, can be treated as before and we skip the details.
	For $(A_1, B_1')$, we observe that $a_2$ is complete to $B_1'$ and $D_1$ is anticomplete to $A_1$,
	and hence $\ch(G[A_1,B_1']) < \ch(G[A_1 \cup \{ a_2 \}, B_1' \cup \{d_1\}]) \leq \ch(G)$.
	
	To establish the desired property for $(A_1, B_1'')$, we first observe that by \cref{rem:2-chain decomp}
	every vertex in $A_1$ has a neighbour in common with $a_2$, and therefore every vertex in $A_1$ has
	a neighbour in $B_1'$. If $G[A_1, B_1'']$ is edgeless the property holds trivially. 
	Otherwise, let $P \subseteq A_1$ and $Q \subseteq B_1''$ be such that $P \cup Q$ induces a 
	$H^{\circ\circ}_s$ in $G[A_1, B_1'']$, where $s \geq 1$ is the chain number of the latter graph.
	Let $x$ be the vertex in $P$ that has degree 1 in $G[P,Q]$, and let $y$ be a neighbour of $x$ in $B_1'$.
	We claim that $y$ is complete to $P$. Indeed, if $y$ is not adjacent to some $x' \in P$, then
	$x', z, x, y, a_2, d_1, c_1$ would induce a forbidden $P_7$, where $z$ is the vertex in $Q$ that is adjacent
	to every vertex in $P$.
	Consequently, $G[P \cup \{ a_2 \}, Q \cup \{ y \}]$ is isomorphic to $H^{\circ\circ}_{s+1}$, and therefore $\ch(G[A_1, B_1'']) < \ch(G)$.
\end{proof}

For two pairs of numbers $(a,b)$ and $(c,d)$ we write $(a,b) \preceq (c,d)$
if $a \leq c$ and $b \leq d$, and we write $(a,b) \prec (c,d)$ if at least one of the
inequalities is strict.
  
\begin{lemma}
	\label{lemma: P7 decomp GcoG}
	Let $G=G(X,Y,E)$ be a $P_7$-free bipartite graph such that both $G$ and $\bc{G}$ are connected, and
	let $c$ be the chain number of $G$. Then there exists 
	a partition of $X$ into $p \leq 2(c+2)$ sets $X_1, X_2, \ldots, X_p$, and
	a partition of $Y$ into $q \leq 2(c+2)$ sets $Y_1, Y_2, \ldots, Y_q$ such that
	for any $i \in[p]$, $j \in [q]$
	$$
		\left( \ch(G_{i,j}), \ch(\bc{G_{i,j}}) \right) \prec  \left( \ch(G), \ch(\bc{G}) \right),
	$$
	where $G_{i,j} = G[X_i, Y_j]$.
\end{lemma}
\begin{proof}
	It is easy to verify that for any $k \geq 2$, the graph $\bc{H^{\circ\circ}_k}$ contains the half graph $H^{\circ\circ}_{k-1}$, which implies that 
	$\ch(\bc{G}) \leq \ch(G) + 1 = c + 1$. Furthermore, the bipartite complement of a $P_7$ is again $P_7$, 
	and hence the bipartite complement of any $P_7$-free bipartite graph is also $P_7$-free.
	
	By \cref{thm:P7_chain_dec}, $G$ or $\bc{G}$ admits a $k$-chain decomposition for some $k \geq 2$.
	Therefore, by \cref{lemma: P7 decomp G} applied to either $G$ or $\bc{G}$, there exist
	a partition of $X$ into at most $p \leq 2(c+2)$ sets $X_1, X_2, \ldots, X_p$, and
	a partition of $Y$ into at most $q \leq 2(c+2)$ sets $Y_1, Y_2, \ldots, Y_q$ such that
	either $\ch(G_{i,j}) <  \ch(G)$ holds for any $i \in[p]$, $j \in [q]$,
	or $\ch(\bc{G_{i,j}}) <  \ch(\bc{G})$ holds for any $i \in[p]$, $j \in [q]$.
	This together with the fact that the chain number of an induced subgraph 
	of a graph is never larger than the chain number of the graph, implies the lemma.
\end{proof}

We are now ready to specify a decomposition scheme for $P_7$-free bipartite graphs.
Let $G=(X,Y,E)$ be a $P_7$-free bipartite graph of chain number $c$.
Let $\cQ$ be the class consisting of bicliques and co-bicliques.
We define a $(\cQ,2(c+2))$-decomposition tree $\cT$ for $G$ inductively as follows. 
Let $G_v$ be the induced subgraph of $G$ associated with node $v$ of the decomposition tree and write $X' \subseteq X$, $Y' \subseteq Y$ for its sets of vertices, so $G_v = G[X',Y']$.
Graph $G$ is associated with the root node of $\cT$.
\begin{itemize}
	\item If $G_v$ belongs to $\cQ$, then terminate the decomposition, so $v$ is a leaf node ($L$-node) of the decomposition tree.
	
	\item If $G_v$ does not belong to $\cQ$ and is disconnected, then $v$ is a $D$-node such
	that the children are the connected components of $G_v$.
	
	\item If $G_v$ does not belong to $\cQ$, is connected, and $\bc{G_v}$ is disconnected,
	then $v$ is a $\overline D$-node.
	There are sets $X_1', \dotsc, X_t' \subseteq X'$ and $Y_1', \dotsc, Y_t' \subseteq Y'$ 
	such that $\bc{G[X_1',Y_1']}, \dotsc,  \bc{G[X_t',Y_t']}$ 
	are the connected components of $\bc{G_v}$. 
	The children of this node are $G[X_1',Y_1'], \dotsc, G[X_t',Y_t']$.
	
	\item If $G_v$ does not belong to $\cQ$, and neither $G_v$, nor $\bc{G_v}$ is disconnected,
	then $v$ is a $P$-node.
	Let $X_1', X_2, \ldots, X_p'$ be a partition of $X'$ into $p \leq 2(c+2)$ sets,
	and $Y_1', Y_2, \ldots, Y_q'$ be a partition of $Y'$ into $q \leq 2(c+2)$ sets,
	as in \cref{lemma: P7 decomp GcoG}.
	The children of this node are $G[X_i',Y_j']$, $i \in[p]$, $j \in [q]$.
\end{itemize}

\begin{claim}
	\label{cl: Qk-properties}
	Let $\cT$ be a decomposition tree as above, and
	let $G_i = G[X_i, Y_i]$, $i = 1,2,3$, be internal nodes in $\cT$ such that 
	$G_3$ is the parent of $G_2$ which is in turn the parent of $G_1$.
	Then 
	\begin{enumerate}[label=\textup{(\arabic*)}]
		\item one of $G_3, G_2$, and $G_1$ is a $P$-node, or $G_i$ is $\overline D$-node and $G_{i-1}$ is a $D$-node for some $i \in \{3,2\}$;
		\item if $G_3$ is a $\overline D$-node and $G_2$ is a $D$-node, then $\ch(G_1) < \ch(G_3)$.
	\end{enumerate}
\end{claim}
\begin{proof}
	We start by proving the first statement. Observe that every child a $D$-node is a connected graph, and therefore
	it is not a $D$-node. Similarly, the bipartite complement of every child of a $\overline D$-node is a connected graph,
	and therefore a $\overline D$-node cannot have a $\overline D$-node as a child.
	Hence, if none of $G_3, G_2, G_1$ is a $P$-node, either $G_3$ is a $\overline D$-node and therefore $G_2$ is a 
	$D$-node, or $G_3$ is a $D$-node, in which case $G_2$ is a $\overline D$-node and $G_1$ is a $D$-node.
	In both cases we have a pair of parent-child nodes, where the parent is a $\overline D$-node and the child is 
	a $D$-node.
	
	To prove the second statement, let now $G_3$ be a $\overline D$-node and $G_2$ be a $D$-node, \ie, 
	$G[X_2,Y_2]$ is disconnected, while $G[X_3,Y_3]$ is connected, but its bipartite complement is disconnected. 
	Then there are sets $X'_1 \subseteq X_2 \setminus X_1$ and $Y'_1 \subseteq Y_2 \setminus Y_1$ such that
	$G[X_1',Y_1']$ and $G[X_1,Y_1]$ are connected components of $G[X_2,Y_2]$ and at least one of $X'_1,Y'_1$ is non-empty.
	Also, at least one of the sets $X'_2 = X_3 \setminus X_2$ and $Y'_2 = Y_3 \setminus Y_2$ is non-empty,
	and every vertex in $X'_2$ is adjacent in $G$ to every vertex in $Y_2$, and
	every vertex in $Y'_2$ is adjacent in $G$ to every vertex in $X_2$. 
	If exactly one of the sets $X'_1$ and $Y'_1$ is non-empty, say $Y'_1$, then $X_2'$ is also non-empty, as otherwise  $G[X_3,Y_3]$ would be disconnected. Hence, any vertices $x' \in X_2'$ and $y' \in Y'_1$
	can augment any half graph $H^{\circ\circ}_k$ in $G[X_1,Y_1]$ into a half graph $H^{\circ\circ}_{k+1}$.
	Consequently, $\ch(G[X_1,Y_1]) < \ch(G[\{ x'\} \cup X_1, \{ y' \} \cup Y_1]) \leq \ch(G_3)$.
	If both sets $X'_1$ and $Y'_1$ are non-empty, the argument is similar and we omit the details.
\end{proof}
\begin{theorem}
	\label{th:P7 labeling}
	Let $\cF$ be a stable class of $P_7$-free bipartite graphs.
	Then $\cF$ admits a constant-size equality-based adjacency labeling scheme, and hence $\RL(\cF_n) = O(1)$.
\end{theorem}
\begin{proof}
	Since $\cF$ is stable, it does not contain $\cC^\dcirc$ as a subclass.
	Let $c$ be the maximum number such that $H^{\circ\circ}_c \in \cF$,
	and let $G=(X,Y,E)$ be an arbitrary graph from $\cF$.
	
	By the above discussion $G$ admits
	a $(\cQ, 2(c+2))$-decomposition tree, where $\cQ$ is the class consisting of bicliques and co-bicliques, and every $P$-node $G' = G[X',Y']$ is specified
	by the partition $X_1', \ldots, X_p'$ of $X'$ and the partition $Y_1', \ldots, Y_q'$ of $Y'$
	as in \cref{lemma: P7 decomp GcoG}, where $p$ and $q$ are bounded from above by $2(c+2)$.
	
	We claim that the depth of such a decomposition tree is at most $6c$.
	To show this we associate with every node $G'$ the pair $\left( \ch(G'), \ch(\bc{G'}) \right)$
	and we will prove that if the length of the path from  the root $G$ to a node $G'$ is at least
	$6c$, then $\ch(G') \leq 1$ or $\ch(\bc{G'}) \leq 1$, which means that $G'$ is either
	a biclique or a co-biclique, and therefore is a leaf node.
	
	Let $\cP$ be the path from the root to the node $G'$.
	By \cref{cl: Qk-properties} (1), among any three consecutive nodes of the path,
	there exists a $P$-node, or a pair of nodes labeled with $\overline D$ and $D$ respectively
	such that the $\overline D$-node is the parent of the $D$-node.
	In the former case, by \cref{lemma: P7 decomp GcoG}, 
	for the child node $H'$ of the $P$-node $H$ on the path $\cP$, we have 
	$\left( \ch(H'), \ch(\bc{H'}) \right) \prec \left( \ch(H), \ch(\bc{H}) \right)$.
	In the latter case, by \cref{cl: Qk-properties} (2), the child of the $D$-node
	on the path $\cP$ has the chain number strictly less than that of the $D$-node.
	In other words, for every node $H$ in the path $\cP$ and its ancestor $H'$ at distance
	$3$ from $H$, we have  $\left( \ch(H'), \ch(\bc{H'}) \right) \prec \left( \ch(H), \ch(\bc{H}) \right)$.
	
	Now, since for the root node $G$ we have $\left( \ch(G), \ch(\bc{G}) \right) \prec (c, c+1)$,
	if $\cP$ has length at least $6c$, then $\ch(G') \leq 1$ or $\ch(\bc{G'}) \leq 1$, as required. 
	The result now follows from \cref{lemma:bipartite decomposition} and a simple observation
	that class $\cQ$ admits a constant-size equality-based adjacency labeling scheme.
\end{proof}

%%%%%%%%%%%%%%%%%%%%%%%%%%%%%%%%%%%%%%%%%%%
\section{Question III. Equality is Not Complete}
\label{section:equality}
%%%%%%%%%%%%%%%%%%%%%%%%%%%%%%%%%%%%%%%%%%%

Recall the notion of reductions between communication problems from \cref{section:reductions}.
Many questions about constant-cost randomized communication would be answered if one could identify
a \emph{complete} problem for this class of problems under these reductions:

\begin{definition}
A communication problem $\cP$ is \emph{complete} for the class of constant-cost randomized
communication if every constant-cost randomized problem $\cA$ reduces to $\cP$.
\end{definition}

The most natural candidate for a complete problem is the \EQUALITY problem. Note that all of our
results in this paper so far, except for Cartesian products, have used reductions to \EQUALITY.  In
this section we will prove that the \EQUALITY problem is not complete for the class of constant-cost
randomized communication.  To this end we will show that \textsc{$1$-Hamming Distance} does not
reduce to \EQUALITY. 

\begin{theorem}
\label{thm:1hd-to-eq}
  \textsc{$1$-Hamming Distance} does not reduce to \textsc{Equality}.
\end{theorem}

This theorem was also proved independently and concurrently in \cite{HHH23dimfree} with a very
different Fourier-analytic argument.

In our proof, we will employ some results from the literature.  We denote by $H_d$ the
$d$-dimensional hypercube, \ie, the $d$-wise Cartesian product $P_2^{\square d}$ of the single edge. 

\begin{theorem}[\cite{ARSV06}]\label{hypercubeColoring}
	For every $k$ and $\ell \geq 6$, there exists $d_0(k,\ell)$ such that 
	for every $d \geq d_0(k,\ell)$, every edge coloring of $H_d$ with $k$
	colors contains a monochromatic \emph{induced}  cycle of length $2 \ell$.
\end{theorem}

For a graph $G$, its \emph{equivalence covering number} $\eqc(G)$ is the minimum number $k$ such that there exist $k$ equivalence graphs $F_i=(V,E_i), i \in [k]$, whose union $(V, \cup_{i=1}^k E_i)$ coincides with $G$.
Similarly, for a bipartite graph $G=(U,W,E)$,
its \emph{bipartite equivalence covering number}
$\beqc(G)$ is the minimum number $k$ such that there exist $k$ bipartite equivalence graphs $F_i=(U,W,E_i), i \in [k]$, whose union $(U,W, \cup_{i=1}^k E_i)$ coincides with $G$.

\begin{theorem}[\cite{LNP80,Alon86}]\label{equivalenceCoveringOfCycles}
	For every $n \geq 3$, it holds that $\eqc(\overline{C_n}) \geq \log n - 1$.
\end{theorem}

\begin{corollary}\label{bipEquivalenceCoveringOfCycles}
	For every even $n \geq 4$, it holds that $\beqc(\bc{C_n}) \geq \log n -2$.
\end{corollary}
\begin{proof}
	Let $G=(U,W,E)$ be a bipartite graph isomorphic to $\bc{C_n}$ for some $n \geq 4$.
  Suppose towards the contradiction that there exist bipartite equivalence graphs $F_i=(U,W,E_i), i
\in [k]$, whose union $(U,W, \cup_{i=1}^k E_i)$ coincides with $G$ and $k < \log n - 2$. For each $i
\in [k]$, let $S_i$ be the (non-bipartite) equivalence graph obtained from $F_i$ by turning each of its bicliques
into a complete graph (this is done by adding edges between vertices of the same biclique that are
in the same part). Also let $S=(U \cup W, E')$ be the equivalence graph with exactly two cliques,
$U$ and $W$.  Then the union $S \cup \left(\bigcup_{i \in [k]} S_i\right)$ coincides with $\overline{C_n}$,
implying $\eqc(\overline{C_n}) \leq k +1 < \log n - 1$, which is in contradiction with
\cref{equivalenceCoveringOfCycles}.
\end{proof}

For two binary vectors $x,y \in \{0,1\}^t$, we write $x \preceq y$ if $x_i \leq y_i$ for all $i \in [t]$, and we write $x \prec y$ if $x \preceq y$ and $x \neq y$.

We now prove \cref{thm:1hd-to-eq}.

\begin{proof}[Proof of Theorem~\ref{thm:1hd-to-eq}]
	Let  $\cA=(A_n)_{n \in \bN}$ 
	and $\cB=(B_n)_{n \in \bN}$ be the \textsc{$1$-Hamming Distance} problem and the \textsc{Equality} problem respectively. 
	More specifically, $A_n$ is a $2^n \times 2^n$ matrix, where $A_n(i,j) = 1$ if and only if the binary representations of $i$ and $j$ differ on exactly one bit; and $B_n$ is an $n \times n$ matrix, where $B_n(i,j) = 1$ if and only if $i = j$, \ie, $B_n$ is the identity matrix.
	
	By \cref{prop:cc-mat-red}, in order to prove the statement, we will show that 
	$\cA$ does not matrix-reduce to $\QS(\cB)$.
	To do this, we interpret the matrices of $\cA$ and $\cB$ as bipartite adjacency matrices of bipartite graphs.
	In particular, $A_n$ corresponds to the disjoint union of two hypercube graphs of dimension $n$,
	and $B_n$ corresponds to the matching graph, \ie, the graph in which every vertex has degree exactly 1.
	The set $\QS(\cB)$ then corresponds to bipartite adjacency matrices of bipartite equivalence graphs.
	
	Thus, we want to show that there does not exist a constant $t$ such that
	every matrix $A \in \cA$ can be expressed as $A = h(M_1, M_2, \ldots, M_t)$, where $M_1, M_2, \ldots, M_t \in \QS(\cB)$ are bipartite adjacency matrices of some bipartite equivalence graphs and $h: \zo^t \rightarrow \zo$ is some Boolean function.

	Suppose towards the contradiction that such a $t$ exists. 
	Let $k = 2^t$, $\ell = 2^{t+2}$, and $d \geq d_0(k, \ell)$, where $d_0(k, \ell)$ is the function from \cref{hypercubeColoring}. 
	Since $A_d$, the bipartite adjacency matrix of two disjoint $d$-dimensional hypercube graphs, can be expressed as above via $t$ matrices from $\QS(\cB)$, the same is true for the bipartite adjacency matrix of a single $d$-dimensional hypercube graph.
	Let $H_d = (U,W,E)$ be the hypercube graph of dimension $d$, and let 
	$h: \zo^t \rightarrow \zo$ and $M_1, M_2, \ldots, M_t \in \QS(\cB)$ be such that $H_d(u,w) = h(M_1(u,w), M_2(u,w), \ldots, M_k(u,w))$ for every $u \in U$, $w \in W$.
	 
	Define $\kappa: U \times W \rightarrow \{0,1\}^t$ as $\kappa(u,w)_s = M_s(u,w)$
	for every $s \in [t]$, $u \in U$, and $w \in W$.
	Color every edge $(u,w)$ of $H_d$ with $\kappa(u,w)$.
	Since the edges of $H_d$ are colored in at most $k = 2^t$ different colors,
	by \cref{hypercubeColoring}, it contains a monochromatic 
	induced cycle $C = (U', W', E')$ of length $2 \ell = 2^{t+3}$.
	Let $\kappa^* \in \{0,1\}^t$ be the color of the edges of $C$.
	
	\medskip
	\noindent	
	\textbf{Claim 1.} \emph{For all $u \in U'$ and $w \in W'$ that are \emph{not} adjacent in $C$, we have $\kappa^* \prec \kappa(u,w)$.}
		
	\smallskip
	\noindent
	\emph{Proof.} 
	Since every connected component of a bipartite equivalence graph is a biclique,
	it follows that for every $i \in [t]$, 
	$\kappa^*_i = 1$ implies $\kappa(u,w)_i = 1$ for all $u \in U'$ and $w \in W'$.
	Hence, $\kappa^* \preceq \kappa(u,w)$.
	Furthermore, if $u$ and $w$ are not adjacent in $C$,
	we have $\kappa(u,w) \neq \kappa^*$, as otherwise we would have $h(\kappa(u,w)) = h(\kappa^*) = 1$ and hence $u$ and $w$ would be adjacent. \qed
	
	\medskip
	\noindent
	Let now $I \subseteq [t]$ be the index set such that $i \in I$ if and only if 
	$\kappa^*_i = 0$ and there exist $u \in U'$ and $w \in W'$ with $\kappa(u,w)_i = 1$. 
	For every $i \in I$, let $F_i = (U',W', E_i')$, where
	$E_i' = \{ (u,w) ~|~ u \in U', w \in W', \kappa(u,w)_i = 1\}$. 
	Note that each $F_i$ is an induced subgraph of the bipartite equivalence graph defined by the bipartite adjacency matrix $M_i$, and thus it is itself a bipartite equivalence graph.
	By construction and Claim 1, we have that the union $\cup_{i \in I} F_i$
	contains none of the edges of $C$ and contains all non-edges between the sets $U'$ and $W'$, \ie,
	the union coincides with $\bc{C}$. This implies that $\beqc(\bc{C}) \leq |I| \leq t$. However, by \cref{bipEquivalenceCoveringOfCycles}, $\beqc(\bc{C}) \geq \log 2\ell - 2 \geq t+1$, a contradiction.
\end{proof}

\begin{corollary}
\textsc{Equality} is not complete for constant-cost randomized communication.
\end{corollary}

\begin{corollary}
There is no equality-based labeling scheme for the class of induced subgraphs of the hypercube.
\end{corollary}

\section*{Acknowledgments} 

We thank Eric Blais for helpful discussions about this work, and we thank Abhinav Bommireddi, Renato
Ferreira Pinto Jr., Sharat Ibrahimpur, and Cameron Seth for their comments on the presentation of
this article. We thank the anonymous reviewers for their comments.

% \bibliographystyle{alpha}
% \bibliography{references.bib}

\printbibliography

\appendix

\section{Missing Proofs from \texorpdfstring{\cref{section:definitions}}{Section 3}}
\label{section:missing-proofs}

\subsection{Probability Boosting and Derandomization}
\propboosting*
\begin{proof}
Write $s(n) \define \RL(\cF_n)$ and let $G \in \cF_n$. Then there is a distribution $\cS$ over
functions $\sk : V(G) \to \zo^{s(n)}$ and a decoder $D : \zo^{s(n) \times s(n)} \to \zo$ satisfying
the definition of an adjacency sketch. Consider the following adjacency sketch.  Sample $\sk_1,
\dotsc, \sk_k$ independently from $\cS$. To each vertex $x \in V(G)$ assign the label $\sk'(x) =
(\sk_1(x), \dotsc, \sk_k(x))$. On input $(\sk'(x),\sk'(y))$, the decoder will output
$\mathsf{majority}(D(\sk_1(x), \sk_1(y)), \dotsc, D(\sk_k(x),\sk_k(y)))$.

For each $i \in [k]$, let $X_i = 1$ if $D(\sk_i(x),\sk_i(y)) = \ind{(x,y) \in E(G)}$ and $X_i = 0$
otherwise; and let $X = \sum_{i=1}^k X_i$.  Observe that for each $i \in [k], \Ex{X_i} \geq 2/3$,
so $\Ex{X} \geq 3k/2$. Then, by the Chernoff bound, the probability that the decoder fails is
\[
  \Pr{ X \leq k/2 }
  \leq \Pr{ X \leq \Ex{X}/3 }
  \leq e^{-\frac{k}{3}} \,.
\]
This is at most $\delta$ when $k = 3\ln(1/\delta)$.
\end{proof}

\propderandomization*
\begin{proof}
Let $G \in \cF_n$.  Using \cref{prop:boosting} for $\delta = 1/n^2$, we obtain an
adjacency sketch with error probability $\delta$ and size $c(n) = O(\RL(\cF_n) \cdot
\log(1/\delta)) = O(\RL(\cF_n) \cdot \log n)$. For a fixed sketch function $\sk : V(G) \to
\zo^{c(n)}$, write $\Delta_\sk$ for the number of pairs $x,y$ such that the decoder outputs the incorrect
value on input $\sk(x),\sk(y)$. Then, by the union bound, for a randomly chosen sketch function $\sk$,
\[
  \Ex{\Delta_\sk} \leq \delta \binom{n}{2} \leq 1/2n \,.
\]
Then there exists a fixed $\sk$ with $\Delta_\sk \leq 1/2n < 1$ so $\Delta_\sk = 0$; this $\sk$
is an adjacency labeling scheme for $G$. Furthermore, for $\mu = \Ex{\Delta_\sk}$, Markov's inequality gives
\[
  \Pr{ \Delta_\sk > 0 } = \Pr{ \Delta_\sk \geq 1 } = \Pr{ \Delta_\sk \geq 2n \mu}
  \leq \frac{\mu}{n \mu} = 1/2n \,,
\]
so a random $\sk$ function is a (deterministic) adjacency labeling scheme for $G$ with probability
at least $1/2$. Therefore, if there is a randomized algorithm producing the adjacency sketch in time
$\poly(n)$, then from the proof of \cref{prop:boosting} we see that we can produce the random $\sk$
function with a randomized algorithm in time $\poly(n)$; in expectation we  must sample 2 instances of $\sk$
to find one with $\Delta_\sk = 0$. It takes time at most $\poly(n)$ to check that all the labels are
correct, since we must check $n^2$ pairs of vertices.
\end{proof}

\subsection{Lower Bound from \textsc{Greater-Than}}
\label{section:greater-than-lb}
\subsubsection{Proof of the lower bound}

For the proof of the next statement, let us briefly describe the message passing (SMP) model of
communication.  In this model, given (private) inputs $x,y \in [n]$ to problem $f_n : [n] \times [n]
\to \zo$, Alice and Bob use shared randomness to send random messages $A(x),B(y)$ to a third-party
referee, who must output $f_n(x,y)$ with probability at least $2/3$ over the choice of messages. The
complexity of the protocol is $\max_{x,y} \max(|A(x)|, |B(y)|)$. It is known that on domain $[n]$,
the SMP complexity of \textsc{Greater-Than} is $\Theta(\log n)$, in contrast to the two-way
randomized communication complexity, which is $\Theta(\log\log n)$, where the upper bound is due to
\cite{Nis93} (see discussion in \cref{app:lower-bound-greater-than}).

\lowerbound*
\begin{proof}
This follows from the fact that an adjacency sketch for $\cF$ can be used to construct a
communication protocol for \textsc{Greater-Than} in the public-coin SMP model of communication.  The construction is as follows. Let $D$ be the decoder for
an adjacency sketch for $\cF$. Given inputs $x,y \in [n]$, Alice and Bob can compute
$\mathrm{GT}_n(x,y)$ in the SMP model by choosing a graph $G \in \cF$ with $\ch(G)=n$, so there
exist disjoint sets of vertices $\{a_1, \dotsc, a_n\}, \{b_1, \dotsc, b_n\}$ such that $(a_i,b_j)$
are adjacent if and only if $i \leq j$.  Since $\cF$ is hereditary, the induced subgraph $H
\sqsubset G$ on vertices $\{a_1, \dotsc, a_n, b_1, \dotsc, b_n\}$ is in $\cF$. Alice and Bob draw
random sketches $\sk(a_x),\sk(b_y)$ according to the adjacency sketch for $H$, and send them to the
referee, who outputs $D(\sk(a_x),\sk(b_y))$. This communication protocol has complexity at most
$\RL(\cF)$, so by the lower bound on the SMP complexity of \textsc{Greater-Than}, we must have
$\RL(\cF) = \Omega(\log n)$.
\end{proof}

%%%%%%%%%%%%%%%%%%%%%%%%%%%%%%%%%%%%%%%%%%%%%%%%%%%%%
\subsubsection{Bibliographic remark on Greater-Than}
\label{app:lower-bound-greater-than}
%%%%%%%%%%%%%%%%%%%%%%%%%%%%%%%%%%%%%%%%%%%%%%%%%%%%%

Recall the lower bound for Greater-Than:
\begin{theorem}
\label{thm:greater-than-lower-bound}
Any public-coin randomized SMP communication protocol for \textsc{Greater-Than} on domain $[n]$
requires $\Omega(\log n)$ bits of communication.
\end{theorem}

Lower bounds for the \textsc{Greater-Than} problem in various models appear in \cite{KNR99, MNSW98,
Viol15, RS15, ATYY17}.  The above theorem is stated in \cite{KNR99} and \cite{MNSW98}; in the latter
it is also credited to \cite{Smir88}. In \cite{KNR99} the theorem is stated for one-way
\emph{private-coin} communication; the result for public-coin SMP communication follows from the
fact that public-coin protocols for problems with domain size $n$ can save at most $O(\log \log n)$
bits of communication over the private-coin protocol due to Newman's theorem \cite{New91}.
We remark that the tight lower bound of $\Omega(\log\log n)$ for \textsc{Greater-Than} in the two-way
communication model follows from \cref{thm:greater-than-lower-bound} since there is at most an
exponential gap between SMP and two-way communication.

However, as noted in a CSTheory StackExchange question of Nikolov \cite{Nik20}, the complete proof
is not provided in either of \cite{KNR99,MNSW98}. The same lower bound for \emph{quantum}
communication complexity is proved in \cite{ATYY17}, which implies the above result.  A direct proof
for classical communication complexity was suggested as an answer to \cite{Nik20} by Chakrabarti
\cite{Chak20}; we state this direct proof here for completeness and we thank Eric Blais for
communicating this reference to us. We require the \textsc{Augmented-Index} communication problem
and its lower bound from \cite{MNSW98}.

\begin{definition}[Augmented-Index]
In the \textsc{Augmented-Index} communication problem, Alice receives input $x \in \zo^k$ and Bob
receives an integer $i \in [k]$ along with the values $x_j$ for all $j > i$. Bob should output the
value $x_i$.
\end{definition}

\begin{theorem}[\cite{MNSW98}]
\label{thm:index}
Any public-coin randomized one-way communication protocol for the \textsc{Augmented-Index} problem requires
$\Omega(k)$ bits of communication.
\end{theorem}

\begin{proof}[{Proof of Theorem~\ref{thm:greater-than-lower-bound}, \cite{Chak20}}]
Given inputs $x \in \zo^k$ and $i \in [k]$ to the \textsc{Augmented Index} problem, Bob constructs
the string $y \in \zo^k$ where $y_j = x_j$ for all $j > i$ and $y_i = 0$, and $y_j = 1$ for all $j
< i$.  Consider the numbers $a,b \in [2^k]$ where the binary representation of $a$ is $x$, with
bit $k$ being the most significant and bit $1$ the least significant, and the binary representation
of $b$ is $y$, with the bits in the same order. If $x_i = 1$, then since $y_i = 0$ and $y_j = x_j$
for $j > i$, it holds that $b < a$. If $x_i = 0$, then since $y_j = x_j$ for $j \geq i$ and $y_j =
1$ for $j < i$ it holds that $b \geq a$. Therefore, computing \textsc{Greater-Than} on inputs $a,b$
will solve \textsc{Augmented Index}. By Theorem \ref{thm:index}, the communication cost of
\textsc{Greater-Than} for $n = 2^k$ is at least $\Omega(k) = \Omega(\log n)$.
\end{proof}

\subsection{Missing Proof for Equality-Based Labeling}
\label{section:eq-labeling-proof}

\lemmaeqlabeling*
\begin{proof}
We prove the equivalence between items (1) and (3). The equivalence between items (2) and (3) holds
essentially by definition.  Suppose $\cF$ admits an $(s,k)$-equality based labeling scheme where
$s,k$ are constants. For any matrix $M \in \Adj_\cF$, which is the adjacency matrix of $G_M \in
\cF$, two players with inputs $x,y$ can use $s$ calls to the \EQUALITY oracle to send the prefix
$p(x) \in \zo^s$ to the other player. They may then use $k^2$ calls to the \EQUALITY oracle to
compute the entries of the matrix $Q_{x,y}$, from which they can compute the output
$D_{p(x),p(y)}(Q_{x,y})$ of the decoder.

Now, suppose $\Adj_\cF$ reduces to \EQUALITY. By \cref{prop:cc-mat-red}, $\Adj_\cF$ matrix-reduces to $\QS(\EQUALITY)$. Notices that $\QS(\EQUALITY)$ is the set of bipartite adjacency matrices of bipartite equivalence graphs, and any bipartite equivalence graph $G=(X,Y,E)$ can be defined by a pair of functions $a : X \rightarrow [|X|]$ and $b : Y \rightarrow [|Y|]$ such that for every $x \in X$ and $y \in Y$, $x$ and $y$ are adjacent in $G$ if and only if $\ind{a(x)=b(y)}=1$.
This implies that for every $n_1 \times n_2$ matrix $M \in \Adj_\cF$ there is a function $h : \zo^k \to \zo$ and maps $a_i : [n_1] \to [n_1]$, $b_i :[n_2] \to [n_2]$, $i \in [k]$, such that for any $x \in [n_1]$ and $y \in [n_2]$
\begin{equation}
	\label{eq:eqlabeling}
	M(x,y) = h(\ind{a_1(x) = b_1(y)} , \ind{a_2(x) = b_2(y)}, \dotsc, \ind{a_k(x) = b_k(y)}) \,.
\end{equation}
We then obtain a constant-size equality-based labeling scheme for $\cF$ by assigning to each vertex
$x$ of $G_M \in \cF$ with adjacency matrix $M \in \Adj_\cF$ a label consisting of a prefix encoding the function $h$, and the equality codes $a_1(x), \dotsc,a_k(x), b_1(x), \dotsc, b_k(x)$, and defining the decoder as the function which computes
\cref{eq:eqlabeling}.% (note that $h$ is a fixed function which depends only on $\cF$).
\end{proof}

\section{The Lattice of Hereditary Graph Classes}
\label{section:graph theory}

We briefly survey some known facts about the lattice of hereditary graph classes, which may help to
put our new questions about constant-size PUGs and constant-cost communication complexity into
context. In \cref{section:appendix-speed} we describe how the lattice is partitioned into ``layers''
based upon the speed of the class, and in \cref{section:appendix-minimal-families} we describe the
``minimal'' hereditary graph classes at different speeds, and how stability relates to these minimal
classes.

\begin{figure}[tbp]
	\centering
	\scalebox{1}{\begin{tikzpicture}[
			xscale=.52,
			yscale=.75,
			every node/.style={font=\footnotesize},
	]
%		\draw[{ultra thick,black!50,-Stealth[]}] 
%			(-20,-1) -- node[above,sloped,pos=.5,black] {speed} ++(0,12) ;
		
		\foreach \n/\x/\y in {ll/-4/-1,lr/4/-1,ul/-7/10,ur/7/10} 
			{ \coordinate (\n) at (\x,\y) ; }
		\foreach \s in {l,r} { \coordinate (bell-\s) at ($ (l\s)!.45!(u\s) $) ; }
		\foreach \s in {l,r} { \coordinate (abovebell-\s) at ($ (l\s)!.53!(u\s) $) ; }
		\def\height{11}
		\draw[fill=black!5,thick] (ur) -- (lr) -- (ll) -- (ul) -- cycle ;
		
		\foreach \c/\y in {constant/0,polynomial/1,exponential/2,factorial/8.5} {
			\foreach \s in {l,r} { \coordinate (\c-\s) at ($ (l\s)!{(\y+1)/\height}!(u\s) $) ; }
		}
%		\node[scale=1.5] at (0,6) {\color{Dark Gray} factorial layer};
		\def\cpug{0.527}
		
		\coordinate (stable-b) at ($ (exponential-l)!\cpug!(exponential-r) $) ;
		\coordinate (stable-u) at ($ (factorial-l)!\cpug!(factorial-r) $) ;
%		\coordinate (bell-m) at ($ (bell-l)!\cpug!(bell-r) $);
%		\coordinate (abovebell-m) at ($ (abovebell-l)!\cpug!(abovebell-r) $);
		\coordinate (stable-m) at ($ (ul)!\cpug!(ur) $);

		\fill[const-det] (exponential-r) -- (lr) -- (ll) -- (exponential-l) -- cycle ;
		
		\begin{pgfonlayer}{foreground}
			\draw[ultra thick,shorten >=-2em] (stable-b) -- (stable-m);
			\foreach \c/\y/\col in {constant/0/black!30,polynomial/1/black!30,exponential/2/black!30,factorial/7.1/black!5} {
				\node at (0,\y-.5) {\c\ speed};%{\nodetextbeforebg{\col}{\c\ layer}};
				\draw[ultra thick] (\c-l) -- (\c-r) ;
			}
			\node at (0,9.25) {\nodetextbeforebg{black!5}{superfactorial speed}};
			\node at (0,7.1-.5) {\nodetextbeforebg{black!5}{factorial speed}};
		\end{pgfonlayer}
		
		\node[yshift=1em] at ($ (ul)!.5!(stable-m) $) {stable} ;
		\node[yshift=1em] at ($ (ur)!.5!(stable-m) $) {not stable} ;

		% Removed shaded area after conjecture fell.
%		\fill[const-pug-conjecture] (exponential-l) -- (stable-b) -- (stable-u) -- (factorial-l) -- cycle ;
%		\fill[igc-positive-conjecture] (exponential-r) -- (stable-b) -- (stable-u) -- (factorial-r) -- cycle ;
		
		\draw[thick,densely dotted,igc-positive] (stable-b) -- (exponential-r) 
			-- ($ (exponential-r)!.4!(factorial-r)$) 
			--plot[smooth,tension=.8] coordinates 
				{++(0,0) ++(-1.5,1) ++(-3.5,1.7)
						($ (stable-b)!.6!(stable-u) $) }
			-- ($ (stable-b)!.8!(stable-u) $) 
		;
		\draw[thick,densely dotted,const-rand] (exponential-l) -- (stable-b) 
			-- ($ (stable-b)!.55!(stable-u)$) 
			--plot[smooth,tension=.8] coordinates 
				{++(0,0) ++(-2,-.25) ++(-4,.6) ++(-5.5,.8) ($ (exponential-l)!.65!(factorial-l)$)}
			-- ($ (exponential-l)!.75!(factorial-l)$) 
		;
		
		\draw[thick,densely dotted,igc-negative] 
			($(stable-b)!0.93!(stable-u)$) --plot[smooth,tension=.5] coordinates {
%				++(0,0) ++(-.3,.3) ++(-2,.45) 
				++(-2,.1) ++(-4,.25) ++(-6,-.1)
				++(-6,-.5)
				++(-4.5,-1) ++(-2,-1.1) ++(-1,-1) 
%				++(-2,-.5) ++(-.2,-.3)
				++(0,-.9)
			}
			--cycle;
		;
		\draw[thick,densely dotted,pug-negative] 
			($(stable-b)!0.89!(stable-u)$) --plot[smooth,tension=.7] coordinates {
				++(0,0) ++(-1.5,0) ++(-1.5,-0.4)
				++(2,-.4) ++(1.6,0)
				++(0,0)
			}
			--cycle;
		;
		
		\node[align=left,anchor=west] at ($(exponential-r)!.2!(factorial-r)$) {~$\poly(n)$ universal\\graphs (IGQ positive)} ; 
%		\node[align=left,anchor=west] at ($(exponential-r)!.8!(factorial-r)$) {Conjectured (IGC)} ; 
		\node[align=right,anchor=east] at ($(exponential-l)!.375!(factorial-l)$) {const PUG} ; 
%		\node[align=right,anchor=east] at ($(exponential-l)!.85!(factorial-l)$) {Conjectured const PUG\;\\ (\cref{conj:grand})} ; 
		
		\draw[thin] (4.1+2.5,7.55)  node[inner sep=0pt,align=left,anchor=west] 
				{~~no $\poly(n)$ \\~universal\\ graph} -- ++(-4.5,0) ;
		\draw[thin] (-4.6-2.2,7.75) node[align=right,anchor=east] {no const PUG} -- ++(+1.5,0)  ;

		\draw[thick] (ur) -- (lr) -- (ll) -- (ul) -- cycle ;

		\foreach \n/\x/\y in {
					const-PUG/-2.2/5
			} {
			\node[circle,fill=black,inner sep=2.2pt] (\n) at (\x,\y) {};
		}
		\foreach \n/\x/\y in {planar/-2/3.5} {
			\node[circle,fill=black,inner sep=2.2pt,label=-90:{\scriptsize planar graphs}] (\n) at (\x,\y) {};
		}
		\foreach \n/\x/\y in {interval/3/5} {
			\node[circle,fill=black,inner sep=2.2pt,label=-90:{\scriptsize interval graphs}] (\n) at (\x,\y) {};
		}
		
		\draw[mybrace,decorate] (ll -| -5.2,0) -- 
				node[left=5pt,align=right,font=\scriptsize,inner sep=1pt] (subfactorial) {} 
				(exponential-l -| -5.2,0) ;

		\begin{scope}[shift={(-20,4)}]
			\draw[thick,comm-problems] (0,0.5) circle [x radius=5.5,y radius=6.5];
			\draw[thick,const-rand] (0,-1.95) circle [x radius=4,y radius=4];
			\draw[thick,const-det] (0,-3.9+.25) circle [x radius=2.25,y radius=2.25];
			\foreach \l/\y in {communication\\ problems/5.5,randomized\\constant/.9,deter-\\ministic\\constant/-2.7} {
				\node[align=center] at (0,\y) {\l} ;
			}
			
			\foreach \n/\x/\y in {
						det-const-1/-.5/-5,%
						det-const-2/.25/-4.2,%
						mapped-planar/2/-1.2,%
						rand-const/-.2/-.5%
%						mapped-interval/.5/2.8,%
				} {
				\node[circle,fill=black,inner sep=2.2pt] (\n) at (\x,\y) {};
			}
			
		\end{scope}
		
		\draw[mapping] (det-const-1) to[out=0,in=190] node[above] {$\mathfrak F$} (subfactorial.south) ;
		\draw[mapping] (subfactorial.north) to[out=175,in=10] node[above] {$\textsc{Adj}$} (det-const-2);
		\draw[mapping] (planar) to[bend left=5] node[above] {$\textsc{Adj}$} (mapped-planar) ;
		\draw[mapping] (rand-const) to[bend left=15] node[above] {$\mathfrak F$} (const-PUG) ;
%		\draw[mapping] (interval) to[bend right=10] node[above] {$\textsc{Adj}$} (mapped-interval) ;
%		\draw[mapping] (superfac) to[bend right=10] node[above] {$\textsc{Adj}$} (comm-prob) ;
		
%		\node[align=center] at (0,-2.5) {lattice of hereditary\\graph classes} ;
%		\node[align=center] at (-20,-3) {communication\\ problems} ;
	
	\end{tikzpicture}}
\caption{The relationship between layers of the lattice of hereditary graph classes (on the right),
and communication complexity (on the left). The stability condition partitions the factorial layer.
The red regions, indicating the existence of graph classes with no constant-size PUG and no
$\poly(n)$-size universal graph, are due to \cite{HH22,HHH22counter}.}
\label{fig:conjecture}
\end{figure}

The hereditary graph classes form a lattice, since for any two hereditary graph classes $\cF$ and $\cH$,
it holds that $\cF \cap \cH$ and $\cF \cup \cH$ are also hereditary graph classes.  In this section we
review the structure of this lattice, and give some basic results that place the set of constant-PUG
classes within this lattice.  For an illustrated summary of this section, see \cref{fig:graph
theory}.

\subsection{The Speed of Hereditary Graph Families}
\label{section:appendix-speed}

The speed $|\cF_n|$ of a hereditary graph class cannot be arbitrary. Classic results of Alekseev
\cite{Ale92,Ale97}, Bollob{\'a}s \& Thomason~\cite{BT95}, and Scheinerman \& Zito~\cite{SZ94} have
classified some of the possible speeds of hereditary graph classes. Scheinerman \& Zito~\cite{SZ94}
and Alekseev~\cite{Ale97} showed that the four smallest \emph{layers} of hereditary graph classes
are the following:

\begin{enumerate}
	\item The \emph{constant} layer contains classes $\cF$ with $\log |\cF_n| = \Theta(1)$, and hence $|\cF_n| = \Theta(1)$,
	\item The \emph{polynomial} layer contains classes $\cF$ with $\log |\cF_n| = \Theta(\log n)$,
	\item The \emph{exponential} layer contains classes $\cF$ with $\log |\cF_n| = \Theta(n)$,
	\item The \emph{factorial} layer contains classes $\cF$ with $\log |\cF_n| = \Theta(n\log n)$.
\end{enumerate}
The graph classes with \emph{subfactorial} speed (the first three layers) have simple structure
\cite{SZ94,Ale97}. As demonstrated by earlier examples, the factorial layer is substantially richer and
includes many graph classes of theoretical or practical importance. Despite this, no general
characterization is known for them apart from the definition.

The first three slowest layers correspond exactly to constant-cost \emph{deterministic}
communication proplems, under the equivalence between communication problems and hereditary graph
classes described in \cref{section:correspondence}. Scheinerman \cite{Sch99} showed that a
hereditary class $\cF$ admits a constant-size adjacency labeling scheme (\ie a constant-cost
deterministic communication protocol for $\Adj_\cF$) if and only if it belongs to the
\emph{constant, polynomial}, or \emph{exponential} layer.  Such classes have a bounded number of
equivalence classes of vertices, where two vertices $x,y$ are equivalent if their neighborhoods
satisfy $N(x) \setminus \{y\} = N(y) \setminus \{x\}$.  

\begin{proposition}
A communication problem $f$ admits a constant-cost deterministic protocol if and only if $\,\mathfrak{F}(f)$ is in the constant, polynomial, or exponential layer. A hereditary graph class
$\cF$ is in the constant, polynomial, or exponential layer if and only if there is a constant-cost
deterministic protocol for $\Adj_\cF$.
\end{proposition}

On the other hand, adjacency labels for a factorial class must have size $\Omega(\log n)$ since
graphs in the minimal factorial classes can have $\Omega(n)$ equivalence classes of vertices, and
each equivalence class requires a unique label. So there is a jump in label size from $O(1)$ (in the
subfactorial layers) to $\Omega(\log n)$ (in the factorial layers), so that there is no hereditary
graph class with label size between $O(1)$ and $\Omega(\log n)$. The original version of this paper
conjectured that a similar gap between $O(1)$ and $\Omega(\log n)$ should occur for adjacency
\emph{sketching} too; but this is not true, see \cite{HHH22counter}. 

\subsection{Minimal Factorial Families}
\label{section:appendix-minimal-families}

The factorial layer has a set of 9 \emph{minimal} classes, which satisfy the following:
\begin{enumerate}
\item Every factorial class $\cF$ contains at least one minimal class;
\item For each minimal class $\cM$, any  hereditary subclass $\cM' \subset \cM$ has subfactorial
speed.
\end{enumerate}
These classes were identified by Alekseev \cite{Ale97}, and similar results were independently
obtained by Balogh, Bollob{\'a}s, \& Weinreich~\cite{BBW00}.  Each minimal factorial class is either
a class of bipartite graphs, or a class of \emph{co-bipartite} graphs (\ie, complements of bipartite
graphs), or a class of \emph{split} graphs (\ie, graphs whose vertex set can be partitioned into a
clique and an independent set).  Six of the minimal classes are the following:
\begin{itemize}
\item $\cM^{\dcirc}$ is the class of bipartite graphs of degree at most 1. 
\item $\cM^\bcirc$ is the class of graphs whose vertex set can be partitioned into a clique and an independent set
such that every vertex in each of the parts is adjacent to at most one vertex in the other part.
\item $\cM^\dbullet$ is the class of graphs whose vertex set can be partitioned into two cliques
such that every vertex in each of the parts is adjacent to at most one vertex in the other part.
\item $\cL^{\dcirc}$, $\cL^\bcirc, \cL^\dbullet$ are defined similarly to the classes 
$\cM^{\dcirc}$, $\cM^\bcirc$, $\cM^\dbullet$, respectively, with the difference that vertices in
each of the parts are adjacent to all but at most one vertex in the other part.
\end{itemize}
The other three minimal classes motivate our focus on the \emph{stable} factorial classes. They
are defined as follows.
For any $k \in \bN$, recall that the \emph{half-graph} (see \cref{fig:half graphs}) is the bipartite graph $H^{\circ\circ}_k$ with vertex
sets $\{a_1, \dotsc, a_k\}$ and $\{b_1, \dotsc, b_k\}$, where the edges are exactly the pairs
$(a_i,b_j)$ that satisfy $i \leq j$. The \emph{threshold graph} $H^{\bullet\circ}_k$ is the graph
defined the same way, except including all edges $(a_i,a_j)$ where $i \neq j$.  The
\emph{co-half-graph} $H^{\bullet\bullet}_k$ is the graph defined the same way as the threshold graph
but also including all edges $(b_i,b_j)$ for $i \neq j$.  We define the following hereditary
classes:
\begin{align*}
  \cC^\dcirc &\define \cl\{ H_k^\dcirc : k \in \bN \} \,,
  &\cC^\bcirc &\define \cl\{ H_k^\bcirc : k \in \bN \} \,,
  &\cC^\dbullet &\define \cl\{ H_k^\dbullet : k \in \bN \} \,.
\end{align*}

\begin{proposition}[\cite{Ale97}]
The minimal factorial classes are
\[
\cM^\dcirc, \cM^\bcirc, \cM^\dbullet,  \cL^\dcirc, \cL^\bcirc, \cL^\dbullet,
\cC^\dcirc, \cC^\bcirc, \cC^\dbullet \,.
\]
\end{proposition}

\resultsfigure
See \cref{fig:graph theory}.  It is clear from the definitions that the classes $\cC^{\dcirc},
\cC^\bcirc, \cC^\dbullet$ are not stable, while the other minimal classes are. The following
statement is easily proved from \cref{prop:lower bound}.

\begin{fact}
$\cM^\dcirc, \cM^\bcirc, \cM^\dbullet,  \cL^\dcirc, \cL^\bcirc, \cL^\dbullet$ 
admit constant-size equality-based labeling schemes (and therefore
constant-size PUGs), while $\cC^{\dcirc}, \cC^{\bcirc}, \cC^{\dbullet}$ have PUGs of size
$n^{\Theta(1)}$.
\end{fact}

A consequence of Ramsey's theorem is that a hereditary
graph class $\cF$ is stable if and only if it does not include any of $\cC^\dcirc, \cC^\bcirc,
\cC^\dbullet$:
\begin{proposition}
	\label{prop:ch-finite-iff-no-chainlike}
	Let $\cF$ be a hereditary class of graphs.  Then $\cF$ has bounded \chainNum if and only if
	$\cC^{\circ\circ},\cC^{\bullet\circ},\cC^{\bullet\bullet} \not\subseteq \cF$.
\end{proposition}
\ignore{
\begin{proof}
	Let $** \in \{\bullet\bullet,\bullet\circ,\circ\circ\}$ and suppose $\cC^{**} \subseteq \cF$. By
	definition, $\cC^{**}$ contains $H_k^{**}$ for any $k \in \bN$, so $\ch(\cC^{**}) = \infty$ and if
	$\cC^{**} \subseteq \cF$ then $\ch(\cF) \ge \ch(\cC^{**}) = \infty$.
	
	Now suppose $\cC^{**} \not\subseteq \cF$ for every $** \in \{\bullet\bullet,\bullet\circ,\circ\circ\}$.
	Then for every $** \in \{\bullet\bullet,\bullet\circ,\circ\circ\}$ there  is some $m^{**}$ such that all graphs $G \in \cF$ are $H^{**}_{m^{**}}$-free.
	Hence, for $m = \max(m^{\bullet\bullet}, m^{\bullet\circ}, m^{\circ\circ})$, 
	all graphs $G \in \cF$ are $\{H^{\bullet\bullet}_m, H^{\bullet\circ}_m, H^{\circ\circ}_m\}$-free.
	
	It was proved in~\cite{CS18} that, due to Ramsey's theorem, for every $m \in \bN$ there exists
	a sufficiently large $k = k(m)$ such that any $\{H^{\bullet\bullet}_m, H^{\bullet\circ}_m,
	H^{\circ\circ}_m\}$-free graph $G$ has $\ch(G) < k$. Hence $\ch(\cF) < k$.
\end{proof}
}

Unlike standard universal graphs, PUGs exhibit a large quantitative gap between the \chainyGraphs
and the other minimal factorial classes, suggesting that stable factorial classes behave much differently than other factorial classes and may be worth studying separately, which has not yet been done in the context of understanding the factorial layer of graph classes.

\paragraph*{The Bell numbers threshold.}
There is another speed threshold within the factorial layer: the Bell numbers threshold.
The Bell number $B_n$ is the number of different set partitions of $[n]$, or equivalently the number
of $n$-vertex equivalence graphs; asymptotically it is $B_n \sim (n / \log n)^n$. Similarly to the
factorial layer itself, there is a set of \emph{minimal} classes above the Bell numbers.  However,
unlike the factorial layer, the set of minimal classes above the Bell numbers is \emph{infinite},
and it has been characterized explicitly~\cite{BBW05,ACFL16}. Once again, the classes $\cC^\dcirc,
\cC^\bcirc, \cC^\dbullet$ are minimal. This means that \emph{all} hereditary classes below the Bell
numbers are stable. Structural properties of these classes were given in \cite{BBW00}, which can be
used to prove the following.

\begin{restatable}{proposition}{bellnumbers}
\label{thm:bell numbers}\RestateRemark
Let $\cF$ be a hereditary graph class. Then:
\begin{enumerate}
\item If $\cF$ is a minimal class above the Bell numbers, then $\cF$ admits a constant-size
equality-based labeling scheme (and therefore a constant-size PUG), unless
$\cF \in \{\cC^\dcirc, \cC^\bcirc, \cC^\dbullet\}$.
\item If $\cF$ has speed below the Bell numbers, then $\cF$ admits a constant-size equality-based
labeling scheme (and therefore a constant-size PUG).
\end{enumerate}
\end{restatable}

The proof of this statement is routine, given the structural characterization of these graph classes
in \cite{BBW00,BBW05,ACFL16}, and we omit it.

\end{document}